\documentclass[a4paper,11pt]{article}
\usepackage{amssymb}
\usepackage{amsmath}
\usepackage{amsthm}

\usepackage{graphicx,color}
\usepackage{subfigure}

\newcommand{\ket}[1]{\left | #1 \right \rangle}
\newcommand{\bra}[1]{\left \langle #1 \right |}

\def\openone{\leavevmode\hbox{\small1\kern-3.8pt\normalsize1}}
\def\tr{{\rm Tr}}

\def\ca{{\cal A}}
\def\cv{{\cal V}}
\def\cw{{\cal W}}

\def\cs{{\cal S}}

\def\cx{{\cal X}}

\newcommand{\be}{\begin{equation}}
\newcommand{\ee}{\end{equation}}
\newcommand{\bea}{\begin{eqnarray}}
\newcommand{\eea}{\end{eqnarray}}
\newcommand{\beann}{\begin{eqnarray*}}
\newcommand{\eeann}{\end{eqnarray*}}

\def\AA{\mathbb{A}}
\def\BB{\mathbb{B}}
\def\RR{\mathbb{R}}
\def\ZZ{\mathbb{Z}}

\def\NN{\mathbb{N}}
\def\CC{\mathbb{C}}
\def\HH{\mathbb{H}}
\def\KK{\mathbb{K}}
\def\DD{\mathbb{D}}
\def\LL{\mathbb{L}}
\def\TT{\mathbb{T}}

\def\RR{\mathbb{R}}
\def\ZZ{\mathbb{Z}}

\def\NN{\mathbb{N}}
\def\CC{\mathbb{C}}
\def\HH{\mathbb{H}}
\def\KK{\mathbb{K}}
\def\DD{\mathbb{D}}
\def\LL{\mathbb{L}}
\def\UU{\mathbb{U}}

\newtheorem{theorem}{Theorem}
\newtheorem{lemma}{Lemma}
\newtheorem{proposition}{Proposition}
\newtheorem{corollary}{Corollary}

\newtheorem{remark}{Remark}
\newtheorem{definition}{Definition}
\newtheorem{example}{Example}

\newcommand{\Tr}{{\rm Tr}}

\title{Quantum spin chain dissipative mean-field dynamics}

\date{\null} 

\author{F. Benatti$^{1,2}$, F. Carollo$^{3}$, R. Floreanini$^2$, H. Narnhofer$^4$\\
\small ${}^1$Dipartimento di Fisica, Universit\`a di Trieste, Trieste, 34151 Italy\\
\small ${}^2$Istituto Nazionale di Fisica Nucleare, Sezione di Trieste, 34151 Trieste, Italy\\
\small ${}^3$School of Physics \& Astronomy, Faculty of Science, University of Nottingham, UK\\
\small ${}^4$Faculty of Physics, University of Wien, A-1091, Vienna, Austria
}

\begin{document}
\maketitle

\begin{abstract}
We study the emergent dynamics resulting from the infinite volume limit of the mean-field dissipative dynamics of quantum spin chains with clustering, but not time-invariant states. We focus upon three algebras of spin operators: the commutative algebra of mean-field operators, the quasi-local algebra of microscopic, local operators and the
collective algebra of fluctuation operators. In the infinite volume limit, mean-field
operators behave as time-dependent, commuting scalar macroscopic averages while quasi-local operators, despite the dissipative underlying dynamics, evolve unitarily in a typical non-Markovian fashion. Instead, the algebra of collective
fluctuations, which is of bosonic type with time-dependent canonical commutation relations, undergoes a time-evolution that retains the dissipative character of the underlying microscopic dynamics and exhibits non-linear features. These latter disappear by extending the time-evolution to a larger algebra where it is represented by a continuous one-parameter semigroup  of completely positive maps. 
The corresponding generator is not of Lindblad form and displays mixed
quantum-classical features, thus indicating that peculiar hybrid systems may naturally emerge at the level of quantum fluctuations in many-body quantum systems endowed with non time-invariant states.
\end{abstract}

\section{Introduction}
\label{intro}

In many physical situations concerning many-body quantum systems with $N$ microscopic components, the relevant observables are not those referring to single constituents, rather the collective ones consisting of suitably scaled sums of microscopic operators. Among them, one usually considers  macroscopic averages that scale as the inverse of $N$ and thus lose all quantum properties in the large $N$ limit thereby providing a description of the emerging commutative, henceforth classical, collective features of many body quantum systems.

Another class of relevant collective observables are the so-called quantum fluctuations: they account for the variations of microscopic quantities around their averages computed with respect to a chosen reference state. 
In analogy with classical fluctuations, they scale with the inverse square root of $N$ so that, unlike macroscopic observables, they can retain quantum features in the large $N$ limit \cite{BOOK,GVVNCCL,GVVDOF}. Indeed, whenever the reference microscopic state presents no long-range correlations, the fluctuations behave as bosonic operators; furthermore, from the microscopic state there emerges a Gaussian state over the corresponding bosonic Canonical Commutation Relation (CCR) algebra.
These collective observables describe a mesoscopic physical scale in between the purely quantum behaviour of microscopic observables and the purely classical one of commuting macroscopic observables \cite{HT}. 

The dynamical structure of quantum fluctuations has been intensively studied both in the unitary \cite{BOOK,GVVDOF,HT,NARN} and in the dissipative case \cite{VerDis,BCF,BCF2}; yet, in all these examples, only time-invariant reference states have been investigated, leading to macroscopic averages not evolving in time. Here, we relax this assumption and consider the possibility, often met in actual experiments, of a non-trivial dynamics of macroscopic averages. We shall do this by focusing on dissipative, Lindbald chain dynamics of mean-field type. The model studied in the following is very general and applies to a large variety of many-body systems consisting of $N$ microscopic finite-dimensional systems
weakly interacting with their environment. We will study the large $N$ limit of such a dissipative time-evolution $1)$ at the macroscopic level of mean-field observables, $2)$
at the microscopic scale of quasi-local observables, that is for arbitrarily large, but finite, number of chain sites, and $3)$ at the mesoscopic level of quantum fluctuations. 
These three scenarios look quite different and lead to features that, to the best of our knowledge, in particular for the cases $2)$ and $3)$, are novel in the field of many-body quantum systems.
\begin{enumerate}
\item
\textbf{Macroscopic observables}: these are described by the large $N$-limit of mean-field observables which yields commuting scalar quantities that evolve in time according to classical macroscopic equations of motion. 
\item
\textbf{Quasi-local observables}: the emerging dynamics is generated by a hamiltonian despite the microscopic dynamics being dissipative for each finite N. Moreover, and more interestingly, whenever macroscopic averages are not constant, such a unitary dynamics is non-Markovian, since it is implemented by a time non-local generator that always depends on the initial time. This latter is an interesting example of a unitary time-evolution manifesting memory effects.
\item
\textbf{Quantum fluctuations}: the emerging dynamics consists of a one-parameter Gaussian family of non-linear maps. In order to make them compatible with the physical requests of  linearity and complete positivity, these maps must be extended to a larger algebra, containing also classical degrees of freedom associated with the macrosocpic averages. 
The extended description gives raise to a dynamical hybrid system, containing both classical and quantum degrees of freedom, whose time-evolution corresponds to a semigroup of completely positive maps.
Unlike in hybrid systems so far studied~\cite{Ciccotti}-\cite{Buric}, the connection between classical and quantum degrees of freedom follows from the time-dependence of the mesoscopic commutation relations. Indeed, the commutator of two fluctuation operators is a time-evolving macroscopic average.
As a consequence, the generator of the dynamics on the larger algebra contains both classical, quantum and mixed classical-quantum contributions. In particular, the dynamical maps are completely positive, even if the purely quantum contribution to the generator need not in general be characterized by a positive semi-definite Kossakowski matrix. This is the first instance where this counter-intuitive fact is reported; notice, however, that in such a hybrid context, Lindblad's theorem does not apply.
\end{enumerate}

The structure of the manuscript is as follows: in Section~\ref{sec1} we introduce mean-field and fluctuation operators for quantum spin chains and define the mesoscopic limit.
In Section~\ref{sec2}, we introduce the dynamics generated by a Hamiltonian free term plus a mean field interaction and made dissipative by Lindblad type contributions of mean-field type.  In Section~\ref{subsec3.1}, we discuss the dynamics of macroscopic quantities and in Section~\ref{subsec3.2} the large $N$ limit of the time-evolution of quasi-local operators. 
In Section~\ref{ExpFun} we study the emerging mesoscopic dynamics of quantum fluctuations,
discussing first the symplectic structure in Section~\ref{reminv}, then the 
time-evolution and its non-linearity in Section~\ref{subsec4.2}.
In Section~\ref{subsec4.3} we focus upon the extension of the non-linear maps to a semi-group of completely positive Gaussian maps on a larger algebra and on the hybrid character of its generator.
Finally, Section~\ref{sec5} contains the proofs of all results presented in the previous sections.

\section{Quantum spin chains: macroscopic ad mesoscopic descriptions}
\label{sec1}

In this section, we discuss the macroscopic, respectively mesoscopic description of the collective behaviour of quantum spin chains given by classical mean-field observables, that scale with the inverse of the number of sites, $N$, respectively by quantum fluctuations that scale as the inverse square-root of $N$.

A quantum spin chain is a one-dimensional bi-infinite lattice, whose sites are indexed by  integers $j\in\ZZ$,
all supporting the same $d$-dimensional matrix algebra $\ca^{(j)}=M_d(\CC)$.
Its algebraic description~\cite{Bratteli,LiBingRen} is by means of the \textit{quasi-local} $C^*$ algebra $\ca$ obtained as the inductive limit of strictly local 
subalgebras $\ca_{[q,p]}=\bigotimes_{j=p}^q\ca^{(j)}$ supported by finite intervals 
$[q,p]$, with $q\leq p$ in $\ZZ$. Namely, one considers the 
algebraic union $\bigcup_{q\leq p}\ca_{[q,p]}$ and its completion with respect to the norm inherited by the local algebras.
Any operator $x\in M_d(\CC)$ at site $j$ can be embedded into $\ca$ as:
\be
x^{(j)}=\mathbf{1}_{j-1]}\otimes x\otimes\mathbf{1}_{[j+1}\ ,
\label{embed}
\ee
where $\mathbf{1}_{j-1]}$ is the tensor product of identity matrices at each site from $-\infty$ to $j-1$, while $\mathbf{1}_{[j+1}$ is the tensor product of identity matrices from site $j+1$ to $+\infty$.
Quantum spin chains are naturally endowed with the translation automorphism 
$\tau:\ca\mapsto\ca$ such that $\tau(x^{(j)})=x^{(j+1)}$.

Generic states $\omega$ on the quantum spin chain are described by positive, normalised linear expectation functionals $\ca\ni a\mapsto\omega(a)$ that assign mean values to  all operators in $\ca$.
In the following, we shall consider translation-invariant states such that
\be
\label{transinv}
\omega(a)=\omega\big(\tau(a)\big)\qquad\forall a\in \mathcal{A}\ .
\ee
At each site $j\in\ZZ$, these states are thus locally represented by a same density matrix $\rho\in M_d(\CC)$:
$\omega(x^{(j)})=\omega(x)=\tr(\rho\,x)$, $x\in M_d(\CC)$.
Furthermore, we shall focus upon
translation-invariant states $\omega$ that are also spatially $L_1$-\textit{clustering}~\cite{BOOK}. These are states that, for all single site operators $x,y$, satisfy
\be
\sum_{\ell\in\mathbb{Z}}\left|\omega\big(x^{(0)}y^{(\ell)}\big)\,-\,\omega(x)\omega(y)\right|<\infty\, ,
\label{stclus}
\ee
and then the weaker \textit{clustering} condition
\be
\label{clustates}
\lim_{n\to\pm\infty}\omega\Big(a^\dag\tau^n(b)c\Big)=\omega(a^\dag\,c)\,\omega(b)\quad \forall a,b,c\in\ca\ .
\ee 

\begin{remark}
{\rm 
The cluster condition \eqref{clustates} is often met in ground states or in thermal states associated to short-range Hamiltonians far from critical behaviours, such as phase transitions: it corresponds to the physical expectation that, in absence of long-range correlations, the farther spatially apart are observables, the closer they become to being statistically independent. On the other hand, the stronger clustering condition \eqref{stclus} is sufficient to ensure that fluctuations of physical observables display a Gaussian character which is again a property physically expected in systems far from phase transitions: such a condition is not strictly necessary for a system to have Gaussian fluctuations; however, it is often assumed for mathematical convenience \cite{BOOK}.}
\qed
\end{remark}

\subsection{Macroscopic scale: mean-field observables}

In an infinite quantum spin chain, the operators belonging to strictly local subalgebras contribute to the microscopic description of the system. In order to pass to a description based on collective observables supported by infinitely many lattice sites, a proper scaling must be chosen. Most often, mean-field observables are considered; these are constructed as averages of $N$ copies of a same single site observables $x$, from site $j=0$ to site $N-1$:
\be
X^{(N)}=\frac{1}{N}\sum_{k=0}^{N-1}x^{(k)}\ ,\qquad x\in M_d(\CC)\ .
\label{macro}
\ee
In the following, operators scaling as $X^{(N)}$ will be referred to as \textit{mean-field} operators; capital letters, like $X^{(N)}$, will refer to averages over specific number of lattice sites, while small case letters, like $x^{(k)}$, to operators at specific lattice sites.
  
Given any state $\omega$ on $\ca$, the Gelfand-Naimark-Segal (GNS) construction~\cite{BOOK4} provides a representation $\pi_\omega:\ca\mapsto \pi_\omega(\ca)$ of $\ca$ on a Hilbert space $\HH_\omega$ with a cyclic vector $\vert\omega\rangle$ such that the linear span of vectors of the form $\vert\Psi_a\rangle=\pi_\omega(a)\vert\omega\rangle$ is dense in $\HH_\omega$ and
$$
\omega(b^\dag\,a\,c)=\langle \Psi_b\vert\pi_\omega(a)\vert\Psi_c\rangle\ ,\qquad a,b,c\in\ca\ .
$$
As shown in Appendix A, given $x,y\in M_d(\CC)$, clustering yields that macroscopic averages $X^{(N)}$ and products of macroscopic averages $X^{(N)}Y^{(N)}$ tend weakly to scalar quantities:
\be
w-\lim_{N\to\infty} X^{(N)} = \omega(x)\, {\bf 1}\ ,\quad 
w-\lim_{N\to\infty}X^{(N)}\,Y^{(N)}\,=\,\omega(x)\,\omega(y)\, {\bf 1}\ ,
\label{macro1bis}
\ee 
in the sense that
$\displaystyle
\lim_{N\to\infty}\omega(a^\dag\, X^{(N)}Y^{(N)}\,b)=\omega(a^\dag b)\,\omega(x)\,\omega(y)$ for all $a,b\in\ca$.\\
Moreover, in the same Appendix it is shown that the $L_1$-clustering condition~\eqref{stclus} provides the following scaling:
\be
\label{macro1}
\left|\omega\big( X^{(N)}\,Y^{(N)}\big)-\,\omega(x)\,\omega(y)\right|=O\left(\frac{1}{N}\right)\, .
\ee
It thus follows that the weak-limits of mean-field observables are scalar quantities giving rise to a commutative (von Neumann) algebra. 
\medskip

Mean-field observables thus describe macroscopic, classical degrees of freedom emerging from the large $N$ limit of the microscopic quantum spin chain with no fingerprints left of its quantumness.
As outlined in the Introduction, we are instead interested in studying collective observables extending over the whole spin chain that may still keep some degree of quantum behaviour; for that a less rapid scaling than $1/N$ is necessary.

\subsection{Mesoscopic scale: quantum fluctuations}
\label{qfsec}

In order to disclose quantum behaviours of collective observables, one needs to look at fluctuations around average values. Indeed, fluctuations are commonly associated to an intermediate level of description in between the microscopic and the macroscopic ones, where one can hope to unveil truly mesoscopic phenomena exhibiting mixed classical-quantum features. In this section we shall review some of the known results about quantum fluctuation operators \cite{BOOK,GVVNCCL,GVVDOF}, introducing also the notation and relevant concepts useful to derive the results presented in the following sections.

Collective, microscopic operators of the form
\be
F^{(N)}(x)=\frac{1}{\sqrt{N}}\sum_{k=0}^{N-1}\left(x^{(k)}-\omega(x)\right)
\label{FL}
\ee
are quantum analogues of fluctuations in classical probability theory: we shall refer to them as
``local quantum fluctuations''.
Their large $N$ limit with respect to clustering states $\omega$ has been thoroughly investigated 
in~\cite{GVVNCCL,BOOK} yielding a non-commutative central limit theorem and  an associated  quantum  fluctuation algebra which turns out to be a Weyl algebra of bosonic degrees of freedom. 

The scaling $1/\sqrt{N}$ does not guarantee convergence in the weak-operator topology.
Nevertheless, consider $x,y\in M_d(\CC)$ such that $\left[x\,,\,y\right]=z$. Since $[x^{(j)}\,,\, y^{(\ell)}]=\delta_{j\ell}\,z^{(j)}$, with respect to a clustering state $\omega$, one has, following the same strategy used in Appendix A,
$$
\lim_{N\to\infty} \omega\left(a^\dag\left[F^{(N)}(x)\,,\,F^{(N)}(y)\right]b\right)=\lim_{N\to\infty} 
\frac{1}{N}\sum_{j=0}^{N-1}\omega\left(a^\dag z^{(j)}\, b\right)=\omega(a^\dag b)\,\omega(z) ,
$$
This means that commutators of local quantum fluctuations behave as mean-field quantities thus being, in the weak-topology, scalar multiples of the identity $\omega(z)\,\mathbf{1}$.
This latter fact clearly indicates that, in the large $N$ limit, a non-commutative structure emerges analogous to the algebra of quantum position and momentum operators.
To proceed to a formal proof of the convergence of the set of these operators to a bosonic algebra, it is convenient to work with unitary exponentials of the form  ${\rm e}^{iF^{(N)}(x)}$; in the large $N$ limit, these are expected to satisfy Weyl-like commutation relations~\cite{BOOK}.
\medskip

\begin{remark}
\label{remmes}
{\rm Because of the scaling $1/\sqrt{N}$, quantum fluctuations  provide a description level in 
between the microscopic (strictly local)  and the macroscopic (mean-field) ones. We will refer to it as to a \textit{mesoscopic} level whereby collective operators keep track of the microscopic non-commutative level they emerge from.}
\qed
\end{remark}
\medskip

In order to construct a quantum fluctuation algebra, one starts by selecting a set of $p$ linearly independent single-site microscopic observables 
$\chi=\{x_j\}_{j=1}^p$, $x_j\in M_d(\CC)$, $x_j=x_j^\dag$, and considers their local elementary fluctuations 
\begin{equation}
\label{FNj}
F^{(N)}_j:=F^{(N)}(x_j)=\frac{1}{\sqrt{N}}\sum_{k=0}^{N-1}\left(x^{(k)}-\omega(x)\right)
\ .
\end{equation}
Because of the assumption \eqref{stclus} on the state $\omega$, one has that the limits 
\be
\label{cormat}
\mathcal{C}^{(\omega)}_{ij}:=\lim_{N\to\infty}\omega\big(F^{(N)}_i\,F_j^{(N)}\big)
\ee
are well-defined and represent the entries of a positive $p\times p$ correlation matrix $\mathcal{C}^{(\omega)}$; moreover, one chooses the elements of $\chi$ in such a way that the characteristic functions $\omega\big(e^{itF_j^{(N)}}\big)$ converge to Gaussian functions of $t$ with zero mean and covariance $\Sigma_{jj}^{(\omega)}$, given by 
\be
\Sigma^{(\omega)}_{jj}=\frac{1}{2}\,\lim_{N\to\infty}\omega\left(\left\{F_j^{(N)}\ ,\ F_j^{(N)}\right\}\right)\ .
\ee 
This can be conveniently summarized by introducing the concept of {\sl normal quantum fluctuations} systems.
\medskip
 
\begin{definition}
\label{2}
A finite set of self-adjoint operators $\chi=\{x_j\}_{j=1}^p$ is said to have ``normal multivariate quantum fluctuations'' 
with respect to a clustering state $\omega$ if the latter obeys the $L_1$-clustering condition: 
$$
\sum_{\ell\in\mathbb{Z}}\Big|\omega(x^{(0)}_ix_j^{(\ell)})-\omega(x_i)\omega(x_j)\Big|<+\infty\ ,\qquad\forall x_i,x_j\in\chi\ ,
$$
and further satisfies
\be
\lim_{N\to\infty}\omega\big(\left(F_j^{(N)}\right)^2\big)= \Sigma^{(\omega)}_{jj}\ ,\quad
\lim_{N\to\infty}\omega(e^{itF_j^{(N)}})={\rm e}^{-\frac{t^2}{2}\Sigma^{(\omega)}_{jj}}\qquad 
\forall x_j\in\chi,\ \forall\, t\in\mathbb{R}\ .	
\label{Gauss2}
\ee
\end{definition}
\medskip

Given a set $\chi$ as in the above {\sl Definition \ref{2}}, by considering all possible real linear combinations of the set elements, one introduces the real-linear span 
\be
\label{spaceChi}
\mathcal{X}=\Big\{x_{\vec{r}}=\sum_{i=1}^p r_i\, x_i,\ x_i\in\chi,\ \vec{r}=\{r_i\}_{i=1}^p\in\mathbb{R}^p\Big\}\ .
\ee
The latter set can be endowed with two real bilinear maps: the first is positive and symmetric, 
\be
\label{BiFo}
(x_{\vec{r}_1},x_{\vec{r}_2})\to \vec{r}_1\cdot\left(\Sigma^{(\omega)}\,\vec{r}_2\right)=\sum_{i,j=1}^p\,r_{1i}\, r_{2j}\,\Sigma^{(\omega)}_{ij}\ ,\qquad\forall\ \vec{r}_{1,2}\in\RR^p\ ,
\ee
with 
\be
\label{covmat}
\Sigma^{(\omega)}_{ij}=\frac{1}{2}\,\lim_{N\to\infty}\omega\left(\left\{F_i^{(N)}\,,\,F_j^{(N)}\right\}\right)\ .
\ee 
The second one is, instead, anti-symmetric 
\be
\label{sympform1}
(\vec{r}_1,\vec{r}_2)\to\vec{r}_1\cdot\left(\sigma^{(\omega)}\vec{ r}_2\right)=\sum_{i,j=1}^p\,r_{1i}\,r_{2j}\,\sigma^{(\omega)}_{ij}\ ,\qquad\forall \vec{r}_{1,2}\in\RR^p\ ,
\ee 
and defined by the real symplectic matrix
$\sigma^{(\omega)}$
 with entries
\be
\label{sympform}
\sigma^{(\omega)}_{ij}:=-i\lim_{N\to\infty}\omega\left(\left[F_i^{(N)}\,,\,F_j^{(N)}\right]\right)=-\sigma^{(\omega)}_{ji}\ .
\ee
Notice that the $p\times p$ matrices introduced so far are related by the following equality
\be
\label{corcovsym}
\mathcal{C}^{(\omega)}=\Sigma^{(\omega)}\,+\,\frac{i}{2}\sigma^{(\omega)}\ .
\ee

For sake of compactness, because of the linearity of the map 
that associates an operator $x$ with its local quantum fluctuation $F^{(N)}(x)$,
the following notations will be used:
\bea
\label{qfa1}
\vec{r}\cdot\vec{F}^{(N)}&:=&\sum_{j=1}^p\,r_j\,F_j^{(N)}=F^{(N)}(x_{\vec{r}})\qquad \forall\,x_{\vec{r}}\in \mathcal{X}\ ,\\
\label{qfa2}
W^{(N)}(\vec{r})&:=&{\rm e}^{i\vec{r}\cdot\vec{F}^{(N)}}={\rm e}^{iF^{(N)}(x_{\vec{r}})}\ ,
\eea
where $\vec{F}^{(N)}$ is the operator-valued vector with components $F_j^{(N)}$, $j=1,2,\ldots,p$.

Given the symplectic matrix $\sigma^{(\omega)}=[\sigma_{ij}^{(\omega)}]$, one constructs the abstract \emph{Weyl} algebra 
$\mathcal{W}(\chi,\sigma^{(\omega)})$, linearly generated by
the Weyl operators $W(\vec{r})$, $\vec{r}\in\mathbb{R}^p$, obeying the relations:
\be
W^\dag(\vec{r})=W(-\vec{r})\ ,\quad
W(\vec{r}_1)\,W(\vec{r}_2)=W(\vec{r}_1+\vec{r}_2)\,{\rm e}^{-\frac{i}{2}\vec{r}_1\cdot\left(\sigma^{(\omega)}\vec{r}_2\right)}\ .
\label{Weyl}
\ee
The following theorem specifies in which sense, in the large $N$ limit, the local exponentials $W^{(N)}(\vec{r})$ yield Weyl operators $W(\vec{r})$ \cite{BOOK}.
\medskip

\begin{theorem}
\label{th1}
Any set $\chi$ with normal fluctuations with respect to a clustering state $\omega$ admits a regular, Gaussian state 
$\Omega$ on $\mathcal{W}(\chi,\sigma^{(\omega)})$ such that,
 for all $\vec{r}_j\in\mathbb{R}^p$, $j=1,2,\ldots,n$,
\be
\label{meslimth}
\lim_{N\to\infty}\omega\left(W^{(N)}(\vec{r}_1)\,W^{(N)}(\vec{r}_2)\,\cdots W^{(N)}(\vec{r}_n)\right)=
\Omega\left(W(\vec{r}_1)\,W(\vec{r}_2)\,\cdots W(\vec{r}_n)\right)\ ,
\ee
where the $W(\vec{r}_j)$ satisfy \eqref{Weyl} and
\be
\Omega\big(W(\vec{r})\big)=
\exp\Bigg(-\frac{\vec{r}\cdot\big(\Sigma^{(\omega)}\,\vec{r}\big)}{2}\Bigg)\ ,\qquad \forall\,\vec{r}\in \RR^p\ .
\label{quasistate}
\ee
\end{theorem}
\medskip

The regular and Gaussian character of $\Omega$ follows from \eqref{Gauss2}. In particular,
its regularity guarantees that one can write 
\be
\label{reg}
W(\vec{r})={\rm e}^{i\,\vec{r}\cdot\vec{F}}\ ,\qquad \vec{r}\cdot\vec{F}=\sum_{i=1}^p\,r_i\,F_i=F(x_{\vec{r}})\ ,
\ee 
where $\vec{F}$ is the operator-valued $p$-dimensional vector with components $F_i$ that are collective field operators satisfying canonical commutation relations 
$\left[F_i\,,\,F_j\right]=i\,\sigma_{ij}^{(\omega)}$, or, more generically,
\be
\left[F(x_{\vec{r}_1})\,,\,F(x_{\vec{r}_2})\right]=\left[\vec{r}_1\cdot\vec{F}\,,\,\vec{r}_2\cdot\vec{F}\right]=i\,\vec{r}_1\cdot\left(\sigma^{(\omega)}\vec{r}_2\right)\ .
\label{COMSIGMA}
\ee
We shall refer to the Weyl algebra $\mathcal{W}(\chi,\sigma^{(\omega)})$ generated by the strong-closure (in the GNS representation based on $\Omega$) of the linear span of Weyl operators as the {\sl quantum fluctuation algebra}.

\subsection{Mesoscopic limit}
\label{subsecmeso}

Later on, we shall focus on the effective dynamics of quantum fluctuations, emerging from the large $N$ limit of a family of microscopic dynamical maps $\{\Phi^{(N)}\}_{N\in\mathbb{N}}$ defined on the strictly local subalgebras $\ca_{[0,N-1]}$.
To formally state our main results, we introduce what we shall refer to as  \textit{mesoscopic limit}.
\medskip

\begin{definition}[\textbf{Mesoscopic limit}]\quad
Given a discrete family of operators $\{X^{(N)}\}_{N\in\mathbb{N}}$, in the quasi-local algebra $\ca$, we shall say that they posses the mesoscopic limit 
$$
\mathcal{W}(\chi,\sigma^{(\omega)})\ni X:=m-\lim_{N\to\infty} X^{(N)}\ ,
$$
if and only if 
\be
\label{meslim1a}
\Omega_{\vec{r}_1\vec{r}_2}\left(X\right)=\lim_{N\to\infty}\omega_{\vec{r}_1\vec{r}_2}\left(X^{(N)}\right)\ ,
\qquad\forall\,\vec{r}_{1,2}\in\RR^p\ ,
\ee
where 
\be\label{meslim1aa}
\omega_{\vec{r}_1\vec{r}_2}(X^{(N)}):=\omega\left(W^{(N)}(\vec{r}_1)\,X^{(N)}\,W^{(N)}(\vec{r}_2)\right)\ ,\
\Omega_{\vec{r}_1\vec{r}_2}(X):=\Omega\left(W(\vec{r}_1)\,X\,W(\vec{r}_2)\right)\ .
\ee
Further, given a sequence of completely positive, unital maps 
$\Phi^{(N)}:\ca_{[0,N-1]}\mapsto\ca_{[0,N-1]}$, one defines the mesoscopic limit, $\Phi:=m-\lim_{N\to\infty}\Phi^{(N)}$ on the Weyl algebra $\mathcal{W}(\chi,\sigma^{(\omega)})$ by 
\be
\label{meslim}
\lim_{N\to\infty}\omega_{\vec{r}_1\vec{r}_2}\left(\Phi^{(N)}\left[W^{(N)}(\vec{r})\right]\right)=\Omega_{\vec{r}_1\vec{r}_2}\left(\Phi\left[W(\vec{r})\right]\right)\ ,\qquad\forall\,\vec{r}_{1,2}\,,\,\vec{r}\in\RR^p\ .
\ee
\label{meslimdef}
\end{definition}

\begin{remark}
\label{remGNS}
{\rm Notice that the right hand side of \eqref{meslim} is the matrix element of $\pi_\Omega\left(\Phi\left[W(\vec{r})\right]\right)$ with respect to two vector states 
$\pi_\Omega(W(\vec{r}_{1,2}))\vert\Omega\rangle$ in the $GNS$-representation of the Weyl algebra generated by the operators $W(r)$ based on $\Omega$.
Since these vectors are dense in the $GNS$-Hilbert space, the mesoscopic map $\Phi$ is defined by the matrix elements of its action on Weyl operators that arise from local quantum fluctuations.   \qed}
\end{remark}
\medskip

According to the above definition and to \eqref{meslimth}, one can then say that the Weyl operators $W(x_j)$ are the mesoscopic limits of the local exponentials $W^{(N)}(x_j)$ and, by taking derivatives with respect to the parameters $r_j$, that the operators $F_j$ are the mesocopic limits of the local quantum fluctuations $F_j^{(N)}$.

\section{Mean-field dissipative dynamics}
\label{sec2}

Typically, a mean-field unitary spin-dynamics emerges in the large $N$ limit from a quadratic interaction hamiltonian scaling as $1/N$ as for the case of the BCS model in the quasi-spin description~\cite{TW}.

In this framework, operators $x\in\ca_{[0,N-1]}$ pertaining to the lattice sites $k=0,1,\ldots,N-1$, evolve in time according to a group of automorphisms of $\ca_{[0,N-1]}$, $\UU^{(N)}_t\,:\,x\mapsto x_t:=\UU^{(N)}_t[x]$, generated by
$$
\partial_tx_t=i\,\Big[h^{(N)}\,+\,H^{(N)}\,,\,x_t\Big]\ ,
$$ 
with a linear and bi-linear terms, the last one scaling as $1/N$:
\be
\label{2Ham}
h^{(N)}=\sum_{\mu=1}^{d^2}\epsilon_\mu\,\sum_{k=0}^{N-1}v^{(k)}_\mu\ ,\qquad
H^{(N)}=\sum_{\mu,\nu=1}^{d^2}h_{\mu\nu}\frac{1}{N}\sum_{k=0}^{N-1} v_\mu^{(k)}\,\sum_{\ell=0}^{N-1}v_\nu^{(\ell)}\ .
\ee
In the expressions above, the single-site operators $v_\mu=v_\mu^\dag$, $\mu=1,2,\ldots,d^2$, are chosen to constitute an hermitian orthonormal basis for the single-site algebra $M_d(\CC)$:
\be
\label{ONB}
\Tr(v_\mu\,v_\nu)=\delta_{\mu\nu}\ ,\qquad 
a=\sum_{\mu=1}^{d^2}\Tr(a\,v_\mu)\,v_\mu\ ,\qquad\forall\, a\in M_{d^2}(\CC)\ ,
\ee
and the coefficients $\epsilon_\mu$, $h_{\mu\nu}$ are chosen such that
\be
\label{Hcoeff}
\epsilon_\mu=\epsilon_\mu^*\ ,\qquad h_{\mu\nu}=h_{\nu\mu}^*\ .
\ee
In the following, we will perturb the hamiltonian generator of the microscopic dynamics with a Lindblad type contribution~\cite{BOOK6} scaling as $1/N$. We shall then study the time-evolution that emerges at the level of collective quantities from a dissipative microscopic master equation  $\partial_t x_t=\LL^{(N)}[x_t]$, with generator
\begin{eqnarray}
\label{mflind1a0}
\LL^{(N)}[x_t]&=&i\,\Big[h^{(N)}\,+\,H^{(N)}\,,\,x_t\Big]\,+\,\DD^{(N)}[x_t]\ ,\\
\label{mflind1a00}
\DD^{(N)}[x_t]&=&\sum_{\mu,\nu=1}^{d^2}\frac{C_{\mu\nu}}{2}\left(\Big[V^{(N)}_\mu\,,\,x_t\Big]\,V^{(N)}_\nu\,+\,
V^{(N)}_\mu\,\Big[x_t\,,\,V^{(N)}_\nu\Big]\right)\ ,\\
\label{mflind1c}
V^{(N)}_\mu&=&\frac{1}{\sqrt{N}}\sum_{k=0}^{N-1}v_\mu^{(k)}\ .
\end{eqnarray}
Notice that the mean-field scaling of $\LL^{(N)}$ is that of the commutator with $H^{(N)}$ and is due to the scaling $\displaystyle\frac{1}{\sqrt{N}}$ of the operators $V^{(N)}_\mu$.

In the above expression, the coefficients $C_{\mu\nu}$ are chosen to form  a positive semi-definite matrix $C=[C_{\mu\nu}]$, known as Kossakowski matrix. Such a property of $C$ ensures that the solution to $\partial_tx_t=\LL^{(N)}[x_t]$ is a one-parameter semigroup of completely positive, unital maps $\gamma^{(N)}_t=\exp(t\LL^{(N)})$ mapping $\ca_{[0,N-1]}$~\cite{BOOK6,Lindblad}:
\be
\label{semig}
\gamma^{(N)}_t[\mathbf{1}]=\mathbf{1}\ ,\qquad  
\gamma_t^{(N)}\circ\gamma_s^{(N)}=\gamma_{t+s}^{(N)}\ ,\qquad\forall\ s,t\geq 0\ .
\ee
\medskip

\begin{remark}\hfill
\label{remgen}
{\rm
\begin{enumerate}
\item
While the purely hamiltonian mean-field dynamics studied in \cite{BOOK, NARN, HT}
preserve the norm, the maps $\gamma^{(N)}_t$ are contractions, namely $\displaystyle \left\|\gamma^{(N)}_t[X]\right\|\leq\|X\|$ for all $X\in\mathcal{A}$. Furthermore,~\cite{Lindblad}
\be
\LL^{(N)}[x\,y]-\LL^{(N)}[x]\,y\,-\,x\,\LL^{(N)}[y]\,=\,\sum_{\mu,\nu=1}^{d^2}C_{\mu\nu}[V^{(N)}_\mu\,,\,x][y\,,\,V^{(N)}_\nu]\ ,
\label{Lind0}
\ee
for the hamiltonian contributions cancel and only $\DD^{(N)}$ contributes.
\item
A generator as in \eqref{mflind1a0} can be obtained by considering $N$ $d$-level systems  interacting with their environment via a hamiltonian of the form
\begin{equation}
\label{systbathcoup}
H_T=h^{(N)}\otimes{\bf 1}+{\bf 1}\otimes H_E+\sum_\alpha V^{(N)}_\alpha\otimes B_\alpha\ ,
\end{equation}
where $H^{(N)},H_E$ represent the hamiltonians of system and environment considered alone,  while the coupling hamiltonian consists of the operators $V^{(N)}_\alpha$ in \eqref{mflind1c} (which thus scale with $1/\sqrt{N}$) and environment operators $B_\alpha=B_\alpha^\dag$.
Notice that the scaling $1/\sqrt{N}$ of the interaction hamiltonian is the same as in the Dicke model for light-matter interaction \cite{HL,sew,sew2} and is the only one that, in the large $N$ limit with respect to clustering states, can lead to a meaningful dynamics with generator as 
in~\eqref{mflind1a0}.
In the weak-coupling limit \cite{Breuer}, when memory effects can be neglected, one retrieves an effective evolution of the $N$-body system alone, implemented by Lindblad generators of the specific type \eqref{mflind1a0}. The contribution $\DD^{(N)}$ describes dissipative and noisy effects due to the system-environment collective coupling in equation \eqref{systbathcoup}, while the hamiltonian $H^{(N)}$ in~\eqref{2Ham}
is an environment induced Lamb shift.\qed
\end{enumerate}
}
\end{remark}

We decompose the coefficients of the mean-field hamiltonian in~\eqref{2Ham} as $h_{\mu\nu}=h^{(re)}_{\mu\nu}+i h^{(im)}_{\mu\nu}$, with the real and imaginary parts  satisfying the relations 
\be
\label{aidd8}
h^{(re)}_{\mu\nu}=h^{(re)}_{\nu\mu}\ ,\qquad h^{(im)}_{\mu\nu}=-h^{(im)}_{\nu\mu}\ .
\ee 
Then, using~\eqref{mflind1c}, the mean-field hamiltonian contribution can be written as
\be
\label{HMF}
i\Big[H^{(N)}\,,\,x_t\Big]=i\sum_{\mu,\nu=1}^{d^2}h^{(re)}_{\mu\nu}\left\{\Big[V^{(N)}_\mu\,,\,x_t\Big]\,,\,V^{(N)}_\nu\right\}\,-\,\sum_{\mu,\nu=1}^{d^2}h^{(im)}_{\mu\nu}\left[\Big[V^{(N)}_\mu\,,\,x_t\Big]\,,\,V^{(N)}_\nu\right]\ .
\ee
In the above expression, $\{x\,,\,y\}=x\,y\,+\,y\,x$ denotes anti-commutator. At the same time, by decomposing the Kossakowski matrix $C=[C_{\mu\nu}]$ in its self-adjoint symmetric and anti-symmetric components as
\be
\label{KosAB}
A:=\frac{C+C^{tr}}{2}=[A_{\mu\nu}]\ ,\qquad B:=\frac{C-C^{tr}}{2}=[B_{\mu\nu}]\ ,
\ee
where $C^{tr}$ denotes transposition, one recasts $\DD^{(N)}[x_t]$ in \eqref{mflind1a0} as $\DD^{(N)}=\AA^{(N)}\,+\,\BB^{(N)}$, with
\beann
&&
\AA^{(N)}[x_t]=\sum_{\mu,\nu=1}^{d^2}\frac{A_{\mu\nu}}{2}\left[\Big[V^{(N)}_\mu\,,\,x_t\Big]\,,\,V^{(N)}_\nu\right]\ ,\quad A_{\mu\nu}=A_{\nu\mu}=A^*_{\mu\nu}\\
&&
\BB^{(N)}[x_t]=\sum_{\mu,\nu=1}^{d^2}\frac{B_{\mu\nu}}{2}\left\{\Big[V^{(N)}_\mu\,,\,x_t\Big]\,,\,V^{(N)}_\nu\right\}\ ,\quad B_{\mu\nu}=-B_{\nu\mu}=-B^*_{\mu\nu}\ .
\eeann
Thus, using~\eqref{HMF} and the above expressions, the generator in~\eqref{mflind1a0} deomposes as a mean-field dissipator-like term plus a free hamiltonian term:
\bea
\label{mflind1a}
\LL^{(N)}[x_t]&=&\Big(\HH^{(N)}\,+\,\widetilde{\DD}^{(N)}\Big)[x_t]\ ,\qquad \widetilde{\DD}^{(N)}:=\widetilde{\AA}^{(N)}+\widetilde{\BB}^{(N)}\\
\label{submap0}
\HH^{(N)}[x_t]&:=&
i\sum_{\mu=1}^{d^2}h_\mu\sum_{k=0}^{N-1}\,\Big[v^{(k)}_\mu\,,\,x_t\Big]=i\sqrt{N}\sum_{\mu=1}^{d^2}\epsilon_\mu\,\Big[V^{(N)}_\mu\,,\,x_t\Big]\\
\label{submap1}
\widetilde{\AA}^{(N)}&:=&\frac{1}{2}\sum_{\mu,\nu=1}^{d^2}\widetilde{A}_{\mu\nu}\,\left[\Big[V^{(N)}_\mu\,,\,x_t\Big]\,,\,V^{(N)}_\nu\right]\ ,\quad
\widetilde{A}_{\mu\nu}:=A_{\mu\nu}\,-\,2h^{(im)}_{\mu\nu}=\widetilde{A}_{\mu\nu}^*
\\
\widetilde{\BB}^{(N)}&:=&\frac{1}{2}\sum_{\mu,\nu=1}^{d^2}\widetilde{B}_{\mu\nu}\,\left\{\Big[V^{(N)}_\mu\,,\,x_t\Big]\,,\,V^{(N)}_\nu\right\}\ ,\quad
\widetilde{B}_{\mu\nu}:=B_{\mu\nu}\,+\,2i\,h^{(re)}_{\mu\nu}=-\widetilde{B}^*_{\mu\nu}\ .
\label{submap2}
\eea
The various coefficients are conveniently regrouped into the following $d^2\times d^2$ matrices 
\bea
\label{hmat1}
&&
h^{(re)}:=[h^{(re)}_{\mu\nu}]=(h^{(re)})^{tr}\ ,\ h^{(im)}=[h^{(im)}_{\mu\nu}]=-(h^{(im)})^{tr}\ ,\\
\label{hmat2}
&&
h:=[h_{\mu\nu}]=h^{(re)}\,+\,i\,h^{(im)}=h^\dag\ ,\qquad
\widetilde{A}=A-2h^{(im)}\ ,\qquad \widetilde{B}=B+2\,i\,h^{(re)}\ ,
\eea
where $\widetilde{A}$ is real, but unlike $A$ in~\eqref{KosAB}, non symmetric, and 
$\widetilde{B}$ is purely imaginary, but, unlike $B$ in~\eqref{KosAB}, not anti-symmetric.

\subsection{Mean-field dissipative dynamics on the quasi-local algebra}
\label{subsec3.1}

In this section we shall deal with the large $N$ limit of the microscopic dissipative dynamics $\gamma^{(N)}_t$ on the quasi-local algebra $\ca$ generated by $\LL^{(N)}$ in~\eqref{mflind1a}-\eqref{submap2}; namely, we shall investigate the behaviour when $N\to\infty$ of $\gamma^{(N)}_t[x]$, where $x\in\ca$ is either strictly local, that is different from the identity matrix, over an arbitrary, but fixed number of sites, or can be approximated in norm by strictly local operators.

\begin{definition}
\label{defsupp}
An operator $O\in\ca$ is \emph{strictly local} if there exists an interval $[a,b]$ of lattice sites, such that $[O\,,\,x^{(k)}]=0$  
$\forall x\in M_d(\CC)$ and $k\notin[a,b]$. 
The smallest such interval is the support $\mathcal{S}(O)$ of $O\in\ca$
whose cardinality will be denoted by $\ell(O)$.
\end{definition}
\medskip

We shall consider microscopic states $\omega$ that are translation invariant and clustering, but not necessarily invariant under the large $N$ limit of the microscopic dynamics; namely such that, in general, on strictly local $x\in\mathcal{A}$,
$$
\lim_{N\to\infty}\omega\left(\gamma^{(N)}_t[x]\right)\neq\omega(x)\ .
$$
Thus, we shall consider the case of macroscopic averages associated with mean-field operators that may also change in time. 
The existence of the following macroscopic averages is first guaranteed for all $t\in[0,R]$ with $R$ defined by the norm-convergence radius of the exponential series 
$$
\gamma^{(N)}_t=\sum_{k=0}^\infty\frac{t^k}{k!}\left(\LL^{(N)}\right)^k
$$   
on local and mean-field operators by Corollary~\ref{exg} in Section~\ref{Applemmacors}, and then extended to all finite times $t\geq0$ by Proposition~\ref{extdyn}.
\medskip

\begin{definition}
\label{defmacro}
The time-dependent macroscopic averages of the commutator of single-site operators, $v_\mu$ and $[v_\mu\,,\,v_\nu]\in M_d(\CC)$, with respect to the microscopic state at any finite time $t\geq 0$ will be denoted by:
\bea
\label{macr1a}
\omega_\alpha(t)&:=&\lim_{N\to\infty}\omega^{(N)}_\alpha(t)\ ,\quad
\omega^{(N)}_\alpha(t):=\omega\left(\gamma^{(N)}_t\left[\frac{1}{N}\sum_{k=0}^{N-1}v_\alpha^{(k)}\right]\right)\ ,\\
\label{macr1b}
\omega_{\alpha\beta}(t)&:=&\lim_{N\to\infty}\omega^{(N)}_{\alpha\beta}(t)\ ,\quad
\omega^{(N)}_{\alpha\beta}(t):=\omega\left(\gamma^{(N)}_t\left[\frac{1}{N}\sum_{k=0}^{N-1}[v^{(k)}_\alpha\,,\,v^{(k)}_\beta]\right]\right)\ .
\eea
\end{definition}
\medskip

Using the relations \eqref{ONB}, one writes
$$
\left[v^{(k)}_\alpha\,,\,v_\beta^{(k)}\right]=\sum_{\gamma=1}^{d^2}\Tr\left(\left[v^{(k)}_\alpha\,,\,v_\beta^{(k)}\right]\,v^{(k)}_\gamma\right)\,v_\gamma^{(k)}\ ,
$$ 
and, since the trace does not depend on the site index, one may set 
\be
\label{aid0}
J^{\gamma}_{\alpha\beta}:=\Tr\left(\left[v^{(k)}_\alpha\,,\,v^{(k)}_\beta\right]\,v^{(k)}_\gamma\right)=J^{\beta}_{\gamma\alpha}=-J^{\gamma}_{\beta\alpha}=-(J^\gamma_{\alpha\beta})^*\ ,\quad \forall k\ , 
\ee
and derive
\be
\label{aid1}
\left[v^{(k)}_\alpha\,,\,v_\beta^{(k)}\right]=\sum_{\gamma=1}^{d^2}\,J^\gamma_{\alpha\beta}\,v^{(k)}_\gamma\ ,\quad
\omega_{\alpha\beta}(t)=\sum_{\gamma=1}^{d^2}\,J^\gamma_{\alpha\beta}\,\omega_\gamma(t)\ .
\ee
Propositions~\ref{mfdyn} and Corollary~\ref{quasilocal} in Section~\ref{Applemmacors} show that the macroscopic averages satisfy the following equations of motion for all times $t\geq 0$:
\bea
\nonumber
\frac{{\rm d}}{{\rm d}t}\omega_\alpha(t)&=&\sum_{\mu=1}^{d^2}\Big(\sum_{\nu=1}^{d^2}\widetilde{B}_{\mu\nu}\,\omega_\nu(t)\,+\,i\,\epsilon_\mu\Big)\,\omega_{\mu\alpha}(t)\\
\label{macrodyn}
&=&\sum_{\mu,\gamma=1}^{d^2}\Big(\sum_{\nu=1}^{d^2}\widetilde{B}_{\mu\nu}\,\omega_\nu(t)\,+\,i\,\epsilon_\mu\Big)\,J^\gamma_{\mu\alpha}\,\,\omega_\gamma(t)\ ,\quad\forall \alpha=1,2,\ldots d^2\ .
\eea
Denoting by $\vec{\omega}_t$ the vector with components $\omega_\alpha(t)$ and using \eqref{aid1}, it proves convenient for later use, in particular for the derivation of the dissipative fluctuation dynamics in Theorem \ref{expfluct}, to recast the equations of motion in the following compact, matrix-like form
\be
\label{d0}
\frac{{\rm d}}{{\rm d}t}\vec{\omega}_t=D(\vec{\omega}_t)\,\vec{\omega}_t\ ,\qquad
D(\vec{\omega}_t)=\Big[D_{\alpha\beta}(\vec{\omega}_t)\Big]=\widetilde{D}(\vec{\omega}_t)\,+\,\,i\,\mathcal{E}\ ,
\ee
where $\widetilde{D}(\vec{\omega}_t)$ and $\mathcal{E}$ have entries
\be
\label{D}
\widetilde{D}_{\alpha\beta}(\vec{\omega}_t)=\sum_{\mu,\nu=1}^{d^2}\,J^\mu_{\alpha\beta}\,\widetilde{B}_{\mu\nu}\,\omega_\nu(t)\ ,\qquad
\mathcal{E}_{\alpha\beta}=\sum_{\mu=1}^{d^2}\,\epsilon_\mu\,J^\mu_{\alpha\beta}\ ,
\ee
and $\widetilde{D}(\vec{\omega}_t)$ depends implicitly on time through the time-evolution: $\vec{\omega}\mapsto\vec{\omega}_t$.
\medskip

Notice that all the scalar quantities multiplying $\omega_\gamma(t)$ change sign under conjugation, whence 
the matrix $D(\vec{\omega}_t)$ is real and 
\be
\label{remD}
J^\mu_{\alpha\beta}=-J^\mu_{\beta\alpha}\quad\hbox{implies}\quad D^\dag(\vec{\omega}_t)=-D(\vec{\omega}_t)\ .
\ee

The non-linear equations \eqref{d0} with initial condition $\vec{\omega}_0=\vec{\omega}$ are formally solved by the matricial expression
\be
\label{formalsol}
\vec{\omega}_t=M_t(\vec{\omega})\,\vec{\omega}\ ,\quad M_t(\vec{\omega}):=\mathbb{T}{\rm e}^{\int_0^t{\rm d}s\,D(\vec{\omega}_s)}
\ee
where $\mathbb{T}$ denotes time-ordering and the dependence of the $d^2\times d^2$ matrix $M_t(\vec{\omega})$ on the time-evolution $\vec{\omega}\mapsto\vec{\omega}_t$ embodies the non-linearity of the dynamics.
However, this is just a formal writing, that will prove to be useful later on: the time-evolution of the macroscopic averages can be found only by directly solving the system of equations \eqref{d0}.

Despite the time-ordering, since there is no explicit time-dependence in the equations \eqref{d0}, the time-evolution of the macroscopic averages composes as a semigroup, 
\be
\label{avtime}
\vec{\omega}\mapsto\vec{\omega}_s\mapsto(\vec{\omega}_s)_t=\vec{\omega}_{s+t}\qquad \forall\ s,t\geq 0\ .
\ee
Moreover, because of the anti-symmetry of $D(\vec{\omega}_t)$ and of the fact that the macroscopic averages are real, the quantity $K(t):=\sum_{\alpha=1}^{d^2}\omega^2_\alpha(t)$ is a constant of the motion
$$
\frac{{\rm d}K(t)}{{\rm d}t}=\sum_{\alpha,\beta=1}^{d^2}\Big(D_{\alpha\beta}(\vec{\omega}_t)\,+\,D_{\beta\alpha}(\vec{\omega}_t)\Big)\omega_\alpha(t)\,\omega_\beta(t)=0\ .
$$
\medskip

\begin{remark}
\label{remavtime}
{\rm The components $\omega_\mu(t)=\lim_{N\to\infty}\omega^{(N)}_t(v_\mu)$ of the vector $\vec{\omega}_t$ lie within the parallelepiped $[-1,1]^{\times d^2}$: this follows since the orthonormal matrices $v_\mu$ are such that $\tr(v_\mu^2)=1$ and thus
$\|v_\mu\|\leq 1$, whence
\be
\label{constr1}
\sum_{\mu=1}^{d^2}(\omega(v_\mu))^2\leq d^2\ .
\ee 
Furthermore, the positivity of the state $\omega$ yields the positivity of the $d^2\times d^2$ matrix of coefficients $\omega(v_\mu\,v_\nu)$:
$$
\sum_{\mu,\nu=1}^{d^2}\lambda_\mu^*\lambda_\nu\,\omega(v_\mu\,v_\nu)=\omega(V^\dag V)\geq 0\ ,\qquad \forall\
V:=\sum_{\mu=1}^{d^2}\lambda_\mu\, v_\mu\in M_d(\CC)\ .
$$   
By expanding the matrix product $v_\mu\, v_\nu$ with respect to the $\{v_\mu\}_{\mu=1}^{d^2}$ orthonormal basis, $v_\mu\,v_\nu=\sum_{\alpha=1}^{d^2} V^\alpha_{\mu\nu}\,v_\alpha$, one derives the positivity constraint
\be
\label{constr2}
\sum_{\alpha=1}^{d^2}V^\alpha\,\omega(v_\alpha)\geq 0\ ,
\ee
where the $d^2\times d^2$ matrix $V^\alpha=\left[V^\alpha_{\mu\nu}\right]$ is fixed by the chosen basis.
Thus, the vectors $\vec{\omega}$ of macroscopic averages $\{\omega(v_\mu)\}_{\mu=1}^{d^2}$ belong to the subset $\cs\subset[-1,1]^{\times d^2}$ satisfying the constraints \eqref{constr1} and \eqref{constr2}. In conclusion,
the macroscopic dynamics generated by the non-linear, time-independent equations of motions \eqref{macrodyn} forms a semigroup and maps $\cs$ into itself.}
\qed
\end{remark}

\subsection{Macroscopic dynamics of local observables}
\label{subsec3.2}

With the time-evolution of macroscopic averages at disposal, we are now able to derive the large $N$ limit of the dynamics of quasi-local operators $x\in\ca$.
\medskip

\begin{theorem}
Let the quasi-local algebra $\ca$ be equipped with a translation-invariant, spatially $L_1$-clustering state $\omega$. In the large $N$ limit, the local dissipative generators $\LL^{(N)}$ in \eqref{mflind1a} define on $\ca$ a one-parameter family of automorphisms that depend on the state $\omega$ and are such that, for any finite $t\geq 0$, 
\be
\label{auto1}
\lim_{N\to\infty}\omega\left(a\,\gamma^{(N)}_t[O]\,b\right)=\omega\left(a\,\alpha_t[O]\,b\right)\ ,
\ee
for all $a,b\in\ca$ and $O\in\ca$. If $O$ has finite support $\cs(O)\subset[0,S-1]$, then
\be
\label{auto2}  
\alpha_t[O]=\big(U^{(S)}_t\big)^\dag\,O\,U^{(S)}_t\ ,\quad
U^{(S)}_t=\mathbb{T}{\rm e}^{-i\int_0^t{\rm d}s\, H^{(S)}_s}\ ,
\ee
with explicitly time-dependent hamiltonian
\be
\label{auto3}
H^{(S)}_t=-i\,\sum_{k=0}^{S-1}\,\sum_{\mu=1}^{d^2}\Big(\sum_{\nu=1}^{d^2}\widetilde{B}_{\mu\nu}\,\omega_\nu(t)\,+\,i\,\epsilon_\mu\Big)\,v_\mu^{(k)} \ ,
\ee
where 
\be
\label{auto2b}
\mathbb{T}{\rm e}^{-i\int_0^t{\rm d}s\, H^{(S)}_s}=
\mathbf{1}+\sum_{k=1}^\infty
(-i)^k\int_0^t{\rm d}s_1\cdots\int_0^{s_{k-1}}{\rm d}s_k\,H^{(S)}_{s_1}\cdots H^{(S)}_{s_k}\ .
\ee
\label{limdyn}
\end{theorem}
\bigskip

The proof of the above theorem is given in Section~\ref{Applemmacors}.
Using~\eqref{submap2}, the hamiltonian reads
\be
\label{auto3a}
H^{(S)}_t=\sum_{k=0}^{S-1}\,\sum_{\mu=1}^{d^2}\Big(\Big(\sum_{\nu=1}^{d^2}-i\,\,B_{\mu\nu}\,+\,2\,h^{(re)}_{\mu\nu}\Big)\,\omega_\nu(t)\,+\,\epsilon_\mu\Big)\,v_\mu^{(k)}\ .
\ee
where $B^*_{\mu\nu}=-B_{\mu\nu}$, $(h^{(re)}_{\mu\nu})^*=h^{(re)}_{\mu\nu}$ and $\epsilon^*_\mu=\epsilon_\mu$ guarantee that
$H^{(S)}_t$ is hermitean.
Notice that, in the large $N$ limit, the microscopic dissipative term $\DD^{(N)}$ only contributes with a correction to the free hamiltonian terms in~\eqref{2Ham} so that the dissipative time-evolution of local observables becomes automorphic.
\medskip

Consider the dynamics of single site observables by choosing in~\eqref{auto2} $O$ equal to one of the orthonormal matrices at site $\ell$, $v_\mu^{(\ell)}$. Then,
\bea
\nonumber
\frac{{\rm d}}{{\rm d}t}\alpha_t[v^{(\ell)}_\gamma]&=&i\,\big(U^{(\ell)}_t\big)^\dag\,\left[H^{(\ell)}_t\,,\,v^{(\ell)}_\gamma\right]\,U^{(\ell)}_t\\
\nonumber
&=&\sum_{\mu=1}^{d^2}\,\Big(\sum_{\nu=1}^{d^2}\Big(B_{\mu\nu}\,+\,2\,i\,h^{(re)}_{\mu\nu}\Big)\,\omega_\nu(t)\,+\,i\,\epsilon_\mu\Big)\,\big(U^{(\ell)}_t\big)^\dag\,\left[v^{(\ell)}_\mu\,,\,v^{(\ell)}_\gamma\right]\,U^{(\ell)}_t\\
\nonumber
&=&
\sum_{\mu,\beta=1}^{d^2}\,\Big(\sum_{\nu=1}^{d^2}\widetilde{B}_{\mu\nu}\,\omega_\nu(t)\,+\,i\,\epsilon_\mu\Big)\,J^\beta_{\mu\gamma}\,
\big(U^{(\ell)}_t\big)^\dag\,v^{(\ell)}_\beta\,\,U^{(\ell)}_t
=\sum_{\beta=1}^{d^2}D_{\gamma\beta}(\vec{\omega}_t)\,\alpha_t[v^{(\ell)}_\beta]\ ,
\label{macdyn2}
\eea
where use has been made of the relations \eqref{ONB} and of the matrix elements \eqref{D}.
Notice that the expectations $\omega\big(\alpha_t\big[v^{(\ell)}_\gamma\big]\big)$ satisfy the same equations \eqref{macrodyn} satisfied by the macroscopic observables $\omega_\gamma(t)$; since these quantities coincide at $t=0$, one has
\be
\label{macreq}
\omega\left(\alpha_t\left[v^{(\ell)}_{\gamma}\right]\right)=\omega_\gamma(t)\ ,\qquad \forall\ \gamma=1,2,\ldots,d^2\ ,\ \forall\ t\geq 0\ .
\ee

\begin{remark}{}\hfill
\label{NonMarkovRem}
{\rm
\begin{enumerate}
\item
The convergence of the mean-field dissipative dynamics $\gamma^{(N)}_t$ to the automorphism $\alpha_t$ of $\ca$ occurs in the weak-operator topology associated with the GNS-representation of $\ca$ based on the state $\omega$.
\item
The 
automorphisms $\alpha_t$ have been derived for positive times, only. This means that, though the inverted automorphisms $\alpha_{-t}$ surely exist, they cannot however arise from the underlying non-invertible microscopic dynamics.
\item
The one-parameter family  $\{\alpha_t\}_{t\geq 0}$ fails to obey the forward-in-time composition law as in \eqref{semig} which is typical of time-independent generators, nor the one corresponding to two-parameter semi-groups, 
$$
\alpha_{t,t_0}=\alpha_{t,s}\circ\alpha_{s,t_0}\ ,\qquad 0\leq t_0\leq s\leq t\ .
$$
which arises from time-ordered integration of generators that depend explicitly on the running time $t$, but not on the initial time $t_0$.
Indeed, if the microscopic dynamics starts at $t_0\geq0$, then the semigroup properties ensure that, at time $t\geq t_0$, any quasi-local initial condition      $O\in\mathcal{A}$ has evolved into $\gamma^{(N)}_{t-t_0}[O]$. Then, adapting Theorem~\ref{limdyn} to a generic initial time $t_0\geq 0$, similarly to \eqref{auto1}, the large $N$ limit yields a one-parameter family of automorphisms $\alpha_{t-t_0}$, $t\geq t_0$, such that
\be
\label{nonM1}
\lim_{N\to\infty}\omega\left(a\,\gamma^{(N)}_{t-t_0}[O]\,b\right)=\omega\left(a\,\alpha_{t-t_0}[O]\,b\right)\ ,
\ee
for all $a,b\in\ca$ and $O$ quasi-local. If the support of $O$ is, for sake of simplicity, $[0,S-1]$, then 
\be
\label{nonM2}  
\alpha_{t-t_0}(O)=\big(U^{(S)}_{t-t_0}\big)^\dag\,O\,U^{(S)}_{t-t_0}\ ,\quad
U^{(S)}_{t-t_0}=\mathbb{T}{\rm e}^{-i\int_0^{t-t_0}{\rm d}s\, H^{(S)}_s}\ .
\ee
Therefore, the time-derivative yields a generator:
\bea
\label{NMGa}
\frac{{\rm d}}{{\rm d}t}\alpha_{t-t_0}(O)&=&\mathcal{K}_{t-t_0}^\omega\left[\alpha_{t-t_0}[O]\right] ,\hspace{2cm}\\
\label{NMGb}
\mathcal{K}_{t-t_0}^\omega\left[O\right]&=&\sum_{\mu,\nu=1}^{d^2}B_{\mu\nu}\,\omega_\nu(t-t_0)\,\sum_{k=0}^{S-1}\Big[\alpha_{t-t_0}\left[v_\mu^{(k)}\right]\,,\,O\Big]\ ,
\eea
which depends on both the running and initial times.
\item
By setting $t_0=0$ in~\eqref{NMGa}, one sees that the one-parameter family 
$\{\alpha_t\}_{t\geq 0}$ is generated by a time-local master equation. However, since in general, that is for $t_0\geq 0$, the generator $\mathcal{K}_{t-t_0}$ depends on both the running time $t$ and the initial time $t_0$, the family of automorphisms is non-Markovian in the sense of~\cite{ChrKos}.
On the other hand, if one uses lack of CP-divisibility as a criterion 
of non-Markovianity~\cite{Rivas}, then $\{\alpha_t\}_{t\geq 0}$ is Markovian. Indeed, being the dynamics unitary, there always exists a completely positive intertwining map $\beta_{t,\tau}$, $t\geq\tau\geq 0$, such that $\alpha_t=\beta_{t,\tau}\circ\alpha_{\tau}$: for any $O\in\ca$ with $\cs(O)\subset[0,S-1]$, one can write
$\alpha_t[O]=\beta_{t,\tau}\circ\alpha_{\tau}[O]$, where
$$
\beta_{t,\tau}[O]=\big(U^{(S)}_t\big)^\dag\,U^{(S)}_\tau\,O\,\big(U^{(S)}_\tau\big)^\dag\,U^{(S)}_t\ .
$$
\item
If the $\lim_{N\to\infty}\omega\circ\gamma^{(N)}_t$ provides a time-invariant state on the quasi-local algebra $\mathcal{A}$, then one recovers the one-parameter semigroup features of~\eqref{semig} (see also~\cite{BCFN}).\qed
\end{enumerate}}
\end{remark}

\begin{example}
\label{example}
{\rm We shall consider a qubit spin chain consisting of a lattice whose sites $j\in\NN$ support the algebra $M_2(\mathbb{C})$. As a Hilbert-Schmidt orthogonal matrix basis $\{v_\mu\}_{\mu=1}^4$,
we choose the spin operators $s_1,s_2,s_3,\mathbf{1}$, normalized in a such way that
\be
\label{scb}
[s_\mu\,,\,s_\nu]=i\,\varepsilon_{\mu\nu\gamma}\,s_\gamma\ ,\qquad \mu\,,\,\nu\,,\,\gamma=1,2,3\ .
\ee
Then, with $\displaystyle S^{(N)}_\mu=\frac{1}{\sqrt{N}}\sum_{k=0}^{N-1}s_\mu^{(k)}$,
we study the following dissipative generator, with Kossakowski matrix 
$\displaystyle C=\begin{pmatrix}
1&-i&0\\
i&1&0\\
0&0&0
\end{pmatrix}$,
\bea
\nonumber
\LL^{(N)}[x]&=&\sum_{\alpha,\beta=1}^{2}\frac{C_{\alpha\beta}}{2}\Big(\left[S^{(N)}_\alpha\,,\,x\right]\,S^{(N)}_\beta\,+\,S^{(N)}_\alpha\,\left[x\,,\,S^{(N)}_\beta\right]\Big)=
\\
\label{scl2}
&=&S^{(N)}_+\,x\,S^{(N)}_-\,-\,\frac{1}{2}\Big\{S^{(N)}_+\,S^{(N)}_-\,,\,x\Big\}\ ,\qquad S^{(N)}_\pm=S^{(N)}_1\pm i\,S^{(N)}_2\ . 
\eea
Therefore, with respect to~\eqref{hmat1},~\eqref{hmat2}, $h=0$ and $\mathcal{E}=0$, so that $\widetilde{A}$ and $\widetilde{B}$ coincide with the symmetric and anti-symmetric components of $C$,
\be
\label{koss}
A=\begin{pmatrix}
1&0&0\\
0&1&0\\
0&0&0
\end{pmatrix}\ ,\quad B=\begin{pmatrix}
0&-i&0\\
i&0&0\\
0&0&0
\end{pmatrix}\ .
\ee
With respect to a translation-invariant clustering state $\omega$, the only non-trivial macroscopic averages $\omega_\mu(t)$ given by \eqref{macr1a} are 
$\omega_{1,2,3}(t)$ while $\omega_4(t)=1$ for all $t\geq 0$.
Since $\|s_\mu\|\leq 1/2$, we will then consider the vector $\vec{\omega}_t=(\omega_1(t),\omega_2(t),\omega_3(t))$ with components belonging to $[-1/2\,,\,1/2]$. 
Furthermore, from \eqref{aid1} and \eqref{scb} one computes
\be
\label{omega}
\omega_{\mu\nu}(t)\,=\,i\,\epsilon_{\mu\nu\gamma}\,\omega_\gamma(t)\ ,\qquad \mu,\nu,\gamma=1,2,3\ ,
\ee
whence \eqref{macrodyn} and $B$ in \eqref{koss} yield the following system of differential equations: 
\be
\label{macrodynex}
\frac{{\rm d}}{{\rm d}t}\omega_1(t)=\omega_1(t)\,\omega_3(t)\ ,\quad
\frac{{\rm d}}{{\rm d}t}\omega_2(t)=\omega_2(t)\,\omega_3(t)\ ,\quad
\frac{{\rm d}}{{\rm d}t}\omega_3(t)=-\omega^2_1(t)\,-\,\omega^2_2(t)
\ ,
\ee
corresponding to the following matrix $D(\vec{\omega}_t)$ in \eqref{d0}:
\be
\label{Dt}
D(\vec{\omega}_t)=
\begin{pmatrix}
0&0&\omega_1(t)\cr
0&0&\omega_2(t)\cr
-\omega_1(t)&-\omega_2(t)&0
\end{pmatrix}\ .
\ee
Then, the norm 
\be
\label{constant}
\|\vec{\omega}_t\|=\sqrt{\sum_{\mu=1}^3|\omega_\mu(t)|^2}=\xi
\ee
is a constant of the motion; thus the third equation can readily be solved, yielding
\be
\label{sol1}
\omega_3(t)=-\xi\tanh\left(\xi(t+b)\right)\ ,
\ee
where the constant $b$ is chosen to implement the initial condition  
$\omega_3:=\omega_3(0)=-\xi\tanh\left(\xi b\right)$.
Knowing $\omega_3(t)$, one also obtains
\be
\label{sol2}
\omega_1(t)=\,\frac{\cosh(b\,\xi)}{\cosh(\xi(t+b))}\,\omega_1\ ,\quad
\omega_2(t)=\,\frac{\cosh(b\,\xi)}{\cos(\xi(t+b))}\,\omega_2\ ,
\ee
where $\omega_{1,2}:=\omega_{1,2}(0)$.
According to~\eqref{auto3}, equipped with these quantities, the Hamiltonian for the first
$S$ chain-sites  reads
\begin{equation}
\label{auxq2}
H^{(S)}_t=\sum_{k=0}^{S-1}\left(\omega_1(t)s_2^{(k)}-\omega_2(t)s_1^{(k)}\right)\ .
\end{equation}
Since $[H^{(S)}_{t_1}\,,\,H^{(S)}_{t_2}]=0$, for all $t_{1,2}\in\RR$,  the unitary operators implementing the automorphic dynamics from $t=0$ to $t>0$ (see \eqref{auto2}) are given by:
\bea
\label{auxq3}
U^{(S)}_t&=&{\rm e}^{-i\int_0^t{\rm d}u\, H^{(S)}_u}=\prod_{k=0}^{S-1}\exp\Big(-i\alpha(t)\left(\omega_1\,s_2^{(k)}-\omega_2\,s_1^{(k)}\right)\Big)\ ,\\
\nonumber
\alpha(t)&:=&\cosh(b\,\xi)\,\int_0^t \, \frac{{\rm d}u}{\cosh(\xi(u+b))}\\
\label{auxq3c}
&=&
\cosh(b\,\xi)\,\Big(\arctan\Big({\rm e}^{-\xi(t+b)}\Big)-\arctan\Big({\rm e}^{-\xi b}\Big)\Big)\ .
\eea 
The mean-field dynamics is thus specified by time-evolution of single-site spin operators: 
\bea
\label{auxq3a}
\hskip-1.7cm
&&
(U^{(S)}_t)^\dagger\,\begin{pmatrix}s_1\cr s_2\cr s_3\end{pmatrix}\,U^{(S)}_t=M_t(\vec{\omega})\,\begin{pmatrix}s_1\cr s_2\cr s_3\end{pmatrix}\ ,\\
\label{auxq3b}
\hskip-1.7cm
&&
M_t(\vec{\omega})=\frac{1}{\xi_{12}}\begin{pmatrix}
\omega_1^2\cos\left(\alpha(t)\xi_{12}\right)+\omega_2^2&\omega_1\omega_2\left(\cos\left(\alpha(t)\xi_{12}\right)-1\right)&\xi_{12}\,\omega_1\sin\left(\alpha(t)\xi_{12}\right)\\
\omega_1\omega_2\left(\cos\left(\alpha(t)\xi_{12}\right)-1\right)&\omega_2^2\cos\left(\alpha(t)\xi_{12}\right)+\omega_1^2&\xi_{12}\,\omega_2\sin\left(\alpha(t)\xi_{12}\right)\\
-\xi_{12}\,\omega_1\sin\left(\alpha(t)\xi_{12}\right)&-\xi_{12}\,\omega_2\sin\left(\alpha(t)\xi_{12}\right)&\xi_{12}^2\cos\left(\alpha(t)\xi_{12}\right)
\end{pmatrix} ,
\eea
where we have set $\xi_{12}=\sqrt{\omega_1^2+\omega_2^2}$.
}
\end{example}

\section{Mean-field dynamics of quantum fluctuations}
\label{ExpFun}

In the previous section, we studied the large $N$ limit of the dissipative dynamics generated by~\eqref{mflind1a0} on (quasi) local spin operators.
In this section we shall instead investigate the time-evolution of fluctuation operators scaling themselves with the inverse square-root of $N$. 

As a set $\cx$ of relevant one-site observables (see~\eqref{spaceChi}), we choose the  orthonormal basis of hermitian matrices $\{v_\mu\}_{\mu=1}^{d^2}$ appearing in $\LL^{(N)}$. Accordingly, we shall focus upon the vector  $\vec{F}^{(N)}$ of local fluctuations 
$$
F^{(N)}_\mu:=\frac{1}{\sqrt{N}}\sum_{k=0}^{N-1}\Big(v_\mu^{(k)}-\omega(v_\mu)\Big)\ ,
$$
and upon the local exponential operators in \eqref{qfa2},
$$
W^{(N)}(\vec{r})={\rm e}^{i\,\vec{r}\cdot\vec{F}^{(N)}}\ , \qquad \vec{r}\in\RR^{d^2}\ .
$$
As seen in Section \ref{qfsec}, if the  matrices $\{v_\mu\}_{\mu=1}^{d^2}$ give rise to normal fluctuations with respect to the translation-invariant, clustering state 
$\omega$, then 
$$
\lim_{N\to\infty}\omega\left(W^{(N)}(\vec{r})\right)=\Omega\left(W(\vec{r})\right)\ .
$$
In the above expression, $W(r)$ are operators with Weyl commutation relations and $\Omega$ is a  
Gaussian state on the Weyl algebra $\cw(\chi,\sigma^{(\omega)})$ arising from the strong-closure of their linear span with respect to the $GNS$-representation based on $\Omega$.

As already remarked in the previous section, the microscopic state $\omega$ need not be time-invariant, $\omega^{(N)}_t:=\omega\circ\gamma^{(N)}_t\neq \omega$, where $\gamma^{(N)}_t$ in~\eqref{semig} acts trivially outside $\ca_{[0,N-1]}$. Then, since fluctuations account for deviations of observables from their mean values that now depend on time, it is necessary to change the time-independent formulation of local quantum fluctuations given in \eqref{FL} into a time-dependent one, 
\be
F^{(N)}_\mu(t)=\frac{1}{\sqrt{N}}\sum_{k=0}^{N-1}\Big(v_\mu^{(k)}\,-\,\omega^{(N)}_t(v^{(k)}_\mu)\Big)\ ,
\label{fluct}
\ee
the time-dependence occurring through the mean-values. Then, the commutator of two such local fluctuations,
\be
\label{tcomm}
\left[F^{(N)}_\mu(t)\,,\,F^{(N)}_\nu(t)\right]=\frac{1}{N}\sum_{k=0}^{N-1}\left[v_\mu^{(k)}\,,\,v_\nu^{(k)}\right]=:T^{(N)}_{\mu\nu}
\ee
is a time-independent mean-field operator. However, the entries of the symplectic matrix in~\eqref{sympform}, 
\be
\label{tsympform1}
\sigma^{(\omega)}_{\mu\nu}(t):=
-i\lim_{N\to\infty}\omega^{(N)}_{\mu\nu}(t)=-i\omega_{\mu\nu}(t)\ ,\quad \omega^{(N)}_{\mu\nu}(t):=\omega^{(N)}_t\left(T^{(N)}_{\mu\nu}\right)\ ,
\ee
will in general explicitly depend on time. Notice that the last two equalities follow from \eqref{macr1b}, while from \eqref{aid1} one derives
\be
\label{tsympform2}
\sigma^{(\omega)}_{\mu\nu}(t)=-i\sum_{\alpha=1}^{d^2}\,J^\alpha_{\mu\nu}\,\omega_\alpha(t)\ .
\ee
As they depend on the initial vector $\vec{\omega}$ of mean-field observables, that is of macroscopic averages, and on the time-evolution  of
$\vec{\omega}$ into $\vec{\omega}_t$, for later convenience, we shall denote by $\sigma(\vec{\omega}_t)$ the symplectic matrix with components $\sigma^{(\omega)}_{\mu\nu}(t)$ and by $\sigma(\vec{\omega})$ the symplectic matrix at time $t=0$ with components
\be
\label{sympt0}
\sigma_{\mu\nu}(\vec{\omega})=-i\lim_{N\to\infty}\frac{1}{N}\sum_{k=0}^{N-1}
\omega\Big(\left[v_\mu^{(k)}\,,\,v_\nu^{((k)}\right]\Big)=-i\omega\Big(\left[v_\mu\,,\,v_\nu\right]\Big)=-i{\rm Tr}\Big(\rho\,\left[v_\mu\,,\,v_\nu\right]\Big)\ ,
\ee
where we have used the assumed translation-invariance of the state $\omega$. 
\medskip

\begin{remark}
{\rm 
Using the matrix $M_t(\vec{\omega})$ in \eqref{formalsol}, one gets 
\be
\label{sympdyn}
\sigma(\vec{\omega}_t)\,=\,M_t(\vec{\omega})\,\sigma(\vec{\omega})\,M^{tr}_t(\vec{\omega})\ .
\ee
Such a relation follows from \eqref{aid1} and \eqref{tsympform2} that yield
$$
\sigma^{(\omega)}_{\mu\nu}(t)=-i\sum_{\alpha=1}^{d^2}\,J^\alpha_{\mu\nu}\,\omega\left(\alpha_t\left[v_\alpha\right]\right)=-i\omega\left(\Big[\alpha_t\left[v_\mu\right]\,,\,\alpha_t\left[v_\mu\right]\Big]\right)\ .
$$
Then, taking the time-derivative of both sides of the above equality and using \eqref{macdyn2}, 
\be
\label{eqsigmadot}
\frac{{\rm d}}{{\rm d}t}\sigma(\vec{\omega}_t)=\Big[D(\vec{\omega}_t)\,,\,\sigma(\vec{\omega}_t)\Big]\ .
\ee
Note that the map $\sigma(\vec{\omega})\mapsto\sigma(\vec{\omega}_t)$ is non-linear since $D(\vec{\omega}_t)$ depends on $\sigma(\vec{\omega}_t)$.
\qed
}
\end{remark}

Let $\vec{F}^{(N)}$ be the operator-valued vector with components 
$$
F^{(N)}_\mu=\frac{1}{\sqrt{N}}\sum_{k=0}^{N-1}\Big(v_\mu^{(k)}\,-\,\omega^{(N)}(v^{(k)}_\mu)\Big)\ ,\ \mu=1,2,\ldots,d^2\ ,
$$
at $t=0$.  Given the local exponential operators
\be
\label{tNWeyl}
W^{(N)}(\vec{r})={\rm e}^{i\,\vec{r}\cdot\vec{F}^{(N)}}\ ,\quad \vec{r}\cdot\vec{F}^{(N)}=\sum_{\mu=1}^{d^2}r_\mu\,F^{(N)}_\mu\ ,
\ee
with respect to a translation invariant, clustering state $\omega$, in the mesoscopic limit (see Definition~\ref{meslimdef} in Section~\ref{subsecmeso}), they give rise to Weyl operators
\be
\label{Weyldiss0}
W(\vec{r})=\exp\Big(i\vec{r}\cdot\vec{F}\Big)=m-\lim_{N\to\infty}W^{(N)}(\vec{r})\ ,
\ee 
where the vector $\vec{F}$ has components $F_\mu$, $1\leq\mu\leq d^2$ given by
\be
\label{Weyldiss1}
F_\mu=m-\lim_{N\to\infty}F^{(N)}_\mu
\ee
and such that
\be
\label{Weyldiss2}
\Big[F_\mu\,,\,F_\nu\Big]=i\sigma_{\mu\nu}(\vec{\omega})\ .
\ee

\subsection{Structure of the symplectic matrix}
\label{reminv}

The density matrix $\rho$ that represents $\omega$ at each lattice site can be expanded as $\rho=\sum_{\mu=1}^{d^2}\,r_\mu\,v_\mu$ with respect the orthonormal matrix basis. It thus turns out that the corresponding generalised Bloch vector $\{r_\mu\}_{\mu=1}^{d^2}$ is in the kernel of the symplectic matrix,
$$
\sum_{\nu=1}^{d^2}\sigma_{\mu\nu}(\vec{\omega})\,r_\nu={\rm Tr}\Big(\rho\,[v_\mu\,,\,\rho]\Big)=0\ ,
$$
whence $\sigma(\vec{\omega})$ is not invertible.
Actually, the kernel of the symplectic matrix is at least $d$-dimensional for it also contains the generalized Bloch vectors corresponding to the eigenprojectors of $\rho$.

By an orthogonal rotation $R(\vec{\omega})$, any non-invertible $\sigma(\vec{\omega})$ can be brought into the form
\be
\label{invsigma}
\widetilde{\sigma}(\vec{\omega})=R(\vec{\omega})\,\sigma(\vec{\omega})\,R^{tr}(
\vec{\omega})=\begin{pmatrix}
0&0\\0&\widetilde{\sigma}^{11}(\vec{\omega})\end{pmatrix}\ ,
\ee
where the diagonal zero entry stands for a zero $d_0(\vec{\omega})\times d_0(\vec{\omega})$ square-matrix, while off-diagonal zeroes stand for $d_0(\vec{\omega})\times d_1(\vec{\omega})$ and $d_1(\vec{\omega})\times d_0(\vec{\omega})$ $0$ zero rectangular matrices, while $\widetilde{\sigma}^{11}(\vec{\omega})$ is a $d_1(\vec{\omega})\times d_1(\vec{\omega})$ invertible symplectic matrix, with $d^2\geq d_0(\vec{\omega})\geq d$ the dimension of the kernel of $\sigma(\vec{\omega})$ and $d_1(\vec{\omega})=d^2-d_0(\vec{\omega})$ an even integer.

The orthogonal rotation matrix $R(\vec{\omega})$ that transforms $\sigma(\vec{\omega})$ into
$\widetilde{\sigma}(\vec{\omega})$ does in general depend on the vector $\vec{\omega}$ and amounts to a rotation of the hermitian matrix basis 
$\{v_\mu\}_{\mu=1}^{d^2}$ into a new hermitian matrix basis $\left\{v_\alpha(\vec{\omega})=\sum_{\mu=1}^{d^2}R_{\alpha\mu}(\vec{\omega})v_\mu\right\}$.
One can thus rotate the operator-valued vector $\vec{F}$ into the form
\be
\label{Gops}
\vec{G}(\vec{\omega})=R(\vec{\omega})\vec{F}
\ee
so that the commutation relations~\eqref{Weyldiss2} turn into
\be
\label{Gops1}
\Big[G_\mu(\vec{\omega})\,,\,G_\nu(\vec{\omega})\Big]=\,i\,\sum_{\alpha,\beta=1}^{d^2}R_{\mu\alpha}(\vec{\omega})\,R_{\nu\beta}(\vec{\omega})\,\sigma_{\alpha\beta}(\vec{\omega})=\,i\,\widetilde{\sigma}_{\mu\nu}(\vec{\omega})\ .
\ee
Therefore, the first $d_0(\vec{\omega})$ components of $\vec{G}(\vec{\omega})$ commute with all the others and among themselves and constitute a commutative set. 
\medskip

\begin{definition}
\label{rotfluct}
By $\vec{G}_0(\vec{\omega})$ we will denote the $d_0(\vec{\omega})$-dimensional operator-valued vector consisting of the commuting components of $\vec{G}(\vec{\omega})$ and by $\vec{G}_1(\vec{\omega})$ the vector whose comonents are the remaining $d_1(\vec{\omega})$ operators.
\end{definition}
\medskip

Then, the Weyl operators~\eqref{Weyldiss0} split into the product of the exponentials of the commuting components of $\vec{G}_0(\vec{\omega})$ and a quantum Weyl operators that cannot be further split:
\be
\label{split}
W(\vec{r})=\left(\prod_{\mu=1}^{d_0(\vec{\omega})}\exp\Big(i\,s_\mu(\vec{\omega})\,G_\mu(\vec{\omega})\Big)\right)\,
\exp\left(i\sum_{\mu=d_0(\vec{\omega})+1}^{d^2}s_\mu(\vec{\omega})\,G_\mu(\vec{\omega})\right)\ ,
\ee
where $\vec{s}(\vec{\omega})=R(\vec{\omega})\vec{r}$. Furthermore, the rotation into the new matrix basis $\{v_\mu(\vec{\omega})\}_{\mu=1}^{d^2}$ amounts to rotating the Kossakowski matrix $C$ in the Lindblad generator~\eqref{mflind1a} into a new, $\vec{\omega}$-dependent Kossakowski matrix $C(\vec{\omega})=R(\vec{\omega})\,C\,R^{tr}(\vec{\omega})$ with symmetric and anti-symmetric components $A(\vec{\omega})=R(\vec{\omega})\,A\,R^{tr}(\vec{\omega})$ and $B(\vec{\omega})=R(\vec{\omega})\,B\,R^{tr}(\vec{\omega})$.

Because of~\eqref{sympdyn}, the matrix $\sigma(\vec{\omega}_t)$ remains non-invertible in the course of time.

\subsection{Mesoscopic dissipative dynamics}
\label{subsec4.2}

Given the local exponential operators 
$W^{(N)}(\vec{r})$ in~\eqref{tNWeyl}, we now study the mesoscopic limit of their dynamics at positive times $t\geq 0$:
$$
W^{(N)}(\vec{r})\mapsto W_t^{(N)}(\vec{r}):=\gamma^{(N)}_t\left[W^{(N)}(\vec{r})\right]\ .
$$ 
We shall prove the existence of the following limit (see Definition \ref{meslimdef}) 
\be
\label{meslim1}
\lim_{N\to\infty}\omega_{\vec{r}_1\vec{r}_2}\left(\gamma^{(N)}_t\left[W^{(N)}(\vec{r})\right]\right)
=\Omega_{\vec{r}_1\vec{r}_2}\left(\Phi^{\vec{\omega}}_t\Big[W(\vec{r})\Big]\right)\ ,\quad \forall \vec{r}_{1,2},r\in\RR^{d^2}\ ,
\ee
where $\Omega$ is the mesoscopic state emerging from the microscopic state 
$\omega$ at $t=0$ according to \eqref{quasistate}, $W(\vec{r})=\exp(i\vec{r}\cdot\vec{F})$ is any element of the Weyl algebra $\mathcal{W}(\chi,\sigma^{(\omega)})$ corresponding to the  matrix $\sigma(\vec{\omega})$ at time $t=0$ with the components of $\vec{F}$ satisfying the commutation relations \eqref{Weyldiss2}. These limits define the maps $\Phi^{\vec{\omega}}_t$ that describe the mesoscopic dynamics corresponding to the microscopic dissipative time-evolution $\gamma^{(N)}_t$; their explicit form is given in the following theorem whose proof is provided in Section \ref{app2}. 
\medskip

\begin{theorem}
\label{expfluct}
According to Definition \ref{meslimdef}, the dynamics of quantum fluctuations is given by 
the mesoscopic limit $\Phi_t^{\vec{\omega}}:=m-\lim_{N\to\infty}\gamma^{(N)}_t$, where
\bea
\label{mainres}
\Phi_t^{\vec{\omega}}\left[W(\vec{r})\right]&=&\exp\Big(-\frac{1}{2}\vec{r}\cdot\Big(Y_t(\vec{\omega})\,\vec{r}\Big)\Big)\,W(X_t^{tr}(\vec{\omega})\vec{r})\ ,
\eea
where, with $\TT$ denoting time-ordering, 
\bea
\label{lastaidd1}
X_t(\vec{\omega})&:&=\TT{\rm e}^{\int_0^t{\rm d}s\, Q(\vec{\omega}_s)}\\
\label{matrices1}
Q(\vec{\omega}_t)&:=&-i\sigma(\vec{\omega}_t)\,\widetilde{B}\,+\,D(\vec{\omega}_t)\\
\label{matrices3}
Y_t(\vec{\omega})&:=&\int_0^t{\rm d}s\, X_{t,s}(\vec{\omega})\,\Big(\sigma(\vec{\omega}_s)\,A\,\sigma^{tr}(\vec{\omega}_s)\Big)\,X^{tr}_{t,s}(\vec{\omega})\ .
\eea
In the above expression, $X_{t,s}(\vec{\omega}):=X_t(\vec{\omega})\,X^{-1}_s(\vec{\omega})$, $A$ is the symmetric component of the Kossakowski matrix $C$ in~\eqref{mflind1a00}, $\widetilde{B}=B+\,2\,i\,h^{(re)}$ in~\eqref{hmat2}.
Finally, $\sigma(\vec{\omega}_t)$ is the time-dependent symplectic matrix with entries given by \eqref{tsympform1} and $D(\vec{\omega}_t)$ is the matrix defined in \eqref{d0}.
\end{theorem}

The structure of the mesoscopic dynamics looks like that of Gaussian maps transforming Weyl operators onto Weyl operators with rotated parameters and further multiplied by a damping Gaussian factor. Indeed, the time evolution sends $\vec{r}$ into $X_t^{tr}(\vec{\omega})\vec{r}$ and the exponent 
$\vec{r}\cdot\Big(Y_t(\vec{\omega})\,\vec{r}\Big)$
in the prefactor is positive since 
$\displaystyle A=\frac{C+C^{tr}}{2}\geq 0$ because such is the Kossakowski matrix $C$.
However, as we shall see in the next section, the dependence on the macroscopic dynamics 
of mean-field quantities makes the maps $\Phi^{\vec{\omega}}_t$ non-linear on the Weyl algebra $\cw(\chi,\sigma^{(\omega)})$.

\subsection{Structure of the mesoscopic dynamics}
\label{subsec4.3}

In this section we discuss in detail the properties of the mesoscopic dynamics defined by the maps $\Phi^{\vec{\omega}}_t$, $t\geq 0$ in \eqref{mainres}.
It turns out that they act non-linearly on products of Weyl operators. Indeed, if $\Phi^{\vec{\omega}}_t$ were linear, using \eqref{Weyl}, one would get
$$
\Phi^{\vec{\omega}}_t\left[W(\vec{r_1})W(\vec{r}_2)\right]=\Phi^{\vec{\omega}}_t\left[{\rm e}^{i\,\vec{r}_2\cdot\left(\sigma(\vec{\omega})\vec{r}_1\right)}\,W(\vec{r_2})W(\vec{r}_1)\right]=
{\rm e}^{i\,\vec{r}_2\cdot\left(\sigma(\vec{\omega})\vec{r}_1\right)}\,\Phi^{\vec{\omega}}_t\left[\,W(\vec{r_2})W(\vec{r}_1)\right]\ .
$$ 
Instead, the following proposition shows that the symplectic matrix in the exponent at the right hand side of the above equality is not  $\sigma(\vec{\omega})$ at $t=0$, rather $\sigma(\vec{\omega}_t)$ at time $t>0$. 
This is a consequence of the fact that the local operators $W^{(N)}(\vec{r})$ and $W^{(N)}(\vec{s})$ satisfy a Baker-Campbell-Haussdorf relation of the form
$$
W^{(N)}(\vec{r})\,W^{(N)}(\vec{s})=W^{(N)}(\vec{s})\,W^{(N)}(\vec{r})\,\exp\left(\Big[\vec{s}\cdot\vec{F}^{(N)}\,,\,\vec{r}\cdot\vec{F}^{(N)}\Big]\,+\,O\left(\frac{1}{\sqrt{N}}\right)\right)\ .
$$
Since the leading order term in the argument of the exponential function is a mean-field quantity, it keeps evolving in time 
under the action of $\gamma^{(N)}_t$ in the large $N$ limit and tends to the scalar quantity $i\,\vec{s}\cdot\left(\sigma(\vec{\omega}_t)\vec{r}\right)$.
This result is formally derived in the proof of the following Proposition given in Section~\ref{subsgen}.
\medskip

\begin{proposition}
The mesoscopic dynamics of the product of two Weyl operators satisfies
\be
\label{pcomm}
\Phi^{\vec{\omega}}_t\left[W(\vec{r})\,W(\vec{s})\right]={\rm e}^{i\,\vec{s}\cdot\left(\sigma(\vec{\omega}_t)\vec{r}\right)}\,\Phi^{\vec{\omega}}_t\left[W(\vec{s})\,W(\vec{r})\right]\ ,\qquad\forall\,\vec{r}\,,\,\vec{s}\in\RR^{d^2}\ .
\ee
\label{ccrt}
\end{proposition}

The non-linearity of the fluctuation dynamics conflicts with the fact that any dissipative quantum dynamics should be described by a semigroup of linear, completely positive maps. Notice that, even if systems with time-dependent macroscopic averages have already been studied~\cite{sew}, the puzzling result of Proposition~\ref{ccrt} had not yet emerged for, in the framework of quantum fluctuations theory only time-invariant states have been considered so far. In order to reconcile the result of Proposition~\ref{ccrt} with the desired behaviour of quantum dynamical maps, one needs to identify the proper mesoscopic algebra suited to time-evolving canonical commutation relations. One has indeed to consider quantum fluctuations obeying different algebraic rules that depend on the macroscopic averages. The proper tool is offered by an extended algebra that allows to account for the dynamics of quantum fluctuations with time-varying commutation relations.
One is thus led to deal with a peculiar hybrid system, in which there appear together quantum and classical degrees of freedom, strongly connected since the commutator of two fluctuations is a classical dynamical variable. Remarkably, the need for such a mathematical setting naturally emerges from a concrete many-body quantum system as the dissipative quantum spin chain discussed above.

The maps $\Phi^{\vec{\omega}}_t$ can be extended to linear maps $\Phi^{ext}_t$ on a
larger algebra than $\cw(\chi,\sigma^{(\omega)})$.  
Via the relations \eqref{tsympform2}, the algebra $\cw(\chi,\sigma^{(\omega)})$ does indeed depend on the vector $\vec{\omega}$ of macroscopic averages at time $t=0$. We shall then denote it by $\cw_{\vec{\omega}}$ and by $W_{\vec{\omega}}(\vec{r}(\vec{\omega}))$ its Weyl operators, where we further include the possibility that the vectors parametrizing the Weyl operators also depend on $\vec{\omega}$.
We shall assume that, for all $\vec{\omega}\in\cs$, the representation of the Weyl algebra be regular so that 
\bea
\label{extweyl0}
W_{\vec{\omega}}(\vec{r}(\vec{\omega}))=\exp\left(i\vec{r}(\vec{\omega})\cdot\vec{F}(\vec{\omega})\right)\ ,
\eea
where $\vec{F}(\vec{\omega})$ is the operator-valued vector with components given by 
the Bosonic operators $F_\mu(\vec{\omega})$, $\mu=1,2,\ldots,d^2$, for each $\vec{\omega}\in\cs$ so that (compare with~\eqref{COMSIGMA}), 
\be
\label{COMSIGMA1}
\Big[\vec{r}_1(\vec{\omega})\cdot\vec{F}(\vec{\omega})\,,\,\vec{r}_2(\vec{\omega})\cdot\vec{F}(\vec{\omega})\Big]=i\,\vec{r}_1(\vec{\omega})\cdot\Big(\sigma(\vec{\omega})\vec{r}_2(\vec{\omega})\Big)
\ .
\ee
We are thus dealing with a so-called field of von Neumann algebras $\{\cw_{\vec{\omega}}\}_{\vec{\omega}\in\cs}$ that can be assembled together into a direct integral von Neumann algebra~\cite{LiBingRen}
\be
\label{dirint}
\cw^{ext}=\int^{\oplus}_{\cs}{\rm d}\vec{\omega}\,\cw_{\vec{\omega}}\ .
\ee 
The most general elements of $\cw^{ext}$ are operator-valued functions of the form
\be
\label{extweyl1}
W^f_{\vec{r}}:\,\cs\ni\vec{\omega}\mapsto f(\vec{\omega})\,W_{\vec{\omega}}(\vec{r}(\vec{\omega}))\ ,
\ee
with $f$ any element of the von Neumann algebra $L^\infty(\cs)$ of \textit{essentially bounded} functions on $\cs$ with respect to the measure ${\rm d}\vec{\omega}$, that is $f$ is measurable and bounded apart from sets of zero measure, while the Weyl operators $W_{\vec{\omega}}(\vec{r}(\vec{\omega}))\in\cw_{\vec{\omega}}$ correspond to the 
operator-valued functions $W^1_{\vec{r}}$ evaluated at 
$\vec{\omega}$; namely, $W^1_{\vec{r}}(\vec{\omega})=W_{\vec{\omega}}(\vec{r}(\vec{\omega}))$.
\medskip

\begin{remark}
\label{remext}
{\rm 
Notice that the extended algebra cannot be written in a simpler tensor form; indeed, each $\vec{\omega}$ determines its own Weyl algebra $\cw_{\vec{\omega}}$ and commutators of operators in $\cw_{\vec{\omega}}$ produce functions on $\cs$.
Only if the algebras $\cw_{\vec{\omega}}$ were the same, $\cw_{\vec{\omega}}=\cw$ for all $\vec{\omega}\in\cs$, one could write $\cw^{ext}=L^\infty(\cs)\otimes\cw$.\\
States on $\cw^{ext}$ are provided by general convex combinations of the form
\be
\label{extstate}
\Omega^{ext}=\int^\oplus_{\cs}{\rm d}\vec{\nu}\,\rho(\vec{\nu})\,\Omega_{\vec{\nu}}\ .
\ee
where $\Omega_{\vec{\nu}}$ is any state on the Weyl algebra $\cw_{\vec{\nu}}$ and $\rho$ is any probability distribution over $\cs$.
One may call Gaussian a state $\Omega^{ext}$ on $\mathcal{W}^{ext}$ if the $\Omega_{\vec{\nu}}$ in~\eqref{extstate} are all Gaussian and a specific Gaussian state $\Omega_{\vec{\omega}}$ on the Weyl algebra $\cw_{\vec{\omega}}$ can be selected by choosing a Dirac delta distribution localised at $\vec{\omega}\in\cs$, $\rho(\vec{\nu})=\delta_{\vec{\omega}}(\vec{\nu})$.
}\qed
\end{remark}
\medskip

On the extended algebra, we can then consider the extended linear maps $\Phi^{ext}_t$ defined by their action on the building blocks $W^f_{\vec{r}}$ of $\cw^{ext}$:
\be
\label{extweyl2}
\left(\Phi^{ext}_t[W^f_{\vec{r}}]\right)(\vec{\omega})=f(\vec{\omega}_t)\,\Phi^{\vec{\omega}}_t\left[W_{\vec{\omega}}(\vec{r}(\vec{\omega}_t))\right]\ .
\ee
Notice that $\Phi^{ext}$ makes all parametric dependences on $\vec{\omega}$ evolve in time
but for the one labelling the Weyl algebra which is left fixed.
Then, functions $f(\vec{\omega})$ and vectors $\vec{r}(\vec{\omega})$ are mapped into $f_t(\vec{\omega}):=f(\vec{\omega}_t)$, respectively  $\vec{r}_t(\vec{\omega}):=\vec{r}(\vec{\omega}_t)$, while, according to \eqref{mainres}, 
\bea
\label{lastaidd0}
\Phi_t^{\vec{\omega}}\left[W_{\vec{\omega}}(\vec{r}(\vec{\omega}_t))\right]&=&
g^{\vec{r},t}(\vec{\omega})\,W_{\vec{\omega}}\left(X_t^{tr}(\vec{\omega})\vec{r}(\vec{\omega}_t\right)\\
\label{extweyl3}
g^{\vec{r},t}(\vec{\omega})&:=&\exp\left(-\frac{1}{2}\Big(\vec{r}(\vec{\omega}_t)\cdot\Big(Y_t(\vec{\omega})\vec{r}(\vec{\omega}_t)\Big)\Big)\right)\ .
\eea
Notice that, because of the dependence of the matrix $Y_t(\vec{\omega})$ on the whole trajectory $\vec{\omega}\mapsto\vec{\omega}_t$, and not only on the end value $\vec{\omega}_t$, the functions $g^{\vec{r},t}(\vec{\omega})\neq g_t^{\vec{r}}(\vec{\omega}):=g^{\vec{r}}(\vec{\omega}_t)$. On the other hand, if the vector $\vec{r}(\vec{\omega})=\vec{r}$ does not explicitly depend on $\vec{\omega}$, then it does not evolve in time 
and one recovers the action \eqref{mainres} of the non-linear maps $\Phi_t$ of which the maps $\Phi^{ext}_t$ are indeed linear extensions.

The action of the extended dynamical maps can then be recast as
\be
\label{extweyl4}
\Phi^{ext}_t[W^f_{\vec{r}}]=f_t\,g^{\vec{r},t}\,W^1_{X_t^{tr}\vec{r}_t}\ ,
\ee
where it is understood that, when evaluating such an operator valued function at $\vec{\omega}\in\cs$, the matrix-valued function $X_t$ becomes $X_t(\vec{\omega})$, so that $\Big(X_t^{tr}\vec{r}_t\Big)(\vec{\omega})=X_t^{tr}(\vec{\omega})\vec{r}(\vec{\omega}_t)$.

Notice that the maps $\Phi^{ext}_t$ reproduce the time-dependent algebraic relations \eqref{pcomm}.
Indeed, setting
$E(\vec{\omega}):=\exp(i\,\vec{r}_1\cdot\left(\sigma(\vec{\omega})\vec{r}_2\right))$, with $\vec{\omega}$-independent vectors, then 
\beann
\left(\Phi^{ext}_t[W^1_{\vec{r}_1}W^1_{\vec{r}_2}]\right)(\vec{\omega})&=&\left(\Phi^{ext}_t[E\,W^1_{\vec{r}_2}W^1_{\vec{r}_1}]\right)(\vec{\omega})
=E(\vec{\omega}_t)\,
\Phi^{ext}_t[W^1_{\vec{r}_2}W^1_{\vec{r}_1}]\\
&=&\exp(i\,\vec{r}_1\cdot\left(\sigma(\vec{\omega}_t)\vec{r}_2\right))\,\Phi^{ext}_t[W^1_{\vec{r}_2}W^1_{\vec{r}_1}]\ ,
\eeann
with $E_t(\vec{\omega})=E(\vec{\omega}_t)$.

The expression \eqref{extweyl4} is best suited to inspect the composition law of the extended maps:
\bea
\nonumber
\Phi_s^{ext}\circ\Phi_t^{ext}[W^f_{\vec{r}}]&=&\Phi^{ext}_s\left[f_t\,g^{\vec{r},t}\,W^1_{\vec{r}_t}\right]=(f_t)_s\,g_s^{\vec{r},t}\,\Phi^{ext}_s\left[W^1_{X_t^{tr}\vec{r}_t}\right]\\
\label{100}
&=&(f_t)_s\,g_s^{\vec{r},t}\,g^{(X^{tr}_t\vec{r}_t)_s,s}\,W^1_{X^{tr}_s(X^{tr}_t\vec{r}_t)_s}\ .
\eea
When evaluated at $\vec{\omega}$, using \eqref{extweyl3}, the right hand side yields
\bea
\label{extweyl6a}
&&\hskip-.7cm
\Phi_s^{ext}\circ\Phi_t^{ext}[W^f_{\vec{r}}](\vec{\omega})=
(f_t)_s(\vec{\omega})\,g^{\vec{r},t}(\vec{\omega}_s)\,g^{(X^{tr}_t\vec{r}_t)_s,s}(\vec{\omega})\,W_{\vec{\omega}}\left(\Big(X^{tr}_s(X^{tr}_t\vec{r}_t)_s\Big)(\vec{\omega})\right)\\
\label{extweyl6d}
&&\hskip-.7cm
\Big(X^{tr}_s(X^{tr}_t\vec{r}_t)_s\Big)(\vec{\omega})=X^{tr}_s(\vec{\omega})\,X_t^{tr}(\vec{\omega}_s)\vec{r}((\vec{\omega}_s)_t)\\
\label{extweyl6b}
&&\hskip-.7cm
g_s^{\vec{r},t}(\vec{\omega})=g^{\vec{r},t}(\vec{\omega}_s)=
\exp\left(-\frac{1}{2}\vec{r}((\vec{\omega}_s)_t)\cdot\Big(Y_t(\vec{\omega}_s)\,\vec{r}((\vec{\omega}_s)_t)\Big)\right)\\
\label{extweyl6c}
&&\hskip-.7cm
g^{(X^{tr}_t\vec{r}_t)_s,s}(\vec{\omega})=\exp\left(-\frac{1}{2}\Big(X^{tr}_t(\vec{\omega}_s)\vec{r}((\vec{\omega}_s)_t)\Big)\cdot\Big(Y_s(\vec{\omega})\,X_t^{tr}(\vec{\omega}_s)\vec{r}((\vec{\omega}_s)_t)\Big)\right)\ .
\eea
The dependence on $\vec{\omega}_s$ of the matrix $Y_t(\vec{\omega}_s)$ means that the macroscopic trajectories over which the various integral \eqref{matrices1}-\eqref{matrices3} are computed originates from $\vec{\omega}_s$. Since the motion along a macroscopic trajectory composes in such a way that 
$(\vec{\omega}_s)_t=(\vec{\omega}_t)_s=\vec{\omega}_{t+s}$  for all $s,t\geq 0$ (see \eqref{avtime}), on one hand $(f_t)_s(\vec{\omega})=f_{t+s}(\vec{\omega})$, $\vec{r}((\vec{\omega}_s)_t)=\vec{r}(\omega_{s+t})$, while
\bea
\label{extweyl7a}
&&
X_t(\vec{\omega}_s)=\TT{\rm e}^{\int_0^t{\rm d}u\, Q(\vec{\omega}_{s+u})}=
\TT{\rm e}^{\int_s^{t+s}{\rm d}u\, Q(\vec{\omega}_u)}=X_{t+s}(\vec{\omega})\,X^{-1}_s(\vec{\omega})\qquad \hbox{\rm whence}\\
\label{extweyl7b}
&&
X_{t,u}(\vec{\omega}_s)=
X_t(\vec{\omega}_s)\,X^{-1}_u(\vec{\omega}_s)=X_{t+s}(\vec{\omega})\,X^{-1}_{u+s}(\vec{\omega})\ .
\eea
From the first relation it follows that
\be
\label{extweyle}
X^{tr}_s(\vec{\omega})\,X^{tr}_t(\vec{\omega}_s)=X^{tr}_{s+t}(\vec{\omega})\ ,
\ee
while the second one yields
\bea
\nonumber
Y_t(\vec{\omega}_s)&=&\int_0^t{\rm d}u\, X_{t,u}(\vec{\omega}_s)\,\sigma(\vec{\omega}_{u+s})\,A\,\sigma^{tr}(\vec{\omega}_{u+s})\,X^{tr}_{t,u}(\vec{\omega}_s)\\
\nonumber
&=&
\int_s^{t+s}{\rm d}u\, X_{t+s,u}(\vec{\omega})\,\sigma(\vec{\omega}_u)\,A\,\sigma^{tr}(\vec{\omega}_u)\,X_{t+s,u}^{tr}(\vec{\omega})\\
&=&Y_{t+s}(\vec{\omega})-\int_0^s{\rm d}u\, X_{t+s,u}(\vec{\omega})\,\sigma(\vec{\omega}_u)\,A\,\sigma^{tr}(\vec{\omega}_u)\,X^{tr}_{t+s,u}(\vec{\omega})\ .
\label{extweyl7c}
\eea
Furthermore, using~\eqref{extweyl7a} and \eqref{extweyl7b},
\be
\label{extweyl7d}
X_t(\vec{\omega}_s)\,Y_s(\vec{\omega})\,X^{tr}_t(\vec{\omega}_s)=\int_0^s{\rm d}u\, X_{t+s,u}(\vec{\omega})\,\sigma(\vec{\omega}_u)\,A\,\sigma^{tr}(\vec{\omega}_u)\,X^{tr}_{t+s,u}(\vec{\omega})\ .
\ee
Together with \eqref{extweyl7c}, it yields
$$
g_s^{\vec{r},t}(\vec{\omega})\,g^{(X^{tr}_t\vec{r}_t)_s,s}(\vec{\omega})=
\exp\left(-\frac{1}{2}\,\vec{r}(\vec{\omega}_{s+t})\cdot\Big(Y_{s+t}(\vec{\omega})\,\vec{r}(\vec{\omega}_{s+t})\Big)\right)=g^{\vec{r},s+t}(\vec{\omega})\ .
$$
In conclusion, \eqref{100} becomes
\be
\label{extweylf}
\Phi_s^{ext}\circ\Phi_t^{ext}[W^f_{\vec{r}}]=f_{t+s}\,g^{\vec{r},t+s}\,W^1_{\vec{r}_{s+t}}=\Phi_{t+s}^{ext}[W^f_{\vec{r}}]\ ,
\ee
whence the extended maps $\Phi^{ext}_t$ satisfy a semigroup composition law.
\medskip

As stated in the following Proposition whose proof is given in Section~\ref{subsgen}, the linear extended maps $\Phi^{ext}_t$ on the direct integral von Neumann algebra $\cw^{ext}$ are also completely positive.
\medskip

\begin{theorem}
The maps $\Phi^{ext}_t$ in \eqref{extweyl2} form a one parameter family of completely positive, unital, Gaussian maps on the von Neumann algebra $\cw^{ext}$.
\label{CPos}
\end{theorem}
\medskip

Since the maps $\Phi_t^{ext}$ form a semigroup on $\cw^{ext}$, their generator $\LL^{ext}$ is obtained by taking the time-derivative of $\Phi_t^{ext}$ at $t=0$ and will be of the form $\displaystyle\LL^{ext}=\int_{\cs}^\oplus{\rm d}\vec{\omega}\,\LL_{\vec{\omega}}$. The components $\LL_{\vec{\omega}}$
cannot be of the typical Lindblad form that is expected of the generators of Gaussian completely positive semigroups,
\beann
\hskip-1cm
\LL_{\vec{\omega}}\Big[W^f_{\vec{r}}(\vec{\omega})\Big]&=&i\,\Big[\sum_{\mu\nu=1}^{d^2}H_{\mu\nu}(\vec{\omega})\,F_\mu(\vec{\omega})\,F_\nu(\vec{\omega})\,,\,W^f_{\vec{r}}(\vec{\omega})\Big]\\
\hskip-1cm
&+&\sum_{\mu,\nu=1}^{d^2}K_{\mu\nu}(\vec{\omega})\Big(F_\mu(\vec{\omega})\,W^f_{\vec{r}}(\vec{\omega})\,F_\nu(\vec{\omega})\,-\,\frac{1}{2}\left\{F_\mu(\vec{\omega})\,F_\nu(\vec{\omega})\,,\,W^f_{\vec{r}}(\vec{\omega})\right\}\Big)\ .
\eeann
If it were so, then
$$
\LL_{\vec{\omega}}[W^f_{\vec{r}}(\vec{\omega})]=f(\vec{\omega})\,\LL_{\vec{\omega}}[W^1_{\vec{r}}(\vec{\omega})]\ ,
$$
and scalar functions would remain constant in time.
We will show that the generator is of hybrid form~\cite{Ciccotti}--\cite{Fratino} with 
\begin{itemize}
\item
a drift contribution that makes $\vec{\omega}$ evolve in time as a solution to the dynamical equation~\eqref{d0};
\item
mixed classical-quantum contributions;
\item 
fully quantum contributions.
\end{itemize}

Intriguingly, despite the complete positivity of the maps $\Phi^{ext}_t$, we will show that the fully quantum terms of the generator need not be of Lindblad form.

As we shall soon see, one has to take into account the non-invertibility of the symplectic matrix $\sigma(\vec{\omega})$. According to Section~\ref{reminv}, by means of a suitable orthogonal transformation $R(\vec{\omega})$, $\sigma(\vec{\omega})$ can always be brought into the form~\eqref{invsigma} and the Weyl operators decomposed 
into a classical and quantum contribution as in~\eqref{split}. In the following, after rotating a given $d^2\times d^2$ matrix into $X(\vec{\omega})=R(\vec{\omega})\,X\,R^{tr}(\vec{\omega})$, we shall decompose it as 
\be
\label{01dec}
X(\vec{\omega})=\begin{pmatrix}X^{00}(\vec{\omega})&
X^{01}(\vec{\omega})\cr
X^{10}(\vec{\omega})&X^{11}(\vec{\omega})\end{pmatrix}\ ,
\ee 
where, as in Remark~\ref{reminv}, $X^{00}(\vec{\omega})$ is a 
$d_0(\vec{\omega})\times d_0(\vec{\omega})$ matrix, $X^{01}(\vec{\omega})$ a $d_0(\vec{\omega})\times d_1(\vec{\omega})$ matrix, $X^{10}(\vec{\omega})$ a 
$d_1(\vec{\omega})\times d_0(\vec{\omega})$ matrix and $X^{11}(\vec{\omega})$ a $d_1(\vec{\omega})\times d_1(\vec{\omega})$ matrix, where $d_0(\vec{\omega})$ is the dimension of the kernel of $\sigma(\vec{\omega})$ and $d_1(\vec{\omega})=d^2-d_0(\vec{\omega})$.
\medskip

\begin{theorem}
\label{extgenth}
The extended dissipative dynamics $\Phi^{ext}_t$ of quantum fluctuations has a generator of the form
$\displaystyle\LL^{ext}=\int_{\cs}^\oplus{\rm d}\vec{\omega}\,\LL_{\vec{\omega}}$.
Every $\vec{\omega}$-component consists of four different contributions,
$\LL_{\vec{\omega}}=\LL^{drift}_{\vec{\omega}}+\LL^{cc}_{\vec{\omega}}+\LL^{cq}_{\vec{\omega}}+\LL^{qq}_{\vec{\omega}}$: a drift term
\be
\label{extgenth1a}
\LL^{drift}_{\vec{\omega}}\Big[W^f_{\vec{r}}(\vec{\omega})\Big]=\dot{\vec{\omega}}\cdot\partial_{\vec{\omega}}\,W^f_{\vec{r}}(\vec{\omega})\ ,
\ee 
a differential operator involving the classical degrees of freedom
$G_\mu(\vec{\omega})$, $\mu=1,2,\ldots,d_0(\vec{\omega})$ introduced in~\eqref{Gops},
\be
\label{extgenth1b}
\LL^{cc}_{\vec{\omega}}\Big[W^f_{\vec{r}}(\vec{\omega})\Big]=\sum_{\mu,\nu=1}^{d_0(\vec{\omega})}\,D^{00}_{\mu\nu}(\vec{\omega})G_\nu(\vec{\omega})\,\frac{\partial\,W^f_{\vec{r}}(\vec{\omega})}{\partial G_\mu(\vec{\omega})}
\ee
a mixed classical-quantum term
\be
\label{extgenth1c}
\LL^{cq}_{\vec{\omega}}\Big[W^f_{\vec{r}}(\vec{\omega})\Big]=\sum_{\mu=1}^{d_0(\vec{\omega})}\sum_{\nu=d_0(\vec{\omega})+1}^{d^2}\frac{D^{01}_{\mu\nu}(\vec{\omega})}{2}\,\left\{G_\nu(\vec{\omega})\,,\,\frac{\partial\,W^f_{\vec{r}}(\vec{\omega})}{\partial G_\mu(\vec{\omega})}\right\}\ ,
\ee
and a purely quantum term 
\bea
\label{extgenth1d}
&&\hskip-1cm
\LL^{qq}_{\vec{\omega}}\Big[W^f_{\vec{r}}(\vec{\omega})\Big]=
i\,\sum_{\mu=1}^{d_0(\vec{\omega})}\sum_{\nu=d_0(\vec{\omega}))+1}^{d^2}\Big(H^{01}_{\mu\nu}(\vec{\omega})+H^{10}_{\nu\mu}(\vec{\omega})\Big)\,G_\mu(\vec{\omega})\,\Big[G_\nu(\vec{\omega})\,,\,W^f_{\vec{r}}(\vec{\omega})\Big]\\
\label{extgenth1e}
&&
+i\,\sum_{\mu,\nu=d_0(\vec{\omega})+1}^{d^2}H^{11}_{\mu\nu}(\vec{\omega})\,\Big[G_\mu(\vec{\omega})\,G_\nu(\vec{\omega})\,,\,W^f_{\vec{r}}(\vec{\omega})\Big]\\
\label{extgenth1g}
\hskip-1cm
&&
+\sum_{\mu,\nu=d_0(\vec{\omega})+1}^{d^2}K^{11}_{\mu\nu}(\vec{\omega})\Big(G_\mu(\vec{\omega})\,W^f_{\vec{r}}(\vec{\omega})\,G_\nu(\vec{\omega})\,-\,\frac{1}{2}\left\{G_\mu(\vec{\omega})\,G_\nu(\vec{\omega})\,,\,W^f_{\vec{r}}(\vec{\omega})\right\}\Big)\ ,
\eea
with the $d_1(\vec{\omega})\times d_1(\vec{\omega})$ matrix of coefficients $H^{11}(\vec{\omega})$ given by 
\be
\label{ht}
H^{11}(\vec{\omega})=(h^{(re)}(\vec{\omega}))^{11}\,+\,\frac{1}{4}\Big\{(\sigma^{11}(\vec{\omega}))^{-1}\,,\,D^{11}(\vec{\omega})\Big\}\ ,
\ee
where $(h^{(re)}(\vec{\omega}))^{11}$ is the $11$-component of the matrix $h^{(re)}$ in~\eqref{hmat1} rotated by $R(\vec{\omega})$ as in~\eqref{01dec}.

Further, the $d_0(\vec{\omega})\times d_1(\vec{\omega})$ matrix $H^{01}(\vec{\omega})$, respectively the $d_1(\vec{\omega})\times d_0(\vec{\omega})$ matrix $H^{01}(\vec{\omega})$ read
\bea
\label{extgenth1h}
H^{10}(\vec{\omega})&=&\frac{1}{2}(\sigma^{11}(\vec{\omega}))^{-1}\,D^{10}(\vec{\omega})\,-\frac{i}{2}\,B^{10}(\vec{\omega})\,+\,(h^{(re)}(\vec{\omega}))^{10}\\ 
H^{01}(\vec{\omega})&=&\frac{1}{2}D^{01}(\vec{\omega})\,
(\sigma^{11}(\vec{\omega}))^{-1}\,+\frac{i}{2}\,B^{01}(\vec{\omega})\,+\,(h^{(re)}(\vec{\omega}))^{10}\ ,
\eea
where $(h^{(re)}(\vec{\omega}))^{10}$ and $(h^{(re)}(\vec{\omega}))^{10}$ are the $10$ and $01$ components of the matrix $h^{(re)}$ in~\eqref{hmat1} rotated by $R(\vec{\omega})$ as in~\eqref{01dec}.

Finally, $K^{11}(\vec{\omega})$ amounts to 
\be
\label{extgenth1i}
K^{11}(\vec{\omega})\,=\,A^{11}(\vec{\omega})\,+\,B^{11}(\vec{\omega})\,+\,\frac{i}{2}\Big[(\sigma^{11}(\vec{\omega}))^{-1}\,,\,D^{11}(\vec{\omega})\Big]\ .
\ee
\end{theorem}
\medskip

\noindent
The generator components $\LL_{\vec{\omega}}$ are easily checked to be 
trace-preserving while hermiticity preservation, $\LL_{\vec{\omega}}[X]^\dag=\LL_{\vec{\omega}}[X^\dag]$, follows since the blocks $D^{ij}(\vec{\omega})$, $i,j=0,1$, are real matrices, and $B^{10}(\vec{\omega})$, $B^{01}(\vec{\omega})$, as much as $B(\vec{\omega})$, change sign under complex conjugation while $(h^{(re)}(\vec{\omega}))^{10}$
and $(h^{(re)}(\vec{\omega}))^{01}$ do not. Then, $H^{11}(\vec{\omega})$ and $K^{11}(\vec{\omega})$ are hermitian matrices whereas $H^{01}(\vec{\omega})$, $H^{10}(\vec{\omega})$ are real matrices.

\begin{remark}{\rm 
Notice that $\LL_{\vec{\omega}}$ contains purely classical, purely quantum and mixed classical-quantum contributions.  
Furthermore, the apparent Lindblad structure of the purely quantum contribution $\LL^{qq}_{\vec{\omega}}$ corresponds to a Kossakowski matrix $K^{11}(\vec{\omega})$ which is in general not positive semi-definite.
This is due to the correction to $C^{11}(\vec{\omega})=A^{11}(\vec{\omega})+B^{11}(\vec{\omega})\geq 0$ given by
$\displaystyle\frac{i}{2}\Big[(\widetilde{\sigma}^{11}(\vec{\omega}))^{-1}\,,\,D^{11}(\vec{\omega})\Big]$; the latter matrix is traceless and cannot thus be positive semi-definite whence the positivity condition $K^{11}(\vec{\omega})\geq 0$ can be violated. 
Interestingly, despite of this, $\LL=\int^{\oplus}_{\cs}{\rm d}\vec{\omega}\,\LL_{\vec{\omega}}$ still generates a semigroup of completely positive maps on the extended algebra. Though the dynamics on the extended algebra consists of a semigroup of completely positive maps, the fact that its generator is not in Lindblad form with positive Kossakowski matrix is because it mixes classical and quantum terms. In order to recover the standard expression one should proceed to a fully quantum rendering of the evolution, by lifting the classical contributions to a larger non-commutative algebra in such a way that the generator in theorem \ref{extgenth} emerges as a restriction to a suitable commutative sub-algebra: a similar approach was proposed in a rather different context in \cite{Diosi}. \qed}
\end{remark}

\begin{remark}
\label{remgen2}
{\rm
Unlike the dissipative fluctuation time-evolution $\Phi^{\vec{\omega}}_t$ which is 
non-linear, the unitary time-evolution $\alpha_t$ on the quasi-local algebra $\ca$ given by Theorem~\ref{limdyn} is linear and does not need to be extended to a larger algebra in order to be an acceptable quantum transformation. 
However, if, in analogy to what has been done for $\Phi^{\vec{\omega}}_t$, one introduces an extended algebra $\ca^{ext}$ whose elements are operator-valued functions $O$ on $\cs$ with values in $\ca$, $\vec{\omega}\mapsto O_{\vec{\omega}}\in\ca$, unlike in Remark~\ref{remext}, at each $\vec{\omega}$ we have the same quasi-local $C^*$ algebra $\ca$, so $\ca^{ext}=L^\infty(\cs)\otimes\ca$.
Then, the extended algebra is generated by operators of the form $O_f$ such that $O_f(\vec{\omega})=f(\vec{\omega})\,O$, with $f\in L^\infty(\cs)$ and $O$ any local spin operator with finite support.
We then define $\alpha_t^{ext}$ on $\ca^{ext}$ as follows,
$$
\alpha_t[O_f](\vec{\omega})=f(\vec{\omega}_t)\,\alpha^{\vec{\omega}}_t[O]\ ,
$$
where $\vec{\omega}\mapsto\vec{\omega}_t$ as in~\eqref{formalsol} and 
$\alpha^{\vec{\omega}}_t$ given by~\eqref{auto2} with unitary operators 
$U^{(S)}_t(\vec{\omega})$ generated by hamiltonians $H_s=H(\vec{\omega}_s)$ where the dependence on $\vec{\omega}$ is now made explicit.
It then follows that we again obtain a semigroup on $\ca^{ext}$; indeed,
$$
\alpha_t\circ\alpha_s[O](\vec{\omega})=\alpha^{\vec{\omega}}_t\circ\alpha^{\vec{\omega}_t}_s[O(\vec{\omega}_{t+s})]=\alpha^{\vec{\omega}}_{t+s}[O]\ ,
$$
since, 
\beann
&&
\hskip-.8cm
U^{(S)}_s(\vec{\omega}_t)U^{(S)}_t(\vec{\omega})=\mathbb{T}{\rm e}^{-i\int_0^s{\rm d}u\, H^{(S)}(\vec{\omega}_{t+u})}\,\mathbb{T}{\rm e}^{-i\int_0^t{\rm d}u\, H^{(S)}(\vec{\omega}_u)}=\mathbb{T}{\rm e}^{-i\int_t^{t+s}{\rm d}u\, H^{(S)}(\vec{\omega}_u)}\,\mathbb{T}{\rm e}^{-i\int_0^t{\rm d}u\, H^{(S)}_u(\vec{\omega})}\\
&&\hskip 2.2cm
=\mathbb{T}{\rm e}^{-i\int_0^{t+s}{\rm d}u\, H^{(S)}(\vec{\omega}_u)}=U^{(S)}_{t+s}(\vec{\omega})\ .
\eeann
By taking the time-derivative of $\alpha^{\vec{\omega}}_t[O]$ at time $t=0$, a time-independent generator is obtained, of the form 
$$
\mathcal{K}=\int_{\cs}^{\oplus}{\rm d}\vec{\omega}\,\mathcal{K}_{\vec{\omega}}\ ,\qquad
\mathcal{K}_{\vec{\omega}}=\mathcal{K}^{drift}_{\vec{\omega}}\,+\,\mathcal{K}_{\vec{\omega}}^{qq}\ .
$$
It is a hybrid generator characterised by the absence of mixed classical-quantum contributions, by a purely classical drift part and a purely quantum contribution; explicitly, they read (compare~\eqref{auto3} at $t=0$):
$$
\mathcal{K}^{drift}_{\vec{\omega}}[O_f(\vec{\omega})]=\Big(\dot{\vec{\omega}}\cdot\partial_{\vec{\omega}}f(\vec{\omega})\Big)\,O\ ,\qquad
\mathcal{K}^{qq}_{\vec{\omega}}[O_f(\vec{\omega})]=f(\vec{\omega})\,\sum_{\mu,\nu=1}^{d^2}B_{\mu\nu}\,\omega_\nu\,\sum_{k=0}^{S-1}\Big[v_\mu^{(k)}\,,\,O\Big]\ .
$$
}\qed
\end{remark}

\section{Conclusions}

We have considered a quantum spin chain subjected to a purely dissipative mean-field
quantum dynamics. By endowing the quantum spin chain with a state not left invariant by the time-evolution, we studied the infinite volume limit of the latter on three algebras of observables.
ù
The first algebra consists of commuting macroscopic averages that behave as classical degrees of freedom obeying macroscopic equations of motions; the second algebra, build from quasi-local spin operators, despite the dissipative character of the microscopic dynamics, undergoes a unitary time-evolution with a homogeneous time-dependent hamiltonian.
Finally, the third class of observables taken into consideration represents a mesoscopic description level associated with suitable quantum fluctuations showing a collective bosonic behaviour. Due to the time-dependence of the canonical commutation relations obeyed by the fluctuations operators, the mesoscopic degrees of freedom also behave dissipatively, but their dynamics is not directly interpretable in terms of linear, completely positive maps.

We have thus extended the algebra of quantum fluctuation to accommodate the fact that
macroscopic averages and quantum fluctuations are both dynamical variables.
The issue is not only mathematically interesting, but also of physical relevance since in almost all experimental setups the macroscopic properties of the system actually vary in time.

On the extended algebra the non-linear fluctuations dynamics becomes linear, Gaussian and completely positive, giving rise to hybrid dynamical semigroups. 
Quantum fluctuations have also been experimentally investigated probing systems made of large number of atoms, and quantum effects have been reported \cite{Krauter,Julsgaard,Hammerer}. Collective spin operators of these atomic many-body systems, once scaled by the inverse square root of the number of particles, have been observed to obey a bosonic algebra. For this reason, they have been named {\sl mechanical oscillators}: they might provide a suitable concrete physical scenario where to test the theoretical results here reported.

\section{Proofs}
\label{sec5}

We first prove Theorem~\ref{limdyn} which provides the unitary dynamics of quasi-local observables and then Theorem~\ref{expfluct} which establishes the form of the dissipative dynamics of quantum fluctuations.

\subsection{Dynamics of local observables}
\label{Applemmacors}

We begin with the proof of Lemma~\ref{LhN} which provides a bound on the norm of the 
action of powers of the generator $\LL^{(N)}$ in~\eqref{mflind1a} on products, $P^{(N)}$, of mean-field and strictly local operators.
Consequences of this fact are Corollary~\ref{exg} which asserts that the series
$$
\gamma^{(N)}_t[P^{(N)}]={\rm e}^{t\,\LL^{(N)}}[P^{(N)}]=\sum_{k=0}^\infty\frac{t^k}{k!}\,\left(\LL^{(N)}\right)^k[P^{(N)}]\ ,
$$ 
converges uniformly in $N$ for $t\geq 0$ in a suitable finite interval of time, and 
Corollary~\ref{spt} which states that $\gamma^{(N)}$ behaves almost automorphically on products of $P^{(N)}$. 
These two latter facts will then be used  to derive firstly the time-evolution of microscopic averages in Proposition~\ref{mfdyn} and then the dynamics of quasi-local operators of the quantum spin chain in Theorem~\ref{limdyn}.
\medskip

\begin{lemma}
Let $P^{(N)}\in\ca$ be a spin operator of the form 
\be
\label{specop}
P^{(N)}=X_1^{(N)}\,X_2^{(N)}\,\dots\, X_{m-1}^{(N)}\,O\,X_m^{(N)}\,\dots\, X_p^{(N)}\ ,
\ee
where $O$ is a strictly local operator and $\displaystyle X_j^{(N)}=\frac{1}{N}\sum_{k=0}^{N-1}x_j^{(k)}$ is a mean-field operator as in \eqref{macro} for all $1\leq j\leq p$. Then, 
$$
\label{boundlem1}
\|\left(\LL^{(N)}\right)^k[P^{(N)}]\|\le\frac{(k+p)!}{p!}\left(2\,c\, v^2\,\ell(O)\,(d^2+2d^4)\right)^k\,\|O\|\,x^p \ ,
$$
where $\LL^{(N)}$ is the generator in~\eqref{mflind1a}, $\ell(O)$ is, according to Definition~\ref{defsupp}, is the finite support of $O$ and
$$
c=\max_{\mu,\nu}\big\{|A_{\mu\nu}|,|B_{\mu\nu}|\,,\,|h^{(re)}_{\mu\nu}|\,,\,|h^{(im)}_{\mu\nu}|\,,\,
|\epsilon_\mu|\big\}\, ,\, v=\max_\mu\{\|v_\mu\|\}\,,\,
x=\max_{1\leq j\leq p}\{\|x_j\|\}\ .
$$
\label{LhN}
\end{lemma}
\medskip

\medskip

\begin{proof}
Firstly, let us consider the action on $P^{(N)}$ of 
$\mathbb{H}^{(N)}$ in~\eqref{submap0}: it consists of 
the sum of at most $d^2$ terms of the form 
\beann
\sqrt{N}\,\Big[V^{(N)}_\mu\,,\,P^{(N)}\Big]&=&
\sqrt{N}\,\sum_{j=1}^{m-1}X_1^{(N)}\cdots X_{j-1}^{(N)}\,\Big[V^{(N)}_\mu\,,\,X_j^{(N)}\Big]\,X_{j+1}^{(N)}\cdots X_{m-1}^{(N)}\,O\cdots X_p^{(N)}\\
&+&
\sqrt{N}\,X_1^{(N)}\,\cdots\,X_{m-1}^{(N)}\,\Big[V^{(N)}_\mu,\,\,O\Big]\,X_m^{(N)}\cdots
X_p^{(N)}\\
&+&\sqrt{N}\,\sum_{j=m}^{p}X_1^{(N)}\cdots X_{m-1}^{(N)}\,O\,X_m^{(N)}\cdots\Big[V^{(N)}_\mu\,,\,X_j^{(N)}\Big]\cdots X_p^{(N)}\ .
\eeann
Notice that the commutators 
\be
\label{comm1a}
\Big[V^{(N)}_\mu\,,\,O\Big]=\frac{1}{\sqrt{N}}\sum_{k=0}^{N-1}[v_\mu^{(k)}\,,\,O]=\frac{1}{\sqrt{N}}\sum_{k\in\mathcal{S}(O)}[v_\mu^{(k)}\,,\,O]
\ee
scale as fluctuation operators since the sum is fixed by the finite support of $O$, while commutators of the form
\be
\label{comm1b}
\Big[V^{(N)}_\mu\,,\,X_j^{(N)}\Big]=\frac{1}{\sqrt{N}}\frac{1}{N}\sum_{k=0}^{N-1}
\left[v^{(k)}_\mu\,,\,x^{(k)}_j\right]
\ee
scale as mean-field operators further multiplied by $1/\sqrt{N}$. Therefore, the action of $\HH^{(N)}$ on $P^{(N)}$ reduces to the sum of at most $d^2(p+1)$ monomials consisting of the products of a local operator and $p$ mean-field operators multiplied by the coefficients $\epsilon_\mu$. Moreover,
$$
\left\|X_j^{(N)}\right\|\leq x\ ,\ 
\left\|\frac{1}{N}\sum_{k=0}^{N-1}\left[v^{(k)}_\mu\,,\,x^{(k)}_j\right]\right\|\leq 2\,v\,\,x\ ,\ \left\|\sum_{k\in\mathcal{S}(O)}\Big[v_\mu^{(k)}\,,\,O\Big]\right\|\leq 2\,v\,\ell(O)\,\|O\|\ .
$$
On the other hand, $\widetilde{\DD}^{(N)}$ yields sums of at most $d^4$ terms of the form 
$[V^{(N)}_\mu\,,\,P^{(N)}]\,V^{(N)}_\nu$ and $V^{(N)}_\mu\,[P^{(N)}\,,\,V^{(N)}_\nu]$.
Then, the factor $1/\sqrt{N}$ in~\eqref{comm1a} and~\eqref{comm1b} can be used to turn the operator $V_\nu^{(N)}$ that scales as a fluctuation operator, into a mean-field one $V_\nu^{(N)}/\sqrt{N}$.
It thus follows that the action of the generator gives rise to the sum of at most
$2\,d^4(p+1)$ monomials consisting of the products of a local operator and $p+1$ mean-field operators multiplied by either the coefficients $\widetilde{A}_{\mu\nu}$ or 
$\widetilde{B}_{\mu\nu}$. With respect to the monomials contributed by $\HH^{(N)}$, they contain one additional term,
$$
\left\|\frac{V^{(N)}_\mu}{\sqrt{N}}\right\|\leq v\ .
$$ 
Since, by~\eqref{ONB}, $\|v_\mu\|\leq1$, it follows that $v=\max_\mu\{\|v_\mu\|\}\leq1$; thus, the norms of the monomials provided by $\HH^{(N)}$ can be bounded as those provided by $\widetilde{\DD}$. 
Therefore, one can estimate the norm of the action of $\LL^{(N)}$ by means of the norms of $d^2+2\,d^4$ monomials containing $(p+1)$ mean-field operators and a single local one. Furthermore, one sees that the monomials not containing commutators with the local operator $O$ are bounded by
$2\,v^2\,\|O\|\,x^p$ while those containing it by $2\,v^2\,\ell(O)\,\|O\|\,x^p$.

Iterating this argument, $\left(\LL^{(N)}\right)^h[P^{(N)}]$ will then contain at most 
$\displaystyle(d^2+2d^{4})^h\frac{(p+h)!}{p!}$ monomials, each one with a norm that can be upper bounded as if consisted of the product of $h+p$ mean-field quantities and strictly local operators all supported within $\cs(O)$ and thus by at most $\ell(O)$ sites.
Finally, the result follows since each of the coefficients multiplying the monomials is bounded from above by $2c$ and the worst case scenario is when all successive commutators act on local operators as each of them provides a factor $\ell(O)>1$.
\qed
\end{proof}
\bigskip

The previous Lemma can now be used to show that $\gamma^{(N)}_t$ maps mean-field quantities into infinite sums of products of mean-field quantities that converge in norm for all times $t$ in a certain time interval $[0,R]$.
\medskip

\begin{corollary}
Let $P^{(N)}$ be as in \eqref{specop}; then, $\displaystyle\sum_{k=0}^\infty\frac{t^k}{k!}(\LL^{(N)})^k[P^{(N)}]$ converges in norm to $\gamma^{(N)}_t[P^{(N)}]$ for $0\leq t<R=(2\,c\, v^2\,\ell(O)\,(d^2+2\,d^4))^{-1}$.
Consider
$f_N(z)=\omega\left(a\,{\rm e}^{z\,\mathbb{L}^{(N)}}[P^{(N)}]\,b\right)$, where $z\in\mathbb{C}$ and $a,b\in\ca$ are strictly local operators.
Then, for $|z|<R$, 
$$
\lim_{N\to\infty}f_N(z)=\sum_{k=0}^{\infty}\frac{z^k}{k!}\lim_{N\to\infty}\omega\left(a\,\left(\mathbb{L}^{(N)}\right)^k[P^{(N)}]\,b\right)\ .
$$
\label{exg}
\end{corollary}
\bigskip

\begin{proof}
Given the power series expansion 
$\displaystyle\sum_{k=0}^{\infty}\frac{z^k}{k!}(\mathbb{L}^{(N)})^k[P^{(N)}]$, $z\in\CC$, from Lemma \ref{LhN} it follows that
$$
\left\|\frac{z^k}{k!}\,(\mathbb{L}^{(N)})^k[P^{(N)}]\,\right\|\leq
\left(\frac{|z|}{R}\right)^k\frac{(k+p)!}{k!p!}\,\|O\|\,\|x\|^p\ .
$$
Since the bound is independent of $N$, the convergence is uniform in $N$ for all
$|z|<R$ and one can exchange the infinite sum with the large $N$ limit. 
\qed
\end{proof}
\bigskip

Using the previous corollary one can show that the dissipative dynamics of products of operators of the form $P^{(N)}$ factorizes in the large $N$ limit, despite the fact that, for each finite $N$, the time-evolution is not an automorphism of $\ca$. 
\medskip

\begin{corollary}
For all $0<t<R=(2\,c\, v^2\,\ell(O)\,(d^2+2\,d^4))^{-1}$ and operators $P^{(N)}$ and $Q^{(N)}$ of the form \eqref{specop},
$$
\left\|\gamma^{(N)}_t[P^{(N)}\,Q^{(N)}]\,-\,\gamma^{(N)}_t[P^{(N)}]\,\gamma^{(N)}_t[Q^{(N)}]\right\|\,=\,O\left(\frac{1}{N}\right)\ .
$$
\label{spt}
\end{corollary}
\bigskip

\begin{proof}
The norm of the difference we want to show to vanish in the large $N$ limit, can be recast as
$$
I=\left\|
\int_0^t{\rm d}s\ \frac{{\rm d}}{{\rm d}s}\Big(\gamma^{(N)}_s\Big[
\gamma^{(N)}_{t-s}\Big[P^{(N)}\Big]\,
\gamma^{(N)}_{t-s}\Big[Q^{(N)}\Big]\Big]\Big)
\right\|=\left\|\int_0^t{\rm d}s\,\gamma^{(N)}_s[D_N(t-s)]\right\|\ ,
$$
where, using~\eqref{Lind0},
\beann
D_N(t-s)&=&\LL^{(N)}\Big[\gamma^{(N)}_{t-s}\Big[P^{(N)}\Big]\,
\gamma^{(N)}_{t-s}\Big[Q^{(N)}\Big]\Big]\,-\,
\LL^{(N)}\Big[\gamma^{(N)}_{t-s}\Big[P^{(N)}\Big]\Big]\,
\gamma^{(N)}_{t-s}\Big[Q^{(N)}\Big]\\
&-&\gamma^{(N)}_{t-s}\Big[P^{(N)}\Big]\,\LL^{(N)}\Big[
\gamma^{(N)}_{t-s}\Big[Q^{(N)}\Big]\Big]\\
&=&
\sum_{\mu,\nu=1}^{d^2}C_{\mu\nu}\frac{1}{N}\left[\sum_{k=0}^{N-1}v_\mu^{(k)}\,,\,\gamma^{(N)}_{t-s}\left[P^{(N)}\right]\right]\,\left[\gamma^{(N)}_{t-s}\left[Q^{(N)}\right]\,,\,\sum_{h=0}^{N-1}v_\nu^{(h)}\right]\ .
\eeann

Since $\gamma_t^{(N)}$ is a contraction for any $N\geq 1$ (see Remark~\ref{remgen}.1), $\left\|\gamma_t^{(N)}[X]\right\|
\leq\|X\|$, whence one estimates
$I\leq\int_0^t{\rm d}s\,\left\|D_N(s)\right\|$, where
$$
\left\|D_N(s)\right\|\leq
\sum_{\mu,\nu=1}^{d^2}\,\Big|C_{\mu\nu}\Big|\,\frac{1}{N}\left\|\left[\sum_{k=0}^{N-1}v_\mu^{(k)}\,,\,\gamma^{(N)}_{s}\Big[P^{(N)}\Big]\right]\right\|\left\|\left[\sum_{h=0}^{N-1}v_\nu^{(h)}\,,\,\gamma^{(N)}_{s}\Big[Q^{(N)}\Big]\right]\right\|\ .
$$
The result then follows by showing that
$$
\lim_{N\to\infty}\left\|\left[\sum_{k=0}^{N-1}v_\mu^{(k)}\,,\,\gamma^{(N)}_s\left[P^{(N)}\right]\right]\right\|\leq \lim_{N\to\infty}\sum_{\ell=0}^{\infty}\frac{s^\ell}{\ell!}\left\|\left[\sum_{k=0}^{N-1}v_\mu^{(k)},{\mathbb{L}^\ell}_N[P^{(N)}]\right]\right\|<\infty\ .
$$
In order to prove it, one can use the argument of the proof of Lemma \ref{exg}: operators from different sites commute, hence commutators of $\displaystyle\sum_{k=0}^{N-1}v_\mu^{(k)}$ with mean-field operators yield mean-field operators, while commutators of $\displaystyle\sum_{k=0}^{N-1}v_\mu^{(k)}$ with local operators yield local operators with the same or a smaller support. Therefore, the radius of norm-convergence with respect to $s\geq 0$ of  
$$
\sum_{\ell=0}^{\infty}\frac{s^\ell}{\ell!}\left\|\left[\sum_{k=0}^{N-1}v_\mu^{(k)},\mathbb{L}^\ell_N[P^{(N)}]\right]\right\|
$$
can be estimated by the $R$ in the previous corollary.
\qed
\end{proof}
\bigskip

In order to proceed with the proof of Theorem~\ref{limdyn}, we first derive the time evolutions of the macroscopic averages introduced in Definition~\ref{defmacro}.
\medskip

\begin{proposition}
Let $\omega$ be a translation-invariant, clustering state on the quasi-local algebra $\ca$ and $\LL^{(N)}$ the dissipative Lindblad generator in \eqref{mflind1a} with $R$ as in Corollary \ref{exg}. Then, for $0\leq t<R$, the macroscopic averages in \eqref{macr1a}  evolve according to the set of non-linear equations
\beann
\nonumber
\frac{{\rm d}}{{\rm d}t}\omega_\alpha(t)&=&\sum_{\mu=1}^{d^2}\Big(\sum_{\nu=1}^{d^2}\widetilde{B}_{\mu\nu}\,\omega_\nu(t)\,+\,i\,\epsilon_\mu\Big)\,\omega_{\mu\alpha}(t)\\
&=&\sum_{\mu,\gamma=1}^{d^2}\Big(\sum_{\nu=1}^{d^2}\widetilde{B}_{\mu\nu}\,\omega_\nu(t)\,+\,i\,\epsilon_\mu\Big)\,J^\mu_{\alpha\gamma}\,\,\omega_\gamma(t)\ ,\quad\forall \alpha=1,2,\ldots d^2\ .
\eeann
\label{mfdyn} 
\end{proposition}
\medskip

\medskip

\begin{proof}
Consider the expression of the generator $\LL^{(N)}$ as given in~\eqref{mflind1a}, 
Corollary~\ref{exg} states that, for all $0\leq t<R$, the series in~\eqref{macr1a}, obtained by expanding $\gamma_t^{(N)}$, converges uniformly in $N$; one can then  exchange the large $N$ limit with the time-derivative obtaining:
$$
\frac{{\rm d}}{{\rm d}t}\omega_\alpha(t)=\lim_{N\to\infty}\omega\left(\gamma^{(N)}_t\left[\Big(\HH^{(N)}\,+\,\widetilde{\AA}^{(N)}\,+\,\widetilde{\BB}^{(N)}\Big)\Big[\frac{1}{N}\sum_{k=0}^{N-1}v_\alpha^{(k)}\Big]\right]\right)\ ,
$$
where 
\bea
\label{maph}
\HH^{(N)}\left[\frac{1}{N}\sum_{k=0}^{N-1}v_\alpha^{(k)}\right]
&=&i\,\sum_{\mu=1}^{d^2} \epsilon_\mu\,\frac{1}{N}\sum_{k=0}^{N-1}\,[v^{(k)}_\mu\,,\,v^{(k)}_\alpha]\\
\label{mapA}
\widetilde{\AA}^{(N)}\left[\frac{1}{N}\sum_{k=0}^{N-1}v_\alpha^{(k)}\right]&=&
\frac{1}{2N^2}\sum_{\mu,\nu=1}^{d^2}\sum_{k=0}^{N-1}\widetilde{A}_{\mu\nu}\,\left[\left[v^{(k)}_\mu\,,\,v^{(k)}_\alpha\right]\,,\,v^{(k)}_\nu\right]\\
\label{mapB}
\widetilde{\BB}^{(N)}\left[\frac{1}{N}\sum_{k=0}^{N-1}v_\alpha^{(k)}\right]&=&
\frac{1}{2N^2}\sum_{\mu,\nu=1}^{d^2}\sum_{k,\ell=0}^{N-1}\widetilde{B}_{\mu\nu}\left\{\left[v^{(k)}_\mu\,,\,v^{(k)}_\alpha\right]\,,\,v^{(\ell)}_\nu\right\}\ .
\eea
Using~\eqref{macr1b} and~\eqref{aid1}, the large $N$ limit of the hamiltonian contribution yields
$$
\lim_{N\to\infty}\omega\left(\gamma^{(N)}_t\left[\HH^{(N)}\Big[\frac{1}{N}\sum_{k=0}^{N-1}v_\alpha^{(k)}\Big]\right]\right)=i\,\sum_{\mu=1}^{d^2}\epsilon_\mu\omega_{\mu\alpha}(t)
=i\,\sum_{\mu,\beta=1}^{d^2}\epsilon_\mu\,J^\mu_{\alpha\beta}\,\omega_{\beta}(t)\ .
$$
Concerning the dissipative contribution, since $\gamma^{(N)}_t$ is a contraction one has
$$
\left|\omega\left(\gamma^{(N)}_t\left[\widetilde{\AA}^{(N)}\left[\frac{1}{N}\sum_{k=0}^{N-1}v_\alpha^{(k)}\right]\right]\right)\right|\leq \left\|\widetilde{\AA}^{(N)}\left[\frac{1}{N}\sum_{k=0}^{N-1}v_\alpha^{(k)}\right]\right\|\leq\frac{2\,v^3}{N}\,\sum_{\mu,\nu=1}^{d^2}|\,\widetilde{A}_{\mu\nu}|\,\ ,
$$
whence $\widetilde{\AA}^{(N)}$ does not contribute. 

On the other hand, using \eqref{mapB} and Corollary \ref{spt}, it follows that, in norm,
\beann
&&
\hskip-.5cm
\gamma^{(N)}_t\left[\widetilde{\BB}^{(N)}\left[\frac{1}{N}\sum_{k=0}^{N-1}v_\alpha^{(k)}\right]\right]
=
\frac{1}{2}\sum_{\mu,\nu=1}^{d^2}\widetilde{B}_{\mu\nu}\,\gamma^{(N)}_t\left[\left\{\frac{1}{N}\sum_{k=0}^{N-1}[v^{(k)}_\mu\,,\,v^{(k)}_\alpha]\,,\,\frac{1}{N}\sum_{\ell=0}^{N-1}v_\nu^{(\ell)}\right\}\right]\\
&&
\hskip.5cm
=\frac{1}{2}\sum_{\mu,\nu=1}^{d^2}\widetilde{B}_{\mu\nu}\,\left\{\gamma^{(N)}_t\left[\frac{1}{N}\sum_{k=0}^{N-1}[v^{(k)}_\mu\,,\,v^{(k)}_\alpha]\right]\,,\,\gamma^{(N)}_t\left[\frac{1}{N}\sum_{\ell=0}^{N-1}v_\nu^{(\ell)}\right]\right\}\,+\,O\left(\frac{1}{N}\right)\ .
\eeann
From Corollary~\ref{exg} one knows that, for $0\leq t<R$, mean-field operators are turned  into norm-convergent series of mean-field operators; moreover, these latter behave as stated in~\eqref{macro1} in the large $N$ limit. Then, using Corollary~\ref{spt} together with~\eqref{macr1a} and~\eqref{macr1b} one obtains
\begin{eqnarray*}
&&\hskip-2cm
\lim_{N\to\infty}\omega\left(\gamma^{(N)}_t\left[\widetilde{\BB}^{(N)}\left[\frac{1}{N}\sum_{k=0}^{N-1}v_\alpha^{(k)}\right]\right]\right)=\\
&&\hskip-1cm
=\sum_{\mu,\nu=1}^{d^2}\widetilde{B}_{\mu\nu}\lim_{N\to\infty}\omega\left(\gamma^{(N)}_t\left[\frac{1}{N}\sum_{k=0}^{N-1}[v^{(k)}_\mu\,,\,v^{(k)}_\alpha]\right]\,\gamma^{(N)}_t\left[\frac{1}{N}\sum_{\ell=0}^{N-1}v_\nu^{(\ell)}\right]\right)\\ 
&&\hskip-1cm
=\sum_{\mu,\nu=1}^{d^2}\widetilde{B}_{\mu\nu}\lim_{N\to\infty}\omega\left(\gamma^{(N)}_t\left[\sum_{k=0}^{N-1}\frac{([v_\mu,v_\alpha])^{(k)}}{N}\right]\right)\,\omega\left(\gamma^{(N)}_t\left[\sum_{\ell=0}^{N-1}\frac{v_\nu^{(\ell)}}{N}\right]\right)\\
&&\hskip-1cm
=\sum_{\mu,\nu=1}^{d^2}\widetilde{B}_{\mu\nu}\,\omega_{\mu\alpha}(t)\,\omega_\nu(t)=
\sum_{\mu,\beta,\nu=1}^{d^2}\widetilde{B}_{\mu\nu}\,\omega_\nu(t)\,J^\mu_{\alpha\beta}\,\omega_\beta(t)\ .
\label{estimate}
\end{eqnarray*}
\qed
\end{proof}
\bigskip

By means of the time-evolution of macroscopic averages, we move on to prove Theorem \ref{limdyn}: we first show that the result holds for times $0\leq t\leq R$, $R$ as in Corollary~\ref{exg}, and for strictly local operators and then relax these two constraints.
\bigskip

\noindent
\textbf{Theorem 2.}\quad
\textit{
Let the quasi-local algebra $\ca$ be equipped with a translation-invariant, clustering state $\omega$. In the large $N$ limit, the local dissipative generators $\LL^{(N)}$ in \eqref{mflind1a} define on $\ca$ a one-parameter family of automorphisms that depend on the state $\omega$ and are such that, for all $0\leq t\leq T$, $T\geq 0$ arbitrary, 
$$
\lim_{N\to\infty}\omega\left(a\,\gamma^{(N)}_t[O]\,b\right)=\omega\left(a\,\alpha_t[O]\,b\right)\ ,
$$
for all $a,b\in\ca$ and $O\in\ca$. If $O$ has finite support $\cs(O)\subset[0,S-1]$, then
\be
\label{auto2a}  
\alpha_t[O]=\big(U^{(S)}_t\big)^\dag\,O\,U^{(S)}_t\ ,\quad
U^{(S)}_t=\mathbb{T}{\rm e}^{-i\int_0^t{\rm d}s\, H^{(S)}_s}\ ,
\ee
with explicitly time-dependent hamiltonian
\be
\label{auto3aa}
H^{(S)}_t=-i\,\sum_{\mu}^{d^2}\sum_{k=0}^{S-1}\,\Big(\sum_{\nu=1}^{d^2}\widetilde{B}_{\mu\nu}\,\omega_\nu(t)\,+\,i\,\epsilon_\mu\Big)\,v_\mu^{(k)} \ ,
\ee
where 
$$
\mathbb{T}{\rm e}^{-i\int_0^t{\rm d}s\, H^{(S)}_s}=
\mathbf{1}+\sum_{k=1}^\infty
(-i)^k\int_0^t{\rm d}s_1\cdots\int_0^{s_{k-1}}{\rm d}s_k\,H^{(S)}_{s_1}\cdots H^{(S)}_{s_k}\ .
$$}
\medskip

\begin{proof}
Given $a,b\in\ca$ and $O$ strictly local with fixed finite support $\cs(O)=[0,S-1]$,
we consider the hamiltonian $H^{(S)}_t$ as in~\eqref{auto3aa}, set 
$$
O_t:=\left(U_t^{(S)}\right)^\dag\,O\,U^{(S)}_t\ ,\
D^{(N)}(t,s):=\frac{{\rm d}}{{\rm d}s}\gamma^{(N)}_s\left[O_{t-s}\right]
$$
and study the large $N$ limit of
\be
\label{Th2a}
I^{(N)}(t)=\omega\left(a\,\left(\gamma^{(N)}_t[O]\,-\,O_t\right)\,b\right)\\ 
=\int_0^t{\rm d}s\,\omega\left(a\,D^{(N)}(t,s)\,b\right)\ .
\ee
One finds
\bea
\nonumber
&&\hskip-1cm
D^{(N)}(t,s)=
\gamma^{(N)}_s\Bigg[\LL^{(N)}\left[O_{t-s}\right]-i\left[H^{(S)}_s\,,\,O_{t-s}\right]\Bigg]\\
\label{p30}
&&
=\gamma^{(N)}_s\left[\Big(\widetilde{\AA}^{(N)}\,+\,\widetilde{\BB}^{(N)}\Big)\left[O_{t-s}\right]\right]\,-\,\gamma^{(N)}_s\left[\sum_{k=0}^{S-1}\sum_{\mu,\nu=1}^{d^2}\widetilde{B}_{\mu\nu}\,\omega_\nu(s)\,v^{(k)}_\mu\,,\,O_{t-s}\right]\ .
\eea
Since $H^{(S)}_t$ is the sum of single-site contributions, 
$O_{t-s}$ is a strictly local operator with the same support as $O$.
Thus, as in the proof of the previous proposition, the action of $\widetilde{\AA}^{(N)}$ of the generator $\LL^{(N)}$ (see \eqref{submap1}) is such that
$\lim_{N\to\infty}\left\|\widetilde{\AA}^{(N)}\left[O_{t-s}\right]\right\|=0$. 
Instead, the contribution $\widetilde{\BB}^{(N)}$ in~\eqref{submap2} yields
$$
\widetilde{\BB}^{(N)}\left[O_{t-s}\right]=\sum_{\mu,\nu=1}^{d^2}\frac{\widetilde{B}_{\mu\nu}}{2}\left\{X^{(S)}_\mu(t-s)\,,\,\frac{1}{N}\sum_{k=0}^Nv^{(k)}_\nu\right\}\ ,\ X^{(S)}_\mu(t-s):=\left[\sum_{\ell=0}^{N-1}v^{(\ell)}_\mu\,,\,O_{t-s}\,\right]\ ,
$$
with $X^{(S)}_\mu(t-s)$ a strictly local operator with support
fixed by $O$. Then,
\beann
&&
\lim_{N\to\infty}\left\|\left\{X_\mu^{(S)}(t,s)\,,\,\frac{1}{N}\sum_{k=0}^Nv^{(k)}_\nu\right\}-2\,X_\mu^{(S)}(t,s)\,\frac{1}{N}\sum_{k=0}^Nv^{(k)}_\nu\right\|=0\qquad\hbox{yields}\\
&&
\lim_{N\to\infty}\left\|D^{(N)}(t,s)\right\|=\lim_{N\to\infty}\left\|\gamma^{(N)}_s\left[\sum_{\mu,\nu=1}^{d^2}\,B_{\mu\nu}\,X_\mu^{(S)}(t-s)\,Z_\nu^{(N)}(s)\right]\right\|\ ,\ \hbox{where}\\
&&
Z_\nu^{(N)}(s):=\frac{1}{N}\sum_{k=0}^{N-1}\Big(v_\nu^{(k)}\,-\,\omega_\nu(s)\Big)\ .
\eeann
Therefore, one can focus upon the limit
\beann
&&
\lim_{N\to\infty}\left|I^{(N)}(t)\right|\le\sum_{\mu,\nu=1}^{d^2}\left|B_{\mu\nu}\right|\int_0^t{\rm d}s\, 
\lim_{N\to\infty}\,\Big|I_{\mu\nu}^{(N)}(t,s)\Big|\qquad\hbox{where}\\
&&
I_{\mu\nu}^{(N)}(t,s):= \omega\left(a\,\gamma^{(N)}_s\left[Z^{(N)}_\nu(s)\, X^{(N)}_\mu(t,s)\right]\,b\right)\ .
\eeann
Using the Cauchy-Schwarz inequality and the Kadison inequality for completely positive maps $\gamma$, $\gamma[x^\dag x]\geq (\gamma[x])^\dag\,\gamma[x]$, we have:
$$
\Big|I^{(N)}_{\mu\nu}(t,s)\Big|\le\sqrt{\omega\left(a a^\dagger\right)}\,\Big\|X^{(S)}_\mu(t,s)\Big\|\,\sqrt{\omega\left(b^\dag\,
\gamma^{(N)}_s\left[\left(Z^{(N)}_\nu(s)\right)^2\right]\,b\right)} \ .
$$
Both $X_\mu^{(S)}(t,s)$ and $Z_\nu^{(N)}(s)$ have norms independent of $N$; moreover,  
$$
\lim_{N\to\infty}\omega\left(\gamma^{(N)}_s\left[Z^{(N)}_\nu(s)\right]\right)=0
$$ 
because of \eqref{macr1a} and of the fact that $\gamma^{(N)}_t[\mathbf{1}]=\mathbf{1}$. Furthermore, $Z_\nu^{(N)}(s)$ is a mean-field quantity, whence $\lim_{N\to\infty}I^{(N)}(t)=0$ follows from Corollary~\ref{spt} which yields
$$
\lim_{N\to\infty}\omega\left(a\,\gamma^{(N)}_s\left[\left(Z^{(N)}_\nu(s)\right)^2\right]\,a^\dagger\right)=\omega\left(a\,a^\dagger\right)\,
\left(\lim_{N\to\infty}\omega\left(\gamma^{(N)}_s[Z^{(N)}_\nu(s)]\right)\right)^2=0\ .
$$
The result just obtained is valid for $0\leq t< R$ and for strictly local operators $O$.
It can be extended to all times in compact subsets of the positive real line and to the whole  quasi-local algebra $\ca$.
While the norm-preserving maps $\alpha_t$, $0\leq t<R$, can be extended by continuity to the quasi-local algebra, the extension to any finite time $t\geq0$ is obtained by the following Proposition~\ref{extdyn}.
\qed

\end{proof}
\bigskip

The first extension regards the time domain and makes use of the following result
\cite{Normal Families}.

\medskip

\begin{theorem}{\bf (Vitali-Porter)}\quad
Let $\Omega$ be a connected open subset of the complex plane  $\mathbb{C}$ and
$D(z_0,r)=\left\{z:|z-z_0|<r,r>0\right\}$ an open disk about $z_0\in\mathbb{C}$.

Let $\left\{f_N(z)\right\}$ be a locally bounded sequence of analytic functions on $\Omega$, namely such that for all $z_0\in\Omega$ there is a positive number $M=M(z_0)$ and a neighbourhood $D(z_0,r)\subseteq \Omega$ such that $\left|f_N(z)\right|\le M$, for all $z\in D(z_0,r)$ and all $f_N$.

If $\lim_{N\to\infty}f_N(z)$ exists for all $z$ in a subset $E\subseteq\Omega$  which contains at least one accumulation point in $\Omega$, then $\lim_{N\to\infty}f_N(z)=f(z)$ 
uniformly on any compact subset of $\Omega$ where $f$ is then analytic.
\end{theorem}
\medskip

\begin{proposition}
The convergence of the microscopic dynamics $\gamma^{(N)}_t$ to an automorphism $\alpha^\omega_t$ on strictly local operators $O\in\ca$ as established in Theorem \ref{limdyn} holds for all times $t\in[0,T]$ for any fixed $T\geq 0$.
\label{extdyn}
\end{proposition}

\begin{proof}
Given the connected open subset 
$\Omega=\left\{z=t+iy\,:\,t>0,\ |y|<R\right\}\subset\CC$, consider the sequence of complex functions $f_N(z)=\omega\left(a\,{\rm e}^{z\LL^{(N)}}[O]\,b\right)$ defined in Corollary \ref{exg}.
These are analytic functions and locally bounded on $\Omega$ for $\gamma^{(N)}$ is a contraction,
$$
\Big|f_N(z)\Big|=\Big|\omega\left(a\,{\rm e}^{t\LL^{(N)}}\circ {\rm e}^{iy\LL^{(N)}}[O]\,b\right)\Big|\le \|a\|\,\|b\|\,\left\|{\rm e}^{iy\LL^{(N)}}[O]\right\|\ .
$$
The last norm is bounded uniformly in $N$ for $|y|<R$; this follows by applying Corollary \ref{exg}, which also shows that $\lim_{N\to\infty} f_N(z)$ exists for all $z\in E=\left\{z=t+iy\,:\,t>0\,,\,|z|<R\right\}$. Then, the Vitali-Porter theorem ensures that $\lim_{N\to\infty}f_N(z)=f(z)$ with $f(z)$ an analytic function, uniformly on any compact subset of $\Omega$.

With $\mathcal{S}(O)=[0,S-1]$ the support of $O$, let us consider the time-evolutor $U^{(S)}_t$ in \eqref{auto2a} and complexify the time-dependence sending $H^{(S)}_s$ in~\eqref{auto3aa}  into $H^{(S)}_{z(s)}$, where
$z(s)=s(1+iy/t)$ so that $z(t)=t+iy$, and consider
$\displaystyle
U^{(S)}_{-}(z(t))=\mathbb{T}\exp\left(-i\int_0^t{\rm d}s\,H^{(S)}(z(s))\right)
$,
together with the inverted time-ordered exponential
$\displaystyle
U^{(S)}_+(z(t))=\mathbb{T}^{-1}\exp\left(i\int_0^t{\rm d}s\,H^{(S)}(z(s))\right)$:
they satisfy $U^{(S)}_+(z(t))\,U^{(S)}_-(z(t))=\mathbf{1}$, while, if 
$z(s)=s$ for all $s\in[0,t]$, then
$$
U^{(S)}_-(z(t))=U^{(S)}_t\ ,\qquad U^{(S)}_+(z(t))=\big(U^{(S)}_t\big)^\dagger\ .
$$
Since the hamiltonians in \eqref{auto3} are sums of single-site operators that do not modify the support of the time-evolving strictly local operator $O$, the functions
$$
u_S(z(t))=\omega\left(a\,U^{(S)}_+(z(t))\,O\,U^{(S)}_-(z(t))\,b\right)\qquad a,b\in\ca\ ,
$$
are also analytic on $\Omega$; indeed, they are bounded:
$$
\left|u_S(z(t))\right|\le\|a\|\,\|b\|\,\|O\|\,\left\|U^{(S)}_+(z(t))\right\|\,\left\|U^{(S)}_-(z(t))\right\|\leq \|a\|\,\|b\|\,\|O\|\,{\rm e}^{2t\,H_{max}}\ ,
$$ 
where $H_{max}:=\max_{0\leq s\leq t}\|H^{(S)}(z(s))\|$.
Consider now $T>R$ and the subset 
$$
\Omega_{T}=\left\{z=t+iy\,:\,0<t<T,\ |y|<R\right\}\ .
$$
We have that $f(z)$ and $u_S(z)$ are both analytic functions on $\Omega_T$. Moreover, due to Theorem \ref{limdyn}, $f(z(t))=u_{S}(z(t))$ for $z=t\in[0,R)$; therefore, $f(z(t))=u_{S}(z(t))$ for all $z(t)\in \Omega_T$, so that the restriction to the real line yields the result.
\qed
\end{proof}
\medskip

We can now conclude by extending the previous results from strictly local
operators $O$ to mean-field operators.
\medskip

\begin{corollary}
The convergence of the microscopic dynamics $\gamma^{(N)}_t$ to the automorphisms $\alpha_t$ on the quasi-local algebra $\ca$ as in Theorem \ref{limdyn} holds for operators arising as strong limits of mean-field operators.
\label{quasilocal}
\end{corollary}

\begin{proof}
Consider the mean-field operator $\displaystyle X^{(N)}=\frac{1}{N}\sum_{k=0}^{N-1}x^{(k)}$. For $t\in[0,R)$ the dynamics implemented by $U^{(N)}_t$ satisfies, in the large $N$ limit, the equation of motion of
Proposition \ref{mfdyn}. Furthermore, with the notations of the previous proposition,
\beann
\lim_{N\to\infty}\omega\left(U^{(N)}_+(z(t))\,X^{(N)}\,U^{(N)}_-(z(t))\right)&=&\lim_{N\to\infty}\frac{1}{N}\sum_{k=0}^{N-1}\omega\left(U^{(N)}_+(z(t))\,x^{(k)}\,U^{(N)}_-(z(t))\right)\\
&\le&\left\|U^{(N)}_+(z(t))\,x\,U^{(N)}_-(z(t))\right\|\ .
\eeann
with $x\in\ca$ strictly local. Then, as  in Proposition \ref{extdyn}, 
$\lim_{N\to\infty}\omega\left(U^{(N)}_+(z(t))\,X^{(N)}\,U^{(N)}_-(z(t))\right)$ provides an analytic function on compact subsets of $\Omega_{T}=\left\{z=t+iy\,:\,0<t<T,\ |y|<R\right\}$, $T\geq R$, and its restriction to $t\in[0,T)$ implements the large $N$ dynamics induced by the generator $\LL^{(N)}$.
\qed
\end{proof}

\subsection{Dynamics of quantum fluctuations}
\label{app2}

This section will be devoted to the proofs of the results concerning the structure and properties of the generator of the dissipative dynamics of quantum fluctuations.
We start with the proof of Theorem \ref{expfluct} which is divided into several steps, the first ones concerning the algebraic behaviour of quantum fluctuations, mean-field quantities and local exponentials in the large $N$ limit. 
\medskip

\begin{lemma}
\label{lemtool1}
For all $\vec{r}_{1,2}\in\RR^{d^2}$, it holds that
$$
\lim_{N\to\infty}\left\|\left(W^{(N)}_t(\vec{r}_2)\right)^\dag\,\Big(\vec{r}_1\cdot\vec{F}^{(N)}_t\Big)\,W^{(N)}_t(\vec{r}_2)-\vec{r}_1\cdot\vec{F}^{(N)}_t\,+\,i\,\vec{r}_2\cdot\left(T^{(N)}\,\vec{r}_1\right)\right\|=0\ ,
$$
where $T^{(N)}=[T^{(N)}_{\mu\nu}]$ is the mean-field operator-valued matrix with entries \eqref{tcomm}.
\end{lemma}

\begin{proof}
Using 
\be
\label{aux1}
{\rm e}^x\,y\,{\rm e}^{-x}=\sum_{n=0}\frac{1}{n!}\KK^n_x[y]\ ,\quad
\KK^n_x[y]=\left[x\,,\,\KK^{n-1}_x[y]\right]\ ,\quad \KK^0_x[y]=y\ ,\quad\forall\,x\,,\,y\in\ca\ ,
\ee
by means of \eqref{tcomm} we write 
\beann
\label{aux1a}
&&
\left(W^{(N)}_t(\vec{r}_2)\right)^\dag\,\Big(\vec{r}_1\cdot\vec{F}^{(N)}_t\Big)\,W^{(N)}_t(\vec{r}_2)=\vec{r}_1\cdot\vec{F}^{(N)}_t\,-\,i\left[\vec{r}_2\cdot\vec{F}^{(N)}_t\,,\,\vec{r}_1\cdot\vec{F}^{(N)}_t\right]\,+\,Z^{(N)}(t)\\ 
&&\hskip 2cm
=\vec{r}_1\cdot\vec{F}^{(N)}_t\,-\,i\vec{r}_2\cdot\left(T^{(N)}\,\vec{r}_1\right)
\,+\,Z^{(N)}(t)\ .
\label{aux1b}
\eeann
In order to deal with
$\displaystyle
Z^{(N)}(t)=\sum_{n=2}^\infty\frac{(-i)^n}{n!} \KK^n_{\vec{r}_2\cdot\vec{F}^{(N)}_t}[\vec{r}_1\cdot\vec{F}^{(N)}_t]$,
notice that, since operators at different lattice sites commute,
$$
\label{aux2}
\KK^n_{\vec{r}_2\cdot\vec{F}^{(N)}_t}[\vec{r}_1\cdot\vec{F}^{(N)}_t]=
\sum_{k=0}^{N-1}\sum_{\nu,\mu_1,\ldots,\mu_n=1}^{d^2}
\frac{r_{1\nu}r_{2\mu_1}\cdots r_{2\mu_n}}{\sqrt{N^{n+1}}}\left[v^{(k)}_{\mu_1}\,,\,\left[v^{(k)}_{\mu_2}\,,\,\cdots\,\left[v^{(k)}_{\mu_n}\,,\,v^{(k)}_\nu\right]\right]\cdots\right]\ ,
$$
whence, with $\xi=\max_{j}\{|r_{1j}|\,,\,|r_{2j}|\}$ and $v=\max_{\mu}\|v_\mu\|$,
$$
\left\|\KK^n_{\vec{r}_2\cdot\vec{F}^{(N)}_t}[\vec{r}_1\cdot\vec{F}^{(N)}_t]\right\|\leq
\frac{v\,\xi\,d^2(2\,\xi\,v\,d^2)^n}{\sqrt{N^{n-1}}}\Rightarrow\|Z^{(N)}(t)\|\leq\frac{v\,\xi\,d^2}{\sqrt{N}}
\,{\rm e}^{2v\xi d^2}\ ,
$$
so that $\lim_{N\to\infty}\|Z^{(N)}(t)\|=0$ and the result follows.
\qed
\end{proof}
\medskip

\begin{lemma}
\label{lemtool2}
In the large $N$ limit any mean-field quantity $\displaystyle X^{(N)}=\frac{1}{N}\sum_{k=0}^{N-1}x^{(k)}$, $x\in M_d(\CC)$, commutes with the local exponential operators in the sense that
$$
\lim_{N\to\infty}\left\|\left[X^{(N)}\,,\,W_t^{(N)}(\vec{r})\right]\right\|=0\ .
$$
\end{lemma}

\begin{proof}
Using \eqref{aux1} one writes
\bea
\nonumber
\left[X^{(N)}\,,\,W_t^{(N)}(\vec{r})\right]&=&\left(X^{(N)}\,-\,W_t^{(N)}(\vec{r})\,X^{(N)}\,\left(W_t^{(N)}(\vec{r})\right)^\dag\right)\,W_t^{(N)}(\vec{r})\\
\nonumber
&=&-\left(\sum_{n=1}^\infty \KK^n_{\vec{r}\cdot\vec{F}^{(N)}_t}[X^{(N)}]\right)\,W_t^{(N)}(\vec{r})\qquad\hbox{with}\\
\label{auxilium}
\KK^n_{\vec{r}\cdot\vec{F}^{(N)}_t}[X^{(N)}]&=&\sum_{k=0}^{N-1}\sum_{\mu_1,\ldots,\mu_n=1}^{d^2}
\frac{r_{\mu_1}\cdots r_{\mu_n}}{N\sqrt{N^n}}\left[v^{(k)}_{\mu_1}\,,\,\left[v^{(k)}_{\mu_2}\,,\,\cdots\,\left[v^{(k)}_{\mu_n}\,,\,x^{(k)}\right]\right]\cdots\right]\ .
\eea
Then, as in the proof of the previous lemma, the result follows from
$$
\left\|\left[X^{(N)}\,,\,W_t^{(N)}(\vec{r})\right]\right\|\leq
\frac{\|x\|}{\sqrt{N}}\,{\rm e}^{2\,v\,\xi\,d^2}\ .
$$
\qed
\end{proof}
\medskip

The following Proposition specifies the speed with which the limit established in Proposition~\ref{mfdyn} is attained, a result which will be applied in the  coming estimates.
\medskip

\begin{proposition}
\label{scal}
With $\omega^{(N)}_\alpha(t)$ and $\omega^{(N)}_{\mu\alpha}(t)$ defined in \eqref{macr1a}
and \eqref{macr1b} one has that:
\bea
\label{aidd00}
&&
\left|\frac{{\rm d}}{{\rm d}t}\omega^{(N)}_\alpha(t)\,-\,\sum_{\beta=1}^{d^2}\Big(\widetilde{D}^{(N)}_{\alpha\beta}(t)\,+\,i\,\mathcal{E}_{\alpha\beta}\Big)\,\omega^{(N)}_{\beta}(t)\,\right|=O\left(\frac{1}{N}\right)\\ 
\label{Daid}
&&
\widetilde{D}^{(N)}_{\alpha\beta}(t)=\sum_{\beta,\mu,\nu=1}^{d^2}\widetilde{B}_{\mu\nu}\,J^\mu_{\alpha\beta}\,\omega^{(N)}_\nu(t)\ ,
\eea
with $\mathcal{E}_{\alpha\beta}$ the entries of the matrix $\mathcal{E}$ defined in~\eqref{D}.
\end{proposition}
\medskip

\begin{proof}
Consider the time-derivative of $\omega^{(N)}_\alpha(t)$ in~\eqref{macr1a}
\beann
\frac{{\rm d}}{{\rm d}t}\omega^{(N)}_\alpha(t)&=&\omega\left(\gamma_t^{(N)}\left[\HH^{(N)}\left[\frac{1}{N}\sum_{k=0}^{N-1}v_\alpha^{(k)}\right]\right]\right)\,+\,\omega\left(\gamma_t^{(N)}\left[\widetilde{\AA}^{(N)}\left[\frac{1}{N}\sum_{k=0}^{N-1}v_\alpha^{(k)}\right]\right]\right)\\
&+&\frac{1}{2}\sum_{\mu,\nu=1}^{d^2}\widetilde{B}_{\mu\nu}\,\omega\left(\gamma_t^{(N)}\left[\left\{\frac{1}{N}\sum_{k=0}^{N-1}\left[v_\mu^{(k)},v_\alpha^{(k)}\right],\frac{1}{N}\sum_{h=0}^{N-1}v_\nu^{(h)}\right\}\right]\right)\, .
\eeann
Since, using~\eqref{macr1a} and~\eqref{aid1},
$$
\omega\left(\gamma_t^{(N)}\left[\HH^{(N)}\left[\frac{1}{N}\sum_{k=0}^{N-1}v_\alpha^{(k)}\right]\right]\right)=i\sum_{\mu=1}^{d^2}\epsilon_\mu\omega^{(N)}_{\mu\alpha}(t)\ ,
$$
and, as already seen in the proof of Proposition~\ref{mfdyn}, the action of the $\widetilde{\AA}^{(N)}$ term of the generator on mean-field observables is in norm a $O\left(\frac{1}{N}\right)$ quantity, one has 
\begin{equation*}
\begin{split}
\left|\frac{{\rm d}}{{\rm d}t}\omega^{(N)}_\alpha(t)\,-\,\sum_{\mu=1}^{d^2}\Big(\sum_{\nu=1}^{d^2}\widetilde{B}_{\mu\nu}\,\omega^{(N)}_\nu(t)\,+\,i\,\epsilon_\mu\Big)\omega^{(N)}_{\mu\alpha}(t)\,\right|\le O\left(\frac{1}{N}\right)+\hspace{4cm}\\
 +\frac{1}{2}\left|\sum_{\mu,\nu=1}^{d^2}\widetilde{B}_{\mu\nu}\,\left(\omega\left(\gamma_t^{(N)}\left[\left\{\frac{1}{N}\sum_{k=0}^{N-1}\left[v_\mu^{(k)},v_\alpha^{(k)}\right],\frac{1}{N}\sum_{h=0}^{N-1}v_\nu^{(k)}\right\}\right]\right)-2\omega_{\mu\alpha}^{(N)}(t)\omega_{\nu}^{(N)}(t)\right)\right|\, .
\end{split}
\end{equation*}
Using~\eqref{macr1a},~\eqref{aid1},~\eqref{fluct} and the fact that fluctuation have zero mean values, one rewrites
\beann
&&\omega\left(\gamma_t^{(N)}\left[\left\{\frac{1}{N}\sum_{k=0}^{N-1}\left[v_\mu^{(k)},v_\alpha^{(k)}\right],\frac{1}{N}\sum_{h=0}^{N-1}v_\nu^{(k)}\right\}\right]\right)-2\omega_{\mu\alpha}^{(N)}(t)\omega_{\nu}^{(N)}(t)=\\
&&\hskip2cm
=\frac{1}{N}\sum_{\gamma=1}^{d^2}\,J^\gamma_{\mu\alpha}\,
\omega\Big(\gamma^{(N)}_t\Big[\Big\{F^{(N)}_\gamma(t)\,,\,F^{(N)}_\nu(t)\Big\}\Big]\Big)\ .
\eeann
The required scaling results from Lemma~\ref{bound} in Appendix B and the fact that the Cauchy-Schwartz inequality relative to the expectations with respect to the state 
$\omega\circ\gamma^{(N)}_t$ yields
$$
\left|\omega\left(\gamma_t^{(N)}\left[F_\gamma^{(N)}(t)F_{\nu}^{(N)}(t)\right]\right)\right|^2\leq
\omega\left(\gamma_t^{(N)}\left[\left(F_\gamma^{(N)}(t)\right)^2\right]\right)\,\omega\left(\gamma_t^{(N)}\left[\left(F_\nu^{(N)}(t)\right)^2\right]\right)
\ .
$$
\qed
\end{proof}
\medskip

The following proposition establishes the asymptotic form of the action of the generator $\LL^{(N)}$ on local exponential.
\bigskip

\begin{proposition}
Given the local exponentials $W^{(N)}_t(\vec{r})=\exp\left(i\,\vec{r}\cdot\vec{F}^{(N)}_t\right)$, we have that:
\beann
&&
\lim_{N\to\infty}\omega_{\vec{r}_1\vec{r}_2}\left(\gamma^{(N)}_t\left[\LL^{(N)}\left[W^{(N)}_t(\vec{r})\right]\Big)\,W^{(N)}_t(\vec{r})\right]\right)=\\
&&
=\lim_{N\to\infty}\omega_{\vec{r}_1\vec{r}_2}\Bigg(\gamma^{(N)}_t\left[\Bigg(i\,\sqrt{N}\vec{r}\cdot\Big(\widetilde{D}^{(N)}_t\vec{\omega}_t^{(N)}\Big)\,-\,i\,\vec{r}\cdot\left(T^{(N)}\,\widetilde{B}\,-\,\widetilde{D}^{(N)}_t\right)\vec{F}^{(N)}_t\right.\\
&&\hskip3cm
\left.
-\frac{1}{2}\vec{r}\cdot\Big(T^{(N)}\,\Big(\widetilde{A}\,+\,2\,i\,h^{(re)}\Big)\,T^{(N)}\,-\,\widetilde{D}^{(N)}_t\,T^{(N)}\Big)\vec{r}
\,\Bigg)W^{(N)}_t(\vec{r})\right.\Bigg]\Bigg)\ ,
\eeann
where $\vec{\omega}^{(N)}_t\in\RR^{d^2}$ is the vector with components $\omega^{(N)}_\mu(t)$
in \eqref{macr1a}, $\vec{\epsilon}$ is the real vector with component given by the coefficients of the free hamiltonian $h^{(N)}$ in~\eqref{2Ham}, $T^{(N)}=[T^{(N)}_{\mu\nu}]$ is the operator-valued matrix with entries \eqref{tcomm}, while $\widetilde{D}^{(N)}_t$ is the real $d^2\times d^2$ matrix whose entries are defined in~\eqref{D}. 
\label{LINDACT}
\end{proposition}
\medskip

The proof of the above Proposition follows by grouping together the results of the next three Lemmas.
\medskip

\begin{lemma}
\label{lemmaux0}
The action of the $\HH^{(N)}$ component of $\LL^{(N)}$ is such that
\beann
&&\hskip-.5cm
\lim_{N\to\infty}\omega_{\vec{r}_1\vec{r}_2}\left(\gamma^{(N)}_t\left[\HH^{(N)}\left[W^{(N)}_t(\vec{r})\right]\Big)\,W^{(N)}_t(\vec{r})\right]\right)=\\
&&\hskip-.3cm
=\lim_{N\to\infty}\omega_{\vec{r}_1\vec{r}_2}\Bigg(\gamma^{(N)}_t\left[\Big(-\sqrt{N}\vec{r}\cdot\Big(\mathcal{E}\,\vec{\omega}^{(N)}_t\Big)\,-\,
\vec{r}\cdot\Big(\mathcal{E}\,\vec{F}^{(N)}_t\Big)+\frac{i}{2}\vec{r}\cdot\Big(\mathcal{E}\,T^{(N)}\,\vec{r}\Big)\Big)\,W^{(N)}_t(\vec{r})\right]\Bigg)\ ,
\eeann
where $\mathcal{E}$ is the $d^2\times d^2$ matrix defined in~\eqref{D}.
\end{lemma}
\medskip

\begin{proof}
Using \eqref{aux1}, one splits the action 
$$
\HH^{(N)}\left[W^{(N)}_t(\vec{r})\right]=i\sum_{\mu=1}^{d^2}\sum_{k=0}^{N-1}\epsilon_\mu\left[v^{(k)}_\mu\,,\,W^{(N)}_t(\vec{r})\right]
=-i\sum_{\mu=1}^{d^2}h_\mu\sum_{k=0}^{N-1}\sum_{n=1}^\infty\frac{i^n}{n!}\KK^n_{\vec{r}\cdot\vec{F}^{(N)}_t}[v^{(k)}_\mu]\,W^{(N)}_t(\vec{r})
$$
into three terms: the first one, with $n=1$, scales as $1/\sqrt{N}$ and, by means of~\eqref{aid1}, the definition of quantum fluctuations~\eqref{fluct} and of the matrix $\mathcal{E}$ in~\eqref{D}, can be recast as
\beann
&&\hskip-1cm
\sum_{\mu,\nu=1}^{d^2}\epsilon_\mu\,r_\nu\,\frac{1}{\sqrt{N}}\sum_{k=0}^{N-1}\Big[v^{(k)}_\nu\,,\,v^{(k)}_\mu\Big]\,W^{(N)}_t(\vec{r})
=
\sum_{\mu,\nu\,\gamma=1}^{d^2}\epsilon_\mu\,r_\nu\,J^\gamma_{\nu\mu}\frac{1}{\sqrt{N}}\sum_{k=0}^{N-1}v^{(k)}_\gamma\,W^{(N)}_t(\vec{r})\\
&&\hskip1cm
=-\sum_{\mu,\nu\,\gamma=1}^{d^2}\epsilon_\mu\,r_\nu\,J^\mu_{\nu\gamma}\Big(F^{(N)}_\gamma(t)\,+\,\sqrt{N}\omega^{(N)}_\gamma(t)\Big)\,W^{(N)}_t(\vec{r})\\
&&\hskip 2cm
=-\vec{r}\cdot\Big(\mathcal{E}\vec{F}^{(N)}_t\Big)\,W^{(N)}_t(\vec{r})\,-\,\sqrt{N}\,\vec{r}\cdot\Big(\mathcal{E}\vec{\omega}^{(N)}_t\Big)\,W^{(N)}_t(\vec{r})\ .
\eeann
The second one corresponds to $n=2$ and scales as $1/N$: by using the algebraic 
relations~\eqref{aid1} and the expressions in~\eqref{D}, it reads
\beann
&&\hskip-1cm
\frac{i}{2N}\sum_{\mu,\nu,\gamma=1}^{d^2}\epsilon_\mu\,r_\nu\,r_\gamma\,\sum_{k=0}^{N-1}\Big[v^{(k)}_\nu\,,\,\Big[v^{(k)}_\gamma\,,\,v^{(k)}_\mu\Big]\Big]\,W^{(N)}_t(\vec{r})=\\
&&=
\frac{i}{2}\sum_{\mu,\nu,\gamma=1}^{d^2}\epsilon_\mu\,r_\nu\,r_\gamma\,J^\eta_{\gamma\mu}\frac{1}{N}\sum_{k=0}^{N-1}\Big[v^{(k)}_\nu\,,\,v^{(k)}_\eta\Big]\,W^{(N)}_t(\vec{r})
=\frac{i}{2}\vec{r}\cdot\left(\mathcal{E}\,T^{(N)}\,\vec{r}\right)
\,W^{(N)}_t(\vec{r})\ .
\eeann
Finally, the norm of the third remaining term,
$$
-i\sum_{\mu=1}^{d^2}\epsilon_\mu\sum_{k=0}^{N-1}\sum_{n=3}^\infty\frac{i^n}{n!}\HH^{(N)}_{\vec{r}\cdot\vec{F}^{(N)}_t}[v^{(k)}_\mu]\,W^{(N)}_t(\vec{r})\ ,
$$ 
can be shown to vanish as $1/\sqrt{N}$ by similar methods as in the proof of Lemma~\ref{lemtool1}.
\qed
\end{proof}
\bigskip

\begin{lemma}
\label{lemmaux1}
The action of the $\widetilde{\AA}^{(N)}$ component of $\LL^{(N)}$ is such that
$$
\lim_{N\to\infty}\omega_{\vec{r}_1\vec{r}_2}\left(\gamma^{(N)}_t\left[\widetilde{\AA}^{(N)}\left[W^{(N)}_t(\vec{r})\right]\right]\right)=\lim_{N\to\infty}\omega_{\vec{r}_1\vec{r}_2}\Bigg(\gamma^{(N)}_t\left[-\frac{1}{2}\vec{r}\cdot\Big(T^{(N)}\,\widetilde{A}\,T^{(N)}\,\vec{r}\Big)\,W^{(N)}_t(\vec{r})\right]\Bigg)\ .
$$
\end{lemma}
\medskip

\begin{proof}
Using \eqref{aux1}, a same argument as in the proof of Lemma \ref{lemtool1} yields
\bea
\nonumber
\left[V^{(N)}_\mu\,,\,W^{(N)}_t(\vec{r})\right]&=&\left(V^{(N)}_\mu\,-\,W^{(N)}_t(\vec{r})\,V^{(N)}_\mu\,\left(W^{(N)}_t(\vec{r})\right)^\dag \right)\,W^{(N)}_t(\vec{r})\\
\label{aux4a}
&=&-\left(i\left[\vec{r}\cdot\vec{F}^{(N)}_t\,,\,V^{(N)}_\mu\right]\,+\,\Sigma_\mu^{(N)}\right)\,W^{(N)}_t(\vec{r})\\
\label{aux4b}
\Sigma_\mu^{(N)}&:=&\sum_{n=2}^\infty\frac{i^n}{n!}\KK_{\vec{r}\cdot\vec{F}^{(N)}_t}^n[V^{(N)}_\mu]\ .
\eea
The latter contribution can be estimated as follow: first of all, using~\eqref{auxilium} one gets the upper bound
$$
\left\|\KK_{\vec{r}\cdot\vec{F}^{(N)}_t}^n[V^{(N)}_\mu]\right\|\le
\frac{1}{\sqrt{N^{n-1}}}\,(2\,\|q\|)^n\,\|v_\mu\|\ ,
$$ 
where $q:=\sum_{\mu=1}^{d^2}r_\mu\,v_\mu$, so that
$\displaystyle\left\|\Sigma^{(N)}_\mu\right\|\leq\frac{1}{\sqrt{N}}\,{\rm e}^{2\|q\|}$ implies $\lim_{N\to\infty}\left\|\Sigma^{(N)}_\mu\right\|=0$. An analogous argument shows that $\lim_{N\to\infty}\left\|\left[\Sigma^{(N)}_\mu\,,\,V^{(N)}_\nu\right]\right\|=0$.
Then, applying again~\eqref{aux4a}, one gets
\beann
&&\hskip-1cm
\left[\left[V^{(N)}_\mu\,,\,W^{(N)}_t(\vec{r})\right]\,,\,V^{(N)}_\nu\right]=
-\Big(i\left[\left[\vec{r}\cdot\vec{F}^{(N)}_t\,,\,V^{(N)}_\mu\right]\,,\,V^{(N)}_\nu\right]\,+\,
\left[\Sigma_\mu^{(N)}\,,\,V^{(N)}_\nu\right]\Big)\,W^{(N)}_t(\vec{r})\\
&&\hskip1cm
-\Big(i\left[\vec{r}\cdot\vec{F}^{(N)}_t\,,\,V^{(N)}_\mu\right]\,+\,\Sigma_\mu^{(N)}\Big)\,\Big(i\left[\vec{r}\cdot\vec{F}^{(N)}_t\,,\,V^{(N)}_\nu\right]\,+\,\Sigma^{(N)}_\nu\Big)\,W^{(N)}_t(\vec{r})\ .
\eeann
Then, since the terms $\left[\vec{r}\cdot\vec{F}^{(N)}_t\,,\,V^{(N)}_\mu\right]$ scale as mean-field quantities and are thus bounded in the large $N$ limit, to the leading order
\beann
\widetilde{\AA}^{(N)}\left[W^{(N)}_t(\vec{r})\right]&\simeq&\sum_{\mu,\nu=1}^{d^2}\frac{\widetilde{A}_{\mu\nu}}{2}\left[\vec{r}\cdot\vec{F}^{(N)}_t\,,\,V^{(N)}_\mu\right]\left[\vec{r}\cdot\vec{F}^{(N)}_t\,,\,V^{(N)}_\nu\right]\,W^{(N)}_t(\vec{r})\\
&=&
-\frac{1}{2}\vec{r}\cdot\Big(\,T^{(N)}\,\widetilde{A}\,T^{(N)}\vec{r}\Big)\,W^{(N)}_t(\vec{r})\ .
\eeann
\qed
\end{proof}
\bigskip

\begin{lemma}
\label{lemmaux2}
The action of the $\widetilde{\BB}^{(N)}$ component of $\LL^{(N)}$ is such that
\beann
&&\hskip-.8cm
\lim_{N\to\infty}\omega_{\vec{r}_1\vec{r}_2}\left(\gamma^{(N)}_t\left[\widetilde{\BB}^{(N)}\left[W^{(N)}_t(\vec{r})\right]\Big)\,W^{(N)}_t(\vec{r})\right]\right)=\\
&&\hskip-.5cm
=\lim_{N\to\infty}\omega_{\vec{r}_1\vec{r}_2}\Bigg(\gamma^{(N)}_t\Bigg[\Bigg(
\sqrt{N}\,i\,\vec{r}\cdot\widetilde{D}^{(N)}_t\vec{\omega}^{(N)}_t\,-\,i\,\vec{r}\cdot\Big(T^{(N)}\,\widetilde{B}\,-\,\widetilde{D}^{(N)}_t\Big)\,\vec{F}^{(N)}_t\\
&&\hskip1cm
+\,\frac{1}{2}\vec{r}\cdot\Big(\widetilde{D}^{(N)}_t\,T^{(N)}\,+\,2\,i\,T^{(N)}\,h^{(re)}\,T^{(N)}\Big)\,\vec{r}\Bigg)\,W^{(N)}_t(\vec{r})\Bigg]\Bigg)
\ ,
\eeann
where $\vec{\omega}^{(N)}_t$ is the real vector with components given by~\eqref{macr1a} and $\widetilde{D}^{(N)}_t$ is the matrix defined in~\eqref{D}.
\end{lemma}
\medskip

\begin{proof}
The operators  
$\displaystyle V^{(N)}_\nu=\frac{1}{\sqrt{N}}\sum_{k=0}^{N-1}v_\nu^{(k)}$ are turned into fluctuations $F^{(N)}_\nu(t)$ (see relation \eqref{fluct}) by adding and subtracting the scalars $\omega^{(N)}_t(v_\nu^{(k)})$. Thus, with the notation of \eqref{macr1a}, one gets:
\bea
\label{aux9a}
\widetilde{\BB}^{(N)}\left[W^{(N)}_t(\vec{r})\right]&=&\frac{1}{2}\sum_{\mu,\nu=1}^{d^2}\widetilde{B}_{\mu\nu}\,\left\{\left[V^{(N)}_\mu\,,\,W^{(N)}_t(\vec{r})\right]\,,\,F^{(N)}_\nu(t)\right\}\\
\label{aux9b}
&+&\,\sum_{\mu,\nu=1}^{d^2}\widetilde{B}_{\mu\nu}\,\omega^{(N)}_\nu(t)\sum_{k=0}^{N-1}\left[v^{(k)}_\mu\,,\,W^{(N)}_t(\vec{r})\right]\ ,
\eea
with $\omega^{(N)}_\nu(t)$ given by \eqref{macr1a}.
We denote by $\BB_1^{(N)}[W^{(N)}_t(\vec{r})]$, respectively $\BB_2^{(N)}[W^{(N)}_t(\vec{r})]$ the expression in \eqref{aux9a}, respectively in \eqref{aux9b} and treat them separately.
\medskip

\noindent
$\bullet$ \textbf{Term} $\BB_1^{(N)}:$\ A similar argument as the one leading to \eqref{aux4a}-\eqref{aux4b} allows one to recast
\beann
\BB^{(N)}_1\left[W^{(N)}_t(\vec{r})\right]&=&-\frac{1}{2}\sum_{\mu,\nu=1}^{d^2}\widetilde{B}_{\mu\nu}\left\{P\,W^{(N)}_t(\vec{r})\,,\,F^{(N)}_\nu(t)\right\}\ ,\quad\hbox{where}\\
P&=&i\left[\vec{r}\cdot\vec{F}^{(N)}_t\,,\,V^{(N)}_\mu\right]\,+\,\Sigma^{(N)}_\mu\ ,\qquad
\Sigma_\mu^{(N)}:=\sum_{n=2}^\infty\frac{i^n}{n!}\KK_{\vec{r}\cdot\vec{F}^{(N)}_t}^n[V^{(N)}_\mu]\ .
\eeann
The anti-commutators can be studied as follows; firstly, we rewrite 
\bea
\nonumber
&&\hskip-1.5cm
\left\{P\,W^{(N)}_t(\vec{r})\,,\,F^{(N)}_t(v_\nu)\right\}
=\left(P\,\left(W^{(N)}_t(\vec{r})\,F^{(N)}_\nu(t)\,W^{(N)}_t(-\vec{r})\right)\,+\,P\,F^{(N)}_\nu(t)\right)\,W^{(N)}_t(\vec{r})\\
\nonumber
&&\hskip 3cm+\,\left[F^{(N)}_\nu(t)\,,\,P\right]\,W^{(N)}_t(\vec{r})\\
\label{aux5a}
&&\hskip -1cm
=2\,P\,F^{(N)}_\nu(t)\,W^{(N)}_t(\vec{r})\,+\,i\,P\,\left[\vec{r}\cdot\vec{F}^{(N)}_t\,,\,F^{(N)}_\nu(t)\right]\,W^{(N)}_t(\vec{r})\,+\,P\,\Sigma^{(N)}_\nu\,W^{(N)}_t(\vec{r})\\
\label{aux5b}
&&\hskip 3cm
+\,\left[F^{(N)}_\nu(t)\,,\,P\right]\,W^{(N)}_t(\vec{r})\ .
\eea
Then, using the scaling of the local fluctuations $F^{(N)}_\alpha(t)$ and that of 
the infinite sums $\Sigma^{(N)}_\beta$, together with the fact that spin operators at different sites commute, one readily finds that, for all $\alpha,\beta=1,2,\ldots,d^2$,
\be
\label{aux6}
\lim_{N\to\infty}\left\|\left[F^{(N)}_\alpha(t)\,,\,\Sigma^{(N)}_\beta\right]\right\|=\lim_
{N\to\infty}\left\|\left[F^{(N)}_\alpha(t)\,,\,P\right]\right\|=0\ .
\ee
Thus, the only terms in \eqref{aux5a} and \eqref{aux5b} whose norms do not vanish in large $N$ limit are  
\be
\label{aux7}
2\,\Sigma^{(N)}_\mu\,F^{(N)}_\nu(t)\,W^{(N)}_t(\vec{r})\ ,\qquad 2i\left[\vec{r}\cdot\vec{F}^{(N)}_t\,,\,V^{(N)}_\mu\right]F^{(N)}_\nu(t)\,W^{(N)}_t(\vec{r})
\ee
from the first contribution to \eqref{aux5a} and
\be
\label{aux8}
-\,\left[\vec{r}\cdot\vec{F}^{(N)}_t\,,\,V^{(N)}_\mu\right]\,\left[\vec{r}\cdot\vec{F}^{(N)}_t\,,\,F^{(N)}_\nu(t)\right]\,W^{(N)}_t(\vec{r})=
-\sum_{\alpha,\beta=1}^{d^2}r_\alpha\,r_\beta\,T^{(N)}_{\alpha\mu}\,T^{(N)}_{\beta\nu}
\,W^{(N)}_t(\vec{r})
\ee
from the second one, where use has been made of \eqref{tcomm}. As regards the first term in \eqref{aux7}, using \eqref{aux6}, Lemma \ref{tool1} and Lemma \ref{bound} in Appendix B, one finds
\beann
&&
\lim_{N\to\infty}\omega_{\vec{r}_1\vec{r}_2}\left(\gamma^{(N)}_t\left[\Sigma^{(N)}_\mu\,
F^{(N)}_\nu(t)\, W^{(N)}_t(\vec{r})\right]\right)=
\lim_{N\to\infty}\omega_{\vec{r}_1\vec{r}_2}\left(\gamma^{(N)}_t\left[F^{(N)}_\nu(t)\,
\Sigma^{(N)}_\mu\,W^{(N)}_t(\vec{r})\right]\right)\\
&&\hskip 2cm\leq
\lim_{N\to\infty}\|\Sigma_\mu^{(N)}\|\,\sqrt{\omega\left(W^{(N)}(\vec{r}_1)\gamma^{(N)}_t\left[\Big(F^{(N)}_\nu(t)\Big)^2\right]\,\left(W^{(N)}(\vec{r}_1)\right)^\dag\right)}=0\ .
\eeann
On the other hand, the second term in \eqref{aux7} contributes to the large $N$ limit with
\beann
&&\hskip-1cm
-\,i\,\sum_{\mu,\nu=1}^{d^2}\widetilde{B}_{\mu\nu}\left[\vec{r}\cdot\vec{F}^{(N)}_t\,,\,V^{(N)}_\mu\right]\,F^{(N)}_\nu(t)=-\,i\,\sum_{\mu,\nu,\alpha=1}^{N-1}\widetilde{B}_{\mu\nu}\,r_\alpha\,\frac{1}{N}\sum_{k=0}^{N-1}\left[v^{(k)}_\alpha\,,\,v^{(k)}_\mu\right]\,F^{(N)}_\nu(t)\\
&&\hskip 2cm
=-i\,\vec{r}\cdot\Big(T^{(N)}\,\widetilde{B}\,\vec{F}^{(N)}_t\Big)\ .
\eeann
Concerning \eqref{aux8}, it gives rise to a contribution 
$$
\frac{1}{2}\sum_{\alpha,\beta\,\mu,\nu=1}^{d^2}r_\alpha\,r_\beta\,T^{(N)}_{\alpha\mu}\,T^{(N)}_{\beta\nu}\,\widetilde{B}_{\mu\nu}\,W^{(N)}_t(\vec{r})=-\frac{1}{2}\vec{r}\cdot\Big(T^{(N)}\widetilde{B}\,T^{(N)}\vec{r}\Big)\,W^{(N)}_t(\vec{r})
$$
that can be again estimated by means of Lemma~\ref{tool1} in Appendix B. 
With $h^{(re)}$ as defined in~\eqref{hmat1}, one has $\widetilde{B}=B+2ih^{(re)}$. Then 
\beann
&&
\left|\omega_{\vec{r}_1\vec{r}_2}\left(\gamma^{(N)}_t\left[\vec{r}\cdot\Big(T^{(N)}\,\left(\frac{\widetilde{B}}{2}\,-\,i\,h^{(re)}\right)\,T^{(N)}\,\vec{r}\Big)\,W^{(N)}_t(\vec{r})\right]\right)\right|^2\\
&&\hskip 1cm
\le\,\frac{1}{4}\left|\omega\left(W^{(N)}(\vec{r}_1)\,\gamma^{(N)}_t\left[\Big(\vec{r}\cdot\Big((T^{(N)}\,B\,T^{(N)}\vec{r}\Big)\Big)^2\right]\,\left(W^{(N)}(\vec{r}_1)\right)^\dag\right)\right|\ .
\eeann   
Since the operator-valued matrix $T^{(N)}$ has entries which scale as mean-field quantities,
and, in the large $N$ limit, tend to the entries of the symplectic matrix $\sigma(\vec{\omega}_t)$, Lemma \ref{lemtool2} yields
$$
\lim_{N\to\infty}\left|\omega_{\vec{r}_1\vec{r}_2}\left(\gamma^{(N)}_t\left[\vec{r}\cdot\Big(T^{(N)}\,\Big(\widetilde{B}-2ih^{(re)}\Big)\,T^{(N)}\vec{r}\Big)\,W^{(N)}_t(\vec{r})\right]\right)\right|\le\,\left|\vec{r}\cdot\Big(\sigma(\vec{\omega}_t)\,B\,\sigma(\vec{\omega}_t)\vec{r}\Big)\right|^2=0
$$
because of the anti-symmetric character of the matrices $\sigma(\vec{\omega}_t)$ and $B$.
Concluding,
\bea
\nonumber
&&\hskip-1cm
\lim_{N\to\infty}\omega_{\vec{r}_1\vec{r}_2}\left(\gamma^{(N)}_t\left[\BB_1^{(N)}\left[W^{(N)}_t(r)\right]\right]\right)=\\
\label{aux9}
&&\hskip .3cm
=\lim_{N\to\infty}\omega_{\vec{r}_1\vec{r}_2}\Bigg(\gamma^{(N)}_t\left[-i\,\vec{r}\cdot\Big(T^{(N)}\,\widetilde{B}\,\vec{F}^{(N)}_t\,+\,\Big(T^{(N)}\,h^{(re)}\,T^{(N)}\vec{r}\Big)\Big)\,W^{(N)}_t(r)\right]\Bigg)\ .
\eea
\medskip

\noindent
$\bullet$ \textbf{Term} $\BB_2^{(N)}:$\ In analogy with the treatment of previous commutators, we first recast \eqref{aux9b} as follows:
\bea
\label{aux10a}
\hskip-.8cm
\BB_2^{(N)}\left[W^{(N)}_t(r)\right]&=&\sum_{\mu,\nu=1}^{d^2}\,\widetilde{B}_{\mu\nu}\,\omega^{(N)}_\nu(t)\,\left(i\left[\sum_{k=0}^{N-1}v^{(k)}_\mu\,,\,\vec{r}\cdot\vec{F}^{(N)}_t\right]\right.\\
\label{aux10b}
\hskip-.8cm&+&\left.\frac{1}{2}\left[\vec{r}\cdot\vec{F}^{(N)}_t\,,\,\left[\vec{r}\cdot\vec{F}^{(N)}_t\,,\,\sum_{k=0}^{N-1}v_\mu^{(k)}\right]\right]\right)\,W^{(N)}_t(r)\,+\,B_N\ ,
\eea
where $B_N$ is a term which vanishes in norm when $N\to\infty$.
Using the matrix basis relations \eqref{ONB},~\eqref{D} and the anti-symmetry of the operator-valued matrix $T^{(N)}$, the double commutator in \eqref{aux10b} can be recast 
in the form
\beann
&&\hskip-.5cm
\frac{1}{2}\sum_{\mu,\nu=1}^{d^2}\widetilde{B}_{\mu\nu}\,\omega^{(N)}_\nu(t)\,\sum_{\alpha\,\beta=1}^{d^2}\,r_\alpha r_\beta\,\frac{1}{N}\,\sum_{k=0}^{N-1}\left[v^{(k)}_\alpha\,,\,\left[v^{(k)}_\beta\,,\,v_\mu^{(k)}\right]\right]=\\
&&=\frac{1}{2}\sum_{\alpha,\beta,\gamma=1}^{d^2}\,r_\alpha r_\beta\,\left(\sum_{\mu,\nu=1}^{d^2}\widetilde{B}_{\mu\nu}\,J^\gamma_{\beta\mu}\,\omega^{(N)}_\nu(t)\right)\,\frac{1}{N}\sum_{k=0}^{N-1}\left[v^{(k)}_\alpha\,,\,v^{(k)}_\gamma\right]=\frac{1}{2}\vec{r}\cdot\Big(\widetilde{D}^{(N)}_t\,T^{(N)}\vec{r}\Big)\ ,
\eeann
where $\widetilde{D}^{(N)}_t$ is the matrix given in~\eqref{D}.
Analogously, the sum in \eqref{aux10a} can be rewritten as 
\beann
&&i\,\sum_{\alpha,\beta=1}^{d^2}r_\beta\sum_{\mu,\nu=1}^{d^2}\omega^{(N)}_\nu(t)\,\widetilde{B}_{\mu\nu}\,J^\alpha_{\mu\beta}\,\frac{1}{\sqrt{N}}\sum_{k=0}^{N-1}v^{(k)}_\alpha
=\,i\,\sum_{\alpha,\beta=1}^{d^2}r_\beta\,\widetilde{D}^{(N)}_{\beta\alpha}(t)\,\frac{1}{\sqrt{N}}\sum_{k=0}^{N-1}v^{(k)}_\alpha\\
&&\hskip 1cm
=\vec{r}\cdot\Big(\Big(i\,\widetilde{D}^{(N)}_t\,\vec{F}^{(N)}_t\Big)\,+\,\sqrt{N}\,\vec{r}\cdot\Big(\Big(i\,\widetilde{D}^{(N)}_t\,\vec{\omega}^{(N)}_t\Big)\ ,
\eeann
where $\vec{\omega}^{(N)}_t$ denotes the vector with $d^2$ real components $\omega^{(N)}_\mu(t)$ given by \eqref{macr1a}.
\qed
\end{proof}
\bigskip

With the help of the previous Proposition, we now conclude with the proof of Theorem \ref{expfluct}. 
\medskip 

\noindent
\textbf{Theorem 3.}\quad\textit{
According to Definition \ref{meslimdef}, the dynamics of quantum fluctuations is given by 
the mesoscopic limit $\Phi_t^{\vec{\omega}}:=m-\lim_{N\to\infty}\gamma^{(N)}_t$, where
\beann
\Phi_t^{\vec{\omega}}\left[W(\vec{r})\right]&=&\exp\Big(-\frac{1}{2}\vec{r}\cdot\Big(Y_t(\vec{\omega})\,\vec{r}\Big)\Big)\,W(X_t^{tr}(\vec{\omega})\vec{r})\ ,
\eeann
where, with $\TT$ denoting time-ordering, 
\beann
X_t(\vec{\omega})&:&=\TT{\rm e}^{\int_0^t{\rm d}s\, Q(\vec{\omega}_s)}\\
Q(\vec{\omega}_t)&:=&-i\sigma(\vec{\omega}_t)\,\widetilde{B}\,+\,D(\vec{\omega}_t)\\
Y_t(\vec{\omega})&:=&\int_0^t{\rm d}s\, X_{t,s}(\vec{\omega})\,\Big(\sigma(\vec{\omega}_s)\,A\,\sigma^{tr}(\vec{\omega}_s)\Big)\,X^{tr}_{t,s}(\vec{\omega})\ .
\eeann
In the above expression, $X_{t,s}(\vec{\omega}):=X_t(\vec{\omega})\,X^{-1}_s(\vec{\omega})$, $A$ is the symmetric component of the Kossakowski matrix $C$ in~\eqref{mflind1a00}, $\widetilde{B}=B+\,2\,i\,h^{(re)}$ in~\eqref{hmat1}.
Finally, $\sigma(\vec{\omega}_t)$ is the time-dependent symplectic matrix with entries given by \eqref{tsympform1} and $D(\vec{\omega}_t)$ is the matrix defined in \eqref{d0}.}
\medskip

\begin{proof}
We prove the assertion by showing that $\lim_{N\to\infty}I^{(N)}(t)=0$, where
$$
I^{(N)}(t):=\omega_{\vec{r}_1\vec{r}_2}\left(\gamma^{(N)}_t\left[W^{(N)}_t(\vec{r})\right]\,-\,
{\rm e}^{-1/2\,\vec{r}\cdot\left(Y_t(\vec{\omega})\vec{r}\right)}\,W^{(N)}(X_t^{tr}(\vec{\omega})\vec{r})\right)\ .
$$
Indeed, from Theorem \ref{th1} we know that, for all $\vec{r}_{1,2}\in\RR^{d^2}$,
$$
\lim_{N\to\infty}\omega_{\vec{r}_1\vec{r}_2}\left(W^{(N)}(X^{tr}_t(\vec{\omega})\vec{r})\right)=
\Omega_{\vec{r}_1\vec{r}_2}\left(W(X^{tr}_t(\vec{\omega})\vec{r})\right)\ .
$$
Since the matrix $\widetilde{B}$ is such that the conjugated matrix $\widetilde{B}^*=-\widetilde{B}$ and the matrices $\sigma(\omega_t)$ and $D(\vec{\omega}_t)$ are real, such is also $Q(\vec{\omega}_t)$ as well as
the matrix $X_t(\vec{\omega})$ solution to 
$$
\frac{{\rm d}}{{\rm d}t}X_t(\vec{\omega})=Q(\vec{\omega}_t)\,X_t(\vec{\omega})\ ,\qquad
X_0(\vec{\omega})=\mathbf{1}	 .
$$ 
Moreover, its inverse matrix, $X^{-1}_t(\vec{\omega})$, is given by the inverted time-ordered exponential 
$$
X^{-1}_t(\vec{\omega})=\mathbb{T}^{-1}\exp\left(-\int_0^t{\rm d}s\,Q(\vec{\omega}_s)\right)\ ,
\qquad \frac{{\rm d}}{{\rm d}t}X^{-1}_t(\vec{\omega})=-X^{-1}_t(\vec{\omega})\,Q(\vec{\omega}_t)\ .
$$
Then, we set $X_{t,s}(\vec{\omega}):=X_t(\vec{\omega})\,X^{-1}_s(\vec{\omega})$, and introduce the matrix
$$
Y_{t,s}(\vec{\omega})=\int_s^t{\rm d}\tau\, X_{t,\tau}(\vec{\omega})\,\sigma(\vec{\omega}_\tau)\,A\,\sigma^{tr}(\vec{\omega}_\tau)\,X^{tr}_{t,\tau}(\vec{\omega})\ .
$$
Since $W^{(N)}_{s=0}(X^{tr}_{t,0}(\vec{\omega})\vec{r})=W^{(N)}(X^{tr}_t(\vec{\omega})\vec{r})$, we write:
\bea
\nonumber
&&\hskip-1cm
\gamma^{(N)}_t\left[W^{(N)}_t(\vec{r})\right]\,-\,{\rm e}^{-1/2\,\vec{r}\cdot\left(Y_t(\vec{\omega})\vec{r}\right)}\,
W^{(N)}(X_t^{tr}(\vec{\omega})\vec{r})=\\
\nonumber
&&\hskip 2cm
=\int_0^t{\rm d}s\,\frac{{\rm d}}{{\rm d}s}\Big(\gamma_s^{(N)}\left[W^{(N)}_s(X_{t,s}^{tr}(\vec{\omega})\vec{r})\right]\,{\rm e}^{-1/2\,\vec{r}\cdot\left(Y_{t,s}(\vec{\omega})\vec{r}\right)}\Big)=\\
&&\hskip 3cm
=\int_0^t{\rm d}s\, \gamma^{(N)}_s\left[\Delta^{(N)}_{t,s}\right]\,{\rm e}^{-\frac{1}{2}\,\vec{r}\cdot\left(Y_{t,s}(\vec{\omega})\vec{r}\right)}\\
\label{aidd1}
&&\hskip-1cm
\nonumber
\Delta^{(N)}_{t,s}=\LL^{(N)}\left[W^{(N)}_s(X_{t,s}^{tr}(\vec{\omega})\,r)\right]\,+\,\frac{{\rm d}}{{\rm d}s}W^{(N)}_s(X_{t,s}^{tr}(\vec{\omega})\vec{r})\\
&&\hskip 2cm
-\,\frac{1}{2}
\frac{{\rm d}}{{\rm d}s}\Big(\vec{r}\cdot\left(Y_{t,s}(\vec{\omega})\vec{r}\right)\Big)\,W^{(N)}_s(X^{tr}_{t,s}(\vec{\omega})\vec{r})\ .
\eea

Observe that the time-derivative of the exponent of $W^{(N)}_s(X^{tr}_{t,s}(\vec{\omega})\vec{r})$ yields
$$
\frac{{\rm d}}{{\rm d}s}\left(X_{t,s}^{tr}(\vec{\omega})\vec{r}\cdot\vec{F}^{(N)}_s\right)=-\left(X^{tr}_{t,s}(\vec{\omega})\vec{r}\right)\cdot \left(Q(\vec{\omega}_s)\vec{F}^{(N)}_s\right)\,-\,\sqrt{N}\,\left(X^{tr}_{t,s}(\vec{\omega})\vec{r}\right)\cdot\dot{\vec{\omega}}^{(N)}_s\ ,
$$
where $\vec{\omega}^{(N)}_t$ stands for the vector with components 
$\omega^{(N)}_\nu(t)$ (see \eqref{macr1a}).
Then, from the well known result of Lemma \ref{tool3} reported and proved for sake of completeness in Appendix B,
\beann
&&
\frac{{\rm d}}{{\rm d}s}W^{(N)}_s(X^{tr}_{t,s}(\vec{\omega})\vec{r})=\sum_{k=1}^\infty\frac{i^k}{k!}\mathbb{K}_{\left(X_{t,s}^{tr}(\vec{\omega})\vec{r}\right)\cdot\vec{F}^{(N)}_s}^{k-1}\left[\frac{{\rm d}}{{\rm d}s}\left(X_{t,s}^{tr}(\vec{\omega})\vec{r}\right)\cdot\vec{F}^{(N)}_s\right]\,W^{(N)}_s(X^{tr}_{t,s}(\vec{\omega})\vec{r}) \\
&&=\Bigg(-iX^{tr}_{t,s}(\vec{\omega})\vec{r}\cdot\left(Q(\vec{\omega}_s)\vec{F}^{(N)}_s\right)\,-i\,\sqrt{N}\left(X^{tr}_{t,s}(\vec{\omega})\vec{r}\right)\cdot\dot{\vec{\omega}}^{(N)}_s\\
&&+\frac{1}{2}\left[\left(X^{tr}_{t,s}(\vec{\omega})\vec{r}\right)\cdot\vec{F}^{(N)}_s\,,\,\left(X^{tr}_{t,s}(\vec{\omega})\vec{r}\right)\cdot\left(Q(\vec{\omega}_s)\,F^{(N)}_s\right)\right]\,+\,Z^{(N)}_t\Bigg)\,W^{(N)}_s(X^{tr}_{t,s}(\vec{\omega})\vec{r})\ ,
\eeann
where $Z^{(N)}_t$ contains an infinite sum  starting from $k=3$ and thus vanishes in norm when $N\to \infty$, while the commutator yields
$$
\left[\left(X^{tr}_{t,s}(\vec{\omega})\vec{r}\right)\cdot\vec{F}^{(N}_s\,,\,\left(X^{tr}_{t,s}(\vec{\omega})\vec{r}\right)\cdot\left(Q(\vec{\omega}_s)\vec{F}^{(N)}_s\right)\right]\,=\,-\,\left(X^{tr}_{t,s}(\vec{\omega})\vec{r}\right)\cdot\left(Q(\vec{\omega}_s)\,T^{(N)}\,X^{tr}_{t,s}(\vec{\omega})\vec{r}\right)\ ;
$$
thus, through $T^{(N)}$, it exhibits a mean-field scaling when $N\to\infty$. Finally,
$$
\frac{{\rm d}}{{\rm d}s}\vec{r}\cdot\left(Y_{t,s}(\vec{\omega})\vec{r}\right)=\left(X^{tr}_{t,s}(\vec{\omega})\vec{r}\right)\cdot\left(\sigma(\vec{\omega}_s)\,A\,\sigma(\vec{\omega}_s)\,X^{tr}_{t,s}(\vec{\omega})\vec{r}\right)\ .
$$
Setting $\vec{\xi}=X^{tr}_{t,s}(\vec{\omega})\vec{r}$ for sake of simplicity, \eqref{aidd1} can thus be recast as
\beann
\Delta^{(N)}_{t,s}&=&\LL^{(N)}\left[W^{(N)}_s(\vec{\xi})\right]\,-i\,\sqrt{N}\,\left(\vec{\xi}\cdot\dot{\vec{\omega}}^{(N)}_s\right)\,W^{(N)}_s(\vec{\xi})\,-i\,\Big(\vec{\xi}\cdot\left(Q(\vec{\omega}_s)\vec{F}^{(N)}_s\right)\Big)\,W^{(N)}_s(\vec{\xi})\\
&-&\frac{1}{2}\left(\,\vec{\xi}\cdot\Big(Q(\vec{\omega}_s)\,T^{(N)}\,+\,\sigma(\vec{\omega}_s)\,A\,\sigma(\vec{\omega}_s)\Big)\vec{\xi}\right)\,W^{(N)}_s(\vec{\xi})\\
&+&\,Z^{(N)}_t\,W^{(N)}_s(\vec{\xi})\ .
\eeann
The last term does not contribute to the mesoscopic limit and Proposition \ref{LINDACT} provides the mesoscopic behaviour of the first contribution to the right hand side of the equality above. We now group together terms with the same scaling with $1/N$ and show that, in the mesoscopic limit, the following quantities vanish:
\bea
\label{thproof1}
&&\hskip-1cm
\gamma^{(N)}_s\left[\sqrt{N}\,\vec{\xi}\cdot\left(\Big(\widetilde{D}^{(N)}_s+i\mathcal{E}\Big)\,\vec{\omega}^{(N)}_s\,-\,\dot{\vec{\omega}}^{(N)}_s\right)\,W^{(N)}_s(\vec{\xi})\right]\\
\label{thproof3}
&&\hskip-1cm
\gamma^{(N)}_s\left[\left(\vec{\xi}\cdot\Big(\,T^{(N)}\,\widetilde{B}\,-\,\widetilde{D}^{(N)}_s\,-\,i\,\mathcal{E}\,+\,\,Q(\vec{\omega}_s)\Big)\vec{F}^{(N)}_s\right)\,W^{(N)}_s(\vec{\xi})\right]
\\
\label{thproof2} 
&&\hskip-1cm
\gamma^{(N)}_s\left[\left(\vec{\xi}\cdot\Big(T^{(N)}\,A\,T^{(N)}\,+\,\sigma(\vec{\omega}_s)\,A\,\sigma(\vec{\omega}_s)\Big)\vec{\xi}\right)\,W^{(N)}_s(\vec{\xi})\right]
\\\label{thproof4}
&&\hskip-1cm
\gamma^{(N)}_s\left[\left(\vec{\xi}\cdot\Big(D^{(N)}_s\,T^{(N)}\,-\,Q(\vec{\omega}_s)\,T^{(N)}\,-\,2\,i\,T^{(N)}\,h\,T^{(N)}\Big)\vec{\xi}\right)\,W^{(N)}_s(\vec{\xi})\right]\ .
\eea
Notice that $T^{(N)}$ is an operator-valued matrix with entries that scale as mean-field observables; then, we proceed by showing that, in the large $N$ limit, in the above expressions, mean-field operators of the form $\displaystyle M^{(N)}_\alpha:=\frac{1}{N}\sum_{k=0}^{N-1}v_\alpha^{(k)}$ can be substituted by their expectations 
$\displaystyle\omega_\alpha(t)=\lim_{N\to\infty}\omega^{(N)}_t(M^{(N)}_\alpha)$ with respect to the large $N$ limit of the time-evolving state $\omega^{(N)}_t=\omega\circ\gamma^{(N)}_t$.
Indeed, in \eqref{thproof2} and \eqref{thproof4} there appear terms of the type
\be
\gamma_s^{(N)}\left[M^{(N)}_\alpha\, M^{(N)}_\beta\, W_s^{(N)}(\vec{\xi})\right]\ ,
\label{P2}
\ee
while terms of the form
\be
\gamma_s^{(N)}\left[M^{(N)}_\alpha\, F_s^{(N)}(v_\mu)\, W_s^{(N)}(\vec{\xi})\right]\ ,
\label{P3}
\ee
appear in \eqref{thproof3} and terms as
\be
\gamma_s^{(N)}\left[M^{(N)}_\alpha\, W_s^{(N)}(\vec{\xi})\right]\ ,
\label{P1}
\ee
are to be found both in~\eqref{thproof1} and~\eqref{thproof4}.

Let us consider the latter expression  and study the limit
$$
\lim_{N\to\infty}\omega_{\vec{r}_1\vec{r}_2}\left(\gamma_s^{(N)}\left[\left(M^{(N)}_\alpha-\omega_\alpha(s)\right)\, W_s^{(N)}(\vec{\xi})\right]\right)\ .
$$
Using Lemma \ref{tool1} in Appendix B, and the Kadison inequality, one has
\begin{eqnarray*}
&&\Big|\omega_{\vec{r}_1\vec{r}_2}\left(\gamma_s^{(N)}\left[\left(M^{(N)}_\alpha-\omega_\alpha(s)\right)\, W_s^{(N)}(\vec{\xi})\right]\right)\Big|\le\\
&&\hskip 2cm\leq\sqrt{\omega\left(W^{(N)}(\vec{r}_1)\gamma_s^{(N)}\left[\left(M^{(N)}_\alpha-\omega_\alpha(t)\right)^2\right]\left(W^{(N)}(\vec{r}_1)\right)^\dag\right)}\ .
\end{eqnarray*}
Then, Lemma \ref{lemtool2} and Corollary \ref{spt} yield
\begin{eqnarray*}
&&
\lim_{N\to\infty}\omega\left(W^{(N)}(\vec{r}_1)\gamma_s^{(N)}\left[\left(M^{(N)}_\alpha-\omega_\alpha(s)\right)^2\right]\left(W^{(N)}(\vec{r}_1)\right)^\dag\right)=\\
&&\hskip 2cm
=\lim_{N\to\infty}\Big(\omega\left(\gamma_s^{(N)}\left[M^{(N)}_\alpha-\omega_\alpha(t)\right]\right)\Big)^2=0\ .
\end{eqnarray*}
Therefore, we have that
$$
\lim_{N\to\infty}\omega_{\vec{r}_1\vec{r}_2}\left(\gamma_s^{(N)}\left[M^{(N)}_\alpha\, W_s^{(N)}(\vec{\xi})\right]\right)=\omega_\alpha(s)\lim_{N\to\infty}\omega_{\vec{r}_1\vec{r}_2}\left(\gamma_s^{(N)}\left[ W_s^{(N)}(\vec{\xi})\right]\right)\ .
$$
By a similar argument, one shows that
\begin{eqnarray*}
&&
\lim_{N\to\infty}\omega_{\vec{r}_1\vec{r}_2}\left(\gamma_s^{(N)}\left[M^{(N)}_\alpha\, M^{(N)}_\beta\, W_s^{(N)}(\vec{\xi})\right]\right)=\\
&&=\omega_\alpha(s)\lim_{N\to\infty}\omega_{\vec{r}_1\vec{r}_2}\left(\gamma_s^{(N)}\left[M^{(N)}_\beta\, W_s^{(N)}(\vec{\xi})\right]\right)=\omega_\alpha(s)\omega_\beta(s)\lim_{N\to\infty}\omega_{\vec{r}_1\vec{r}_2}\left(\gamma_s^{(N)}\left[ W_s^{(N)}(\vec{\xi})\right]\right)\ .
\end{eqnarray*}
Now we consider the following quantity
$$
\lim_{N\to\infty}\omega_{\vec{r}_1\vec{r}_2}\left(\gamma_s^{(N)}\left[\left(M^{(N)}_\alpha-\omega_{\alpha}(s)\right)\, F_\mu^{(N)}(s)\, W_s^{(N)}(\vec{\xi})\right]\right)\ .
$$
The operator $M^{(N)}_\alpha$ scales as a mean-field quantity; therefore, the norm of its commutator with quantities that scale as fluctuations vanishes in the large $N$ limit. Then, because of Lemma \ref{lemtool2}, we have that
\begin{equation*}
\begin{split}
\lim_{N\to\infty}\omega_{\vec{r}_1\vec{r}_2}\left(\gamma_s^{(N)}\left[\left(M^{(N)}_\alpha-\omega_{\alpha}(s)\right)\, F_\mu^{(N)}(s)\, W_s^{(N)}(\vec{\xi})\right]\right)=\hspace{4cm}\\
\lim_{N\to\infty}\omega_{\vec{r}_1\vec{r}_2}\left(\gamma_s^{(N)}\left[ F_\mu^{(N)}(s)\, W_s^{(N)}(\vec{\xi})\left(M^{(N)}_\alpha-\omega_{\alpha}(s)\right)\right]\right)\ .
\end{split}
\end{equation*}
Finally, Lemma~\ref{tool1} in Appendix B and the Kadison inequality applied to the term on the right-hand side of the equality, yield the following bound
\begin{equation*}
\begin{split}
\Big|\omega_{\vec{r}_1\vec{r}_2}\left(\gamma_s^{(N)}\left[ F_\mu^{(N)}(s)\, W_s^{(N)}(\vec{\xi})\left(M^{(N)}_\alpha-\omega_{\alpha}(s)\right)\right]\right)\Big|\le \hspace{4cm}\\
\sqrt{\omega\left(W^{(N)}(\vec{r}_1)\,\gamma_s^{(N)}\left[\left(F_\mu^{(N)}(s)\right)^2\right]\,\left(W^{(N)}(\vec{r}_1)\right)^\dagger\right)}\ \times\hspace{3cm}\\
\times\sqrt{\omega\Big(\left((W^{(N)}(\vec{r}_2)\right)^\dagger\,\gamma_s^{(N)}\left[\left(M^{(N)}_\alpha-\omega_{\alpha}(s)\right)^2\right]\,W^{(N)}(\vec{r}_2)\Big)}\ .
\end{split}
\end{equation*}
The first term on the right-hand side is bounded by Lemma~\ref{bound} in Appendix B, while the second one, as already shown, vanishes in the large $N$ limit. Therefore, 
\begin{equation*}
\begin{split}
\lim_{N\to\infty}\omega_{\vec{r}_1\vec{r}_2}\left(\gamma_s^{(N)}\left[M^{(N)}_\alpha\, F_s^{(N)}(v_\mu)\, W_s^{(N)}(\xi)\right]\right)=\hspace{4cm}\\
=\omega_\alpha(s)\lim_{N\to\infty}\omega_{\vec{r}_1\vec{r}_2}\left(\gamma_s^{(N)}\left[F_s^{(N)}(v_\mu)\, W_s^{(N)}(\xi)\right]\right)\, .
\end{split}
\end{equation*}
Applying these considerations to the quantities \eqref{thproof1}--\eqref{thproof4}, one thus sees that~\eqref{thproof1} vanishes in the large $N$ limit because of Proposition~\ref{scal}.
Furthermore
$$
\lim_{N\to\infty}\omega\left(\gamma_s^{(N)}\left[\vec{\xi}\cdot\Big(T^{(N)}\,A\,T^{(N)}\vec{\xi}\Big)\right]\right)=-\vec{\xi}\cdot\Big(\sigma(\vec{\omega}_{s})A\sigma(\vec{\omega}_{s})\vec{\xi}\Big)\, ,
$$ 
whence \eqref{thproof2} vanishes in the arge $N$ limit and analogously
\begin{equation*}
\lim_{N\to\infty}\omega	\left(\gamma_s^{(N)}\left[\left(T^{(N)}\,\widetilde{B}\,-\,\widetilde{D}_s^{(N)}\,-\,i\,\mathcal{E}\,+\,Q(\vec{\omega}_s)\right)\right]\right)=0\ .
\end{equation*}
Finally, as regards the large $N$ limit of~\eqref{thproof4}, using~\eqref{D}, \eqref{matrices1} and~\eqref{hmat1},\eqref{hmat2}, the scalar product behaves as
$$
\vec{\xi}\cdot\Big(i\,D(\vec{\omega}_t)\,\sigma(\vec{\omega}_s)-iQ(\vec{\omega}_s)\,\sigma(\vec{\omega}_s)+2i\sigma(\vec{\omega}_s)\,h\,\sigma(\vec{\omega}_s)\Big)\vec{ \xi}
=-\vec{\xi}\cdot\Big(\sigma(\vec{\omega}_s)\Big(B+2h^{(im)}\Big)\sigma(\vec{\omega}_s)\vec{\xi}\Big)=0\ ,
$$
the latter equality resulting form the fact that $\sigma(\vec{\omega}_t)\,\Big(B\,+\,2\,h^{(im)}\Big)\,\sigma(\vec{\omega}_t)$ is antisymmetric. 
\qed
\end{proof}

\subsection{Structure of the dissipative generator}
\label{subsgen}

In this section we prove various properties of the mesoscopic dynamics and its generator. We start by showing that the maps $\Phi^{\vec{\omega}}_t$ defined by Theorem~\ref{expfluct} cannot act linearly on the fluctuation algebra.
\medskip

\noindent{\bf Proposition 1.}\quad\textit{
The mesoscopic dynamics of the product of two Weyl operators satisfies
$$
\Phi^{\vec{\omega}}_t\left[W(\vec{r})\,W(\vec{s})\right]={\rm e}^{i\,\vec{s}\cdot\left(\sigma(\vec{\omega}_t)\vec{r}\right)}\,\Phi^{\vec{\omega}}_t\left[W(\vec{s})\,W(\vec{r})\right]\ ,\qquad\forall\,\vec{r}\,,\,\vec{s}\in\RR^{d^2}\ .
$$
}
\medskip

\begin{proof}
According to \eqref{meslim1a} in Definition \ref{meslimdef}, we show that, for all $\vec{r}_{1,2}\in\RR^{d^2}$,
$$
\lim_{N\to\infty}\omega_{\vec{r}_1\vec{r}_2}\left(\gamma^{(N)}_t\left[W^{(N)}_t(\vec{r})\,W^{(N)}_t(\vec{s})\right]\right)={\rm e}^{i\,\vec{s}\cdot\left(\sigma(\vec{\omega}_t)\vec{r}\right)}\,
\lim_{N\to\infty}\omega_{\vec{r}_1\vec{r}_2}\left(\gamma^{(N)}_t\left[W^{(N)}_t(\vec{s})\,W^{(N)}_t(\vec{r})\right]\right)\ .
$$
Since the exponentials $W^{(N)}_t(\vec{r})$ are unitaries, we write
$$
W^{(N)}_t(\vec{r})\,W^{(N)}_t(\vec{s})=W^{(N)}_t(\vec{s})\,\exp\left(i\,\left(W^{(N)}_t(\vec{s})\right)^\dag\,\vec{r}\cdot\vec{F}^{(N)}_t\,W^{(N)}_t(\vec{s})\right)\ ,
$$
so that Lemma \ref{lemtool1} and \ref{lemtool2} in Section~\ref{app2} yield
$$
\lim_{N\to\infty}\omega_{\vec{r}_1\vec{r}_2}\left(\gamma^{(N)}_t\left[W^{(N)}_t(\vec{r})\,W^{(N)}_t(\vec{s})\right]\right)=
\lim_{N\to\infty}\omega_{\vec{r}_1\vec{r}_2}\left(\gamma^{(N)}_t\left[W^{(N)}_t(\vec{s})\,W^{(N)}_t(\vec{r})\,{\rm e}^{\vec{s}\cdot\left(T^{(N)}\vec{r}\right)}\right]\right)\ .
$$
The result then follows by showing that $\lim_{N\to\infty}I^{(N)}(t)=0$, where
$$
I^{(N)}(t):=\left|\omega_{\vec{r}_1\vec{r}_2}\left(\gamma^{(N)}_t\left[W^{(N)}_t(\vec{s})\,W^{(N)}_t(\vec{r})\,\left({\rm e}^{\vec{s}\cdot\left(T^{(N)}\vec{r}\right)}\,-\,{\rm e}^{i\,\vec{s}\cdot\left(\sigma(\vec{\omega}_t)\vec{r}\right)}\right)\right]\right)\right|\ .
$$
By using the Cauchy-Schwartz and Kadison inequalities, one bounds $\left(I^{(N)}(t)\right)^2$ by
$$
\omega\left(\left(W^{(N)}_t(\vec{r}_2)\right)^\dag\,\gamma^{(N)}_t\left[\left({\rm e}^{\vec{s}\cdot\left(T^{(N)}\vec{r}\right)}\,-\,{\rm e}^{i\,\vec{s}\cdot\left(\sigma(\vec{\omega}_t)\vec{r}\right)}\right)^\dag\left({\rm e}^{\vec{s}\cdot\left(T^{(N)}\vec{r}\right)}\,-\,{\rm e}^{i\,\vec{s}\cdot\left(\sigma(\vec{\omega}_t)\vec{r}\right)}\right)\right]\,W^{(N)}(\vec{r}_2)\right)\ .
$$
Then, writing
\beann
{\rm e}^{\vec{s}\cdot\left(T^{(N)}\vec{r}\right)}\,-\,{\rm e}^{i\,\vec{s}\cdot\left(\sigma(\vec{\omega}_t)\vec{r}\right)}
&=&\int_0^1{\rm d}x\, \frac{{\rm d}}{{\rm d}x}\,\left({\rm e}^{x\,\vec{s}\cdot\left(T^{(N)}\vec{r}\right)}\,{\rm e}^{
i(1-x)\,\vec{s}\cdot\left(\sigma(\vec{\omega}_t)\vec{r}\right)}\right)\\
&=&
\int_0^1{\rm d}x\,{\rm e}^{x\,\vec{s}\cdot\left(T^{(N)}\vec{r}\right)}\,
{\rm e}^{i(1-x)\,\vec{s}\cdot\left(\sigma(\vec{\omega}_t)\vec{r}\right)}\,
\Big(\vec{s}\cdot\left(T^{(N)}\vec{r}\right)\,-i\,\vec{s}\cdot\left(\sigma(\vec{\omega}_t)\vec{r}\right)\Big)\ ,
\eeann
and using that $\left(\vec{s}\cdot\left(T^{(N)}\vec{r}\right)\right)^\dag=-\vec{s}\cdot\left(T^{(N)}\vec{r}\right)$ and $\left(\vec{s}\cdot\left(\sigma(\vec{\omega}_t)\vec{r}\right)\right)^*=\vec{s}\cdot\left(\sigma(\vec{\omega}_t)\vec{r}\right)$, one gets
\beann
\left(I^{(N)}(t)\right)^2&\leq&-\omega\left(\left(W^{(N)}_t(\vec{r}_2)\right)^\dag\,\gamma^{(N)}_t\left[\Big(\vec{s}\cdot\left(T^{(N)}\vec{r}\right)\,-\,i\,\vec{s}\cdot\left(\sigma(\vec{\omega}_t)\vec{r}\right)\Big)^2\right]\,W^{(N)}(\vec{r}_2)\right)\\
&=&-\sum_{\mu,\mu'\nu,\nu'=1}^{d^2}s_\mu\,r_\nu s_{\mu'} r_{\nu'}\,\omega\left(\left(W^{(N)}_t(\vec{r}_2)\right)^\dag\,\gamma^{(N)}_t\left[Z^{(N)}_{\mu\nu}(t)\,Z^{(N)}_{\mu'\nu'}(t)\right]\,W^{(N)}(\vec{r}_2)\right)\ ,
\eeann
where, by means of \eqref{tsympform1}, we set $\displaystyle Z_{\mu\nu}^{(N)}(t):=\frac{1}{N}\sum_{k=0}^{N-1}\left(\Big[v_\mu^{(k)}\,,\,v_\nu^{(k)}\Big]\,-\,i\omega_{\mu\nu}(t)\right)$.
This latter is a mean-field quantity;
then, Corollary \ref{spt} yields
\beann
&&
\lim_{N\to\infty}\omega\left(\left(W^{(N)}_t(\vec{r}_2)\right)^\dag\,\gamma^{(N)}_t\left[Z^{(N)}_{\mu\nu}(t)\,Z^{(N)}_{\mu'\nu'}(t)\right]\,W^{(N)}(\vec{r}_2)\right)=\\
&&\hskip 1cm
=\lim_{N\to\infty}\omega\left(\left(W^{(N)}_t(\vec{r}_2)\right)^\dag\,\gamma^{(N)}_t\left[Z^{(N)}_{\mu\nu}(t)\right]\,\gamma^{(N)}_t\left[Z^{(N)}_{\mu'\nu'}(t)\right]\,W^{(N)}(\vec{r}_2)\right)\ .
\eeann
Now, Corollary \ref{exg} ensures that $\gamma^{(N)}_t\left[Z^{(N)}_{\mu\nu}(t)\right]$ is a  series of mean-field quantities, uniformly convergent with respect to $N$.
As such, because of Lemma \ref{lemtool2}, it commutes with the local exponential operators $W^{(N)}_t(r)$ in the large $N$ limit, so that
\beann
&&
\hskip-.8cm
\lim_{N\to\infty}\omega\left(\left(W^{(N)}_t(\vec{r}_2)\right)^\dag\,\gamma^{(N)}_t\left[Z^{(N)}_{\mu\nu}(t)\,Z^{(N)}_{\mu'\nu'}(t)\right]\,W^{(N)}(\vec{r}_2)\right)=\\
&&\hskip 1cm
=
\lim_{N\to\infty}\omega\left(\gamma^{(N)}_t\left[Z^{(N)}_{\mu\nu}(t)\right]\,\gamma^{(N)}_t\left[Z^{(N)}_{\mu'\nu'}(t)\right]\right)\\
&&\hskip 2cm
=\lim_{N\to\infty}\omega\left(\gamma^{(N)}_t\left[Z^{(N)}_{\mu\nu}(t)\right]\right)\lim_{N\to\infty}\omega\left(\gamma^{(N)}_t\left[Z^{(N)}_{\mu'\nu'}(t)\right]\right)=0\ ,
\eeann
where the last two equalities follow from \eqref{macro1} and \eqref{macr1a}.
\qed
\end{proof}
\bigskip

The next step is the proof of Theorem~\ref{CPos} asserting that the linear extended maps $\Phi^{ext}_t$ defined in~\eqref{extweyl2} are completely positive on the direct integral von Neumann algebra $\cw^{ext}$ defined in~\eqref{dirint}.
\medskip

\noindent
{\bf Theorem 4.}\quad
\textit{
The maps $\Phi^{ext}_t$ in \eqref{extweyl2} form a one parameter family of completely positive, unital, Gaussian maps on the von Neumann algebra $\cw^{ext}$.
}
\medskip

\begin{proof}
Gaussian maps transform Gaussian states into Gaussian states.
We shall then consider states $\Omega^{ext}$ on $\cw^{ext}$ such that, according to \eqref{extstate},  $\displaystyle\Omega_{\vec{\omega}}^{ext}=\int_{\cs}^{\oplus}{\rm d}\vec{\nu}\,\delta_{\vec{\omega}}(\vec{\nu})\,\Omega_{\vec{\nu}}$, where 
$\Omega_{\vec{\omega}}$ is a Gaussian state on the Weyl algebra $\cw_{\vec{\omega}}$:
$$
\Omega_{\vec{\omega}}\Big(W_{\vec{\omega}}(\vec{r}(\vec{\omega}))\Big)=\exp\left(-\frac{1}{2}\vec{r}(\vec{\omega})\cdot\Big(\Sigma(\vec{\omega})\vec{r}(\vec{\omega})\Big)\right)\ ,
$$
with covariance matrix $\Sigma(\vec{\omega})$ such that \cite{Olivares}
\be
\label{posgauss0}
\Sigma(\vec{\omega})+\frac{i}{2}\sigma(\vec{\omega})\geq 0\ .
\ee
Using \eqref{extweyl3}, the extended dynamics turns the state $\Omega^{ext}_{\vec{\omega}}$ into
the expectation functional $\Omega^t_{\vec{\omega}}:=\Omega^{ext}_{\vec{\omega}}\circ\Phi^{\vec{\omega}}_t$ on  $\cw_{\vec{\omega}}$, such that, with
$\displaystyle
E(\vec{\omega})=\exp\left(-\frac{i}{2}\vec{r}_1(\vec{\omega})\cdot\Big(\sigma(\vec{\omega})\,\vec{r}_2(\vec{\omega})\Big)\right)$ (see \eqref{Weyl}), 
\beann
&&\hskip-.6cm
\Omega^{ext}_{\vec{\omega}}\Big(\Phi^{ext}_t\Big[W^1_{\vec{r}_1}W^1_{\vec{r}_2}\Big]\Big)=\Omega^{ext}_{\vec{\omega}}\Big(\Phi^{ext}_t\Big[E\,W^1_{\vec{r}_1+\vec{r}_2}\Big]\Big)=E(\vec{\omega}_t)\,\Omega_{\vec{\omega}}\Big(\Phi^{\vec{\omega}}_t\Big[W_{\vec{\omega}}(\vec{r}_1+\vec{r}_2)\Big]\Big)\\
&&\hskip-.3cm
={\rm e}^{-\frac{i}{2}\vec{r}_1(\vec{\omega}_t)\cdot\big(\sigma(\vec{\omega}_t)\,\vec{r}_2(\vec{\omega}_t)\big)}\,
{\rm e}^{-\frac{1}{2}(\vec{r}_1(\vec{\omega}_t)+\vec{r}_2(\vec{\omega}_t))\cdot\big(Y_t(\vec{\omega})\,(\vec{r}_1(\vec{\omega}_t)+\vec{r}_2(\vec{\omega}_t))\big)}\,\times\\
&&\hskip 2cm
\times\,\Omega_{\vec{\omega}}\Big(W_{\vec{\omega}}\big(X_t^{tr}(\vec{\omega})(\vec{r}_1(\vec{\omega}_t)+\vec{r}_2(\vec{\omega}_t))\big)\Big)\\
&&\hskip-.3cm
={\rm e}^{-\frac{i}{2}\vec{r}_1(\vec{\omega})\cdot\big(\sigma(\vec{\omega}_t)\,\vec{r}_2(\vec{\omega})\big)}\,
{\rm e}^{-\frac{1}{2}(\vec{r}_1(\vec{\omega}_t)+\vec{r}_2(\vec{\omega}_t))\cdot\big(\Sigma_t(\vec{\omega})\,(\vec{r}_1(\vec{\omega}_t)+\vec{r}_2(\vec{\omega}_t))\big)}\ ,
\eeann
where $\Sigma_t(\vec{\omega})=Y_t(\vec{\omega})\,+\,X_t(\vec{\omega})\,\Sigma(\vec{\omega})\,X_t^{tr}(\vec{\omega})$ with $Y_t(\vec{\omega})$ the matrix defined in~\eqref{matrices3}. 
Then, $\Omega^{ext}_t$ amounts to a functional on the Weyl algebra $\cw_{\vec{\omega}_t}$ determined by the symplectic matrix $\sigma(\vec{\omega}_t)$. Such a functional is positive and thus corresponds to a Gaussian state, if and only if,
according to \eqref{posgauss0}, the covariance matrix  $\Sigma_t(\vec{\omega})$
satisfies 
\bea
\label{psgauss1}
\Sigma_t(\vec{\omega})\,+\,\frac{i}{2}\sigma(\vec{\omega}_t)&=&X_t(\vec{\omega})\,\left(\Sigma(\vec{\omega})+\frac{i}{2}\sigma(\vec{\omega})\right)\,X^{tr}_t(\vec{\omega})\ +\\
\label{posgauss2}
&+&
Y_t(\vec{\omega})\,+\,\frac{i}{2}\,\Big(\,\sigma(\vec{\omega}_t)\,-\,X_t(\vec{\omega})\,\sigma(\vec{\omega})\,X^{tr}_t(\vec{\omega})\Big)
\geq 0\ .
\eea
Notice that, because of \eqref{posgauss0}, the right hand side of \eqref{psgauss1} is positive; concerning the contribution in \eqref{posgauss2}, we argue as follows.
Since $X_{t,s}(\vec{\omega})=X_t(\vec{\omega})\,X_s^{-1}(\vec{\omega})$, one rewrites
\beann
&&\hskip-1cm
Y_t(\vec{\omega})+\frac{i}{2}\Big(\sigma(\vec{\omega}_t)-X_t(\vec{\omega})\,\sigma(\vec{\omega})\,X^{tr}_t(\vec{\omega})\Big)=\int_{0}^{t}{\rm d}s\ X_{t,s}(\vec{\omega})\,\sigma(\vec{\omega}_s)\,A\,\sigma^{tr}(\vec{\omega}_s)\,X^{tr}_{t,s}(\vec{\omega})\\
&&\hskip 2cm
+\frac{i}{2}\int_{0}^{t}{\rm d}s\, \frac{{\rm d}}{{\rm d}s}\Big(X_{t,s}(\vec{\omega})\,\sigma(\vec{\omega}_s)\,X^{tr}_{t,s}(\vec{\omega})\Big)\ .
\eeann
Using \eqref{matrices1}-\eqref{matrices3} and \eqref{eqsigmadot} together with~\eqref{d0}, one then gets
\beann
&&\hskip -1cm
\frac{{\rm d}}{{\rm d}s}\Big(X_{t,s}(\vec{\omega})\,\sigma(\vec{\omega}_s)\,X^{tr}_{t,s}(\vec{\omega})\Big)=-X_{t,s}(\vec{\omega})\,Q(\vec{\omega}_s)
\sigma(\vec{\omega}_s)\,X^{tr}_{t,s}(\vec{\omega})\,-\,X_{t,s}(\vec{\omega})\,
\sigma(\vec{\omega}_s)\,Q^{tr}(\vec{\omega}_s)\,X^{tr}_{t,s}(\vec{\omega})\\
&&+\,
X_{t,s}(\vec{\omega})\frac{{\rm} d}{{\rm d}s}\sigma(\vec{\omega}_s)\,X^{tr}_{t,s}(\vec{\omega})=
-2i\,X_{t,s}(\vec{\omega})\,\sigma(\vec{\omega}_s)\,B\,\sigma^{tr}(\vec{\omega}_s)\,X^{tr}_{t,s}(\vec{\omega})\,+\\
&&+\,X_{t,s}(\vec{\omega})\,\Big(\frac{{\rm d}}{{\rm d}s} \sigma(\vec{\omega}_s)\,-\,\Big[D(\vec{\omega}_s)\,,\,\sigma(\vec{\omega}_s)\Big]\Big)\,X^{tr}_{t,s}(\vec{\omega})=-2\,i\,
X_{t,s}(\vec{\omega})\,\sigma(\vec{\omega}_s)\,B\,\sigma^{tr}(\vec{\omega}_s)\,X^{tr}_{t,s}(\vec{\omega})\ ,
\eeann
whence the positivity of the Kossakowski matrix $C=A+B$ yields
$$
Y_t(\vec{\omega})+\frac{i}{2}\Big(\sigma(\vec{\omega}_t)\,-\,X_t(\vec{\omega})\,\sigma(\vec{\omega})\,X^{tr}_t(\vec{\omega})\Big)=\int_{0}^{t}ds\ X_{t,s}(\vec{\omega})\,\sigma(\vec{\omega}_s)\,C\,\sigma^{tr}(\vec{\omega}_s)\,X^{tr}_{t,s}(\vec{\omega})\,\geq\,0\ .
$$

Unitality, $\Phi^{ext}_t[\mathbf{1}]=\mathbf{1}$, where $\mathbf{1}$ is the identity in $\cw^{ext}$, follows directly from \eqref{extweyl2}. Complete positivity of $\Phi^{ext}_t$ amounts to showing that
$$
\Phi^{ext}_t\otimes {\bf 1}_n\left[Z^\dagger Z\right]\ge0\, ,\qquad \forall n\in\mathbb{Z},
$$
where $Z$ is any operator $Z\in\cw^{ext}\otimes M_n\left(\mathbb{C}\right)$, where $M_n(\CC)$ is the algebra of $n\times n$ complex matrices. Let 
$\{E_\mu\}_{\mu=1}^{n^2}$ any fixed orthonormal basis of hermitean matrices in $M_n(\CC)$; then, a generic element $Z\in\cw^{ext}$ can be written as
$$
Z=\sum_{i,j,\mu}d_{ij\mu}\,W_{\vec{r}_i}^{f_j}\otimes E_\mu\ ,
$$
where $W_{\vec{r}_i}^{f_j}(\vec{\omega})=f_j(\vec{\omega})\,W_{\vec{\omega}}(\vec{r}_i(\vec{\omega}))$, with $f_j$ measurable functions on $\cs$, $\vec{r}_i(\vec{\omega})$ real vectors in $\RR^{d^2}$ and
$d_{ij\mu}$ suitable complex coefficients. The operator-value at $\vec{\omega}$ of the positive element $Z^\dagger Z$ reads
\beann
(Z^\dagger Z)(\vec{\omega})&=&\sum_{i,j,\mu;k,\ell,\nu}d_{ij\mu}^*\,d_{k\ell\nu}\, f^*_j(\vec{\omega})\,f_k(\vec{\omega})\,W_{\vec{\omega}}^\dagger(\vec{r}_i(\vec{\omega}))\,W_{\vec{\omega}}(\vec{r}_k(\vec{\omega}))\otimes E_\mu\,E_\nu\\
&=&
\sum_{i,j,\mu;k,\ell,\nu}d_{ij\mu}^*\,d_{k\ell\nu}\, f^*_j(\vec{\omega})\,f_k(\vec{\omega})\,E_{ik}(\vec{\omega})\,W_{\vec{\omega}}(\vec{r}_k(\vec{\omega})-\vec{r}_i(\vec{\omega}))\otimes E_\mu\,E_\nu\\
E_{ik}(\vec{\omega})&=&\exp\left(\frac{i}{2}\vec{r}_i(\vec{\omega})\,\cdot\Big(\sigma(\vec{\omega})\,\vec{r}_k(\vec{\omega})\Big)\right)\ ,
\eeann
where again use has been made of the algebraic rules \eqref{Weyl}. By means of \eqref{extweyl4}, the action of 
$\Phi^{ext}_t\otimes{\bf 1}_n$ thus gives
\bea
\nonumber
&&
\Phi^{ext}_t\otimes{\bf 1}_n\left[Z^\dagger Z\right](\vec{\omega})=\sum_{i,j,\mu;k,\ell,\nu}d_{ij\mu}^{*}\,d_{k\ell\nu}\,f^*_j(\vec{\omega}_t)\,f_k(\vec{\omega}_t)\,F^t_{ik}(\vec{\omega})\,\times\\
\label{CP1}
&&\hskip 1cm
\times\,W_{\vec{\omega}}^\dagger\big(X^{tr}_t(\vec{\omega})\vec{r}_i(\vec{\omega}_t)\big)\,W_{\vec{\omega}}\big((X^{tr}_t(\vec{\omega})\vec{r}_k(\vec{\omega}_t)\big)\otimes E_\mu E_\nu\\
\label{CP2}
&&
F^t_{ik}(\vec{\omega})={\rm e}^{\frac{i}{2}\vec{r}_i(\vec{\omega}_t)\cdot\Big(\hat{\sigma}_t(\vec{\omega})\,\vec{r}_k(\vec{\omega}_t)\Big)}\,
{\rm e}^{-\frac{1}{2}\Big((\vec{r}_i(\vec{\omega}_t)-\vec{r}_k(\vec{\omega}_t))\cdot\Big(Y_t(\vec{\omega})(\vec{r}_i(\vec{\omega}_t)-\vec{r}_k(\vec{\omega}_t))\Big)}\\
\label{CP3}
&&
\hat{\sigma}_t(\vec{\omega})=\sigma(\vec{\omega}_t)\,-\,X_t(\vec{\omega})\,\sigma(\vec{\omega})\,X^{tr}_t(\vec{\omega})\ .
\eea
The function \eqref{CP2} can be interpreted as the expectation of the product of two
Weyl operators $V_{\vec{\omega}}(\vec{r}_{i,k}(\vec{\omega}_t))$ satisfying the Weyl algebraic rules 
$$
V_{\vec{\omega}}(\vec{r}_i(\vec{\omega}_t))\,V_{\vec{\omega}}(\vec{r}_k(\vec{\omega}_t))=V_{\vec{\omega}}(\vec{r}_i(\vec{\omega}_t)+\vec{r}_k(\vec{\omega}_t))\,\exp\left(-\frac{i}{2}\vec{r}_i(\vec{\omega}_t)\cdot\Big(\hat\sigma_t(\vec{\omega})\,\vec{r}_k(\vec{\omega}_t)\Big)\right)
$$
with symplectic matrix, that is real and anti-symmetric, $\hat\sigma_t(\vec{\omega})$ given by \eqref{CP3}, the expectation being defined by the functional $\varphi_t$ 
\be
\label{CP4}
\varphi_t\left(V_{\vec{\omega}}(\vec{r}(\vec{\omega}_tß))\right)=\exp\left(-\frac{1}{2}\vec{r}(\vec{\omega}_t)\cdot\Big(Y_t(\vec{\omega})\vec{r}(\vec{\omega}_t)\Big)\right)
\ee
acting on the Weyl algebra $\cv_{\vec{\omega}}$ generated by the $V_{\vec{\omega}}(\vec{r}(\vec{\omega}))$. 

Letting $\displaystyle Z^{\vec{\omega}}_{ij\mu}(t):=d_{ij\mu}\,f_j(\vec{\omega}_t)\,W_{\vec{\omega}}\big(X^{tr}_t(\vec{\omega})\vec{r}_i(\vec{\omega}_t)\big)\otimes E_\mu$\ ,
one can then write
$$
\Phi^{ext}_t\otimes{\bf 1}_n\left[Z^\dagger Z\right](\vec{\omega})
=\sum_{i,j,\mu;k,\ell,\nu}(Z^{\vec{\omega}}_{ij\mu})^\dagger(t)\, Z^{\vec{\omega}}_{k\ell\nu}(t)\, 
\varphi_t\Big((V_{\vec{\omega}}(\vec{r}_i(\vec{\omega}_t)))^\dag\,V_{\vec{\omega}}(\vec{r}_k(\vec{\omega}_t))\Big)\ .
$$
This is a positive operator in $\cw_{\vec{\omega}}\otimes M_n(\CC)$ if the expectation functional $\varphi_t$ on $\cv$ is positive, namely, according to \eqref{CP4}, if $\varphi_t$ amounts to a Gaussian state. This latter property is equivalent to having
$$
Y_t(\vec{\omega})\,+\,\frac{i}{2}\hat{\sigma}_t(\vec{\omega})=Y_t(\vec{\omega})\,+\,\frac{i}{2}\Big(\sigma(\vec{\omega}_t)\,-\,X_t(\vec{\omega})\,\sigma(\vec{\omega})\,X^{tr}_t(\vec{\omega})\Big)\geq0
$$
which has already been proved. 
\qed
\end{proof}
\bigskip

Finally, we prove Theorem~\ref{expfluct} which provides the hybrid form of the generator of the semigroup of completely positive extended maps $\{\Phi^{ext}_t\}_{t\geq 0}$.
\bigskip

\noindent
{\bf Theorem 5.}\quad
\textit{
The extended dissipative dynamics $\Phi^{ext}_t$ of quantum fluctuations has a generator of the form
$\displaystyle\LL^{ext}=\int_{\cs}^\oplus{\rm d}\vec{\omega}\,\LL_{\vec{\omega}}$.
Every $\vec{\omega}$-component consists of four different contributions,
$\LL_{\vec{\omega}}=\LL^{drift}_{\vec{\omega}}+\LL^{cc}_{\vec{\omega}}+\LL^{cq}_{\vec{\omega}}+\LL^{qq}_{\vec{\omega}}$: a drift term
$$
\LL^{drift}_{\vec{\omega}}\Big[W^f_{\vec{r}}(\vec{\omega})\Big]=\dot{\vec{\omega}}\cdot\partial_{\vec{\omega}}\,W^f_{\vec{r}}(\vec{\omega})\ ,
$$ 
a differential operator involving the classical degrees of freedom
$G_\mu(\vec{\omega})$, $\mu=1,2,\ldots,d_0(\vec{\omega})$,
$$
\LL^{cc}_{\vec{\omega}}\Big[W^f_{\vec{r}}(\vec{\omega})\Big]=\sum_{\mu,\nu=1}^{d_0(\vec{\omega})}\,D^{00}_{\mu\nu}(\vec{\omega})G_\nu(\vec{\omega})\,\frac{\partial\,W^f_{\vec{r}}(\vec{\omega})}{\partial G_\mu(\vec{\omega})}
$$
a mixed classical-quantum term
$$
\LL^{cq}_{\vec{\omega}}\Big[W^f_{\vec{r}}(\vec{\omega})\Big]=\sum_{\mu=1}^{d_0(\vec{\omega})}\sum_{\nu=d_0(\vec{\omega})+1}^{d^2}\frac{D^{01}_{\mu\nu}(\vec{\omega})}{2}\,\left\{G_\nu(\vec{\omega})\,,\,\frac{\partial\,W^f_{\vec{r}}(\vec{\omega})}{\partial G_\mu(\vec{\omega})}\right\}\ ,
$$
and a purely quantum term 
\beann
&&\hskip-1cm
\LL^{qq}_{\vec{\omega}}\Big[W^f_{\vec{r}}(\vec{\omega})\Big]=
i\,\sum_{\mu=1}^{d_0(\vec{\omega})}\sum_{\nu=d_0(\vec{\omega}))+1}^{d^2}\Big(H^{01}_{\mu\nu}(\vec{\omega})+H^{10}_{\nu\mu}(\vec{\omega})\Big)\,G_\mu(\vec{\omega})\,\Big[G_\nu(\vec{\omega})\,,\,W^f_{\vec{r}}(\vec{\omega})\Big]\\
&&
+i\,\sum_{\mu,\nu=d_0(\vec{\omega})+1}^{d^2}H^{11}_{\mu\nu}(\vec{\omega})\,\Big[G_\mu(\vec{\omega})\,G_\nu(\vec{\omega})\,,\,W^f_{\vec{r}}(\vec{\omega})\Big]\\
\hskip-1cm
&&
+\sum_{\mu,\nu=d_0(\vec{\omega})+1}^{d^2}K^{11}_{\mu\nu}(\vec{\omega})\Big(G_\mu(\vec{\omega})\,W^f_{\vec{r}}(\vec{\omega})\,G_\nu(\vec{\omega})\,-\,\frac{1}{2}\left\{G_\mu(\vec{\omega})\,G_\nu(\vec{\omega})\,,\,W^f_{\vec{r}}(\vec{\omega})\right\}\Big)\ ,
\eeann
with the $d_1(\vec{\omega})\times d_1(\vec{\omega})$ matrix of coefficients $H^{11}(\vec{\omega})$ given by 
$$
H^{11}(\vec{\omega})=(h^{(re)}(\vec{\omega}))^{11}\,+\,\frac{1}{4}\left\{\Big(\sigma^{11}(\vec{\omega})\Big)^{-1}\,,\,D^{11}(\vec{\omega})\right\}\ ,
$$
where $(h^{(re)}(\vec{\omega}))^{11}$ is the $11$-component of the matrix $h^{(re)}$ in~\eqref{hmat1} rotated by $R(\vec{\omega})$ as in~\eqref{01dec}.\\
\indent
Further, the $d_0(\vec{\omega})\times d_1(\vec{\omega})$ matrix $H^{01}(\vec{\omega})$, respectively the $d_1(\vec{\omega})\times d_0(\vec{\omega})$ matrix $H^{01}(\vec{\omega})$ read
\beann
H^{10}(\vec{\omega})&=&\frac{1}{2}(\sigma^{11}(\vec{\omega}))^{-1}\,D^{10}(\vec{\omega})\,-\frac{i}{2}\,B^{10}(\vec{\omega})\,+\,(h^{(re)(\vec{\omega}))^{10}}\\ 
H^{01}(\vec{\omega})&=&\frac{1}{2}D^{01}(\vec{\omega})\,
(\sigma^{11}(\vec{\omega}))^{-1}\,+\frac{i}{2}\,B^{01}(\vec{\omega})\,+\,(h^{(re)(\vec{\omega}))^{10}}\ ,
\eeann
where $(h^{(re)}(\vec{\omega}))^{10}$ and $(h^{(re)}(\vec{\omega}))^{10}$ are the $10$ and $01$ components of the matrix $h^{(re)}$ in~\eqref{hmat1} rotated by $R(\vec{\omega})$ as in~\eqref{01dec}.\\
\indent
Finally, $K^{11}(\vec{\omega})$ amounts to 
$$
K^{11}(\vec{\omega})\,=\,A^{11}(\vec{\omega})\,+\,B^{11}(\vec{\omega})\,+\,\frac{i}{2}\left[(\sigma^{11}(\vec{\omega}))^{-1}\,,\,D^{11}(\vec{\omega})\right]\ .
$$
}
\medskip

\begin{proof}
From the expression~\eqref{extweyl2} one gets
$$
\left.\frac{{\rm d}}{{\rm d}t}\Phi_t^{ext}[W^f_{\vec{r}}]\right|_{t=0}(\vec{\omega})=
\Big(\dot{\vec{\omega}}\cdot\partial_{\vec{\omega}}f(\vec{\omega})\Big)\,W^1_{\vec{r}}(\vec{\omega})\,+\,
f(\vec{\omega})\,\left.\frac{{\rm d}}{{\rm d}t}\Phi_t^{\vec{\omega}}[W^1_{\vec{r}}]\right|_{t=0}(\vec{\omega})\ ,
$$ 
where $\dot{\vec{\omega}}=D(\vec{\omega})\vec{\omega}$. Notice that $W^1_{\vec{r}}(\vec{\omega})=\exp\Big(i\vec{r}(\vec{\omega})\cdot\vec{F}(\vec{\omega})\Big)$. Then, 
from~\eqref{lastaidd0} and~\eqref{extweyl3}, Lemma~\ref{tool3} and the algebraic relations~\eqref{COMSIGMA1} which make higher commutators vanish, one derives 
\beann
&&\hskip-.5cm
\left.\frac{{\rm d}}{{\rm d}t}\Phi_t^{\vec{\omega}}[W^1_{\vec{r}}]\right|_{t=0}(\vec{\omega})=-\frac{1}{2}\,\vec{r}(\vec{\omega})\cdot\Big(\sigma(\vec{\omega})\,A\,\sigma^{tr}(\vec{\omega})\vec{r}(\vec{\omega})\Big)\,W^1_{\vec{r}}(\vec{\omega})\,+\,
\left.\frac{{\rm d}}{{\rm d}t}W_{\vec{\omega}}\Big(X^{tr}_t(\vec{\omega})\vec{r}(\vec{\omega}_t)\Big)\right|_{t=0}\\
&&\hskip-.5cm
\left.\frac{{\rm d}}{{\rm d}t}W_{\vec{\omega}}\Big(X^{tr}_t(\vec{\omega})\vec{r}(\vec{\omega}_t)\Big)\right|_{t=0}=\Bigg(i(Q^{tr}(\vec{\omega})\vec{r}(\vec{\omega}))\cdot\vec{F}(\vec{\omega})-\frac{1}{2}\Big[\vec{r}(\vec{\omega})\cdot\vec{F}(\vec{\omega})\,,\,(Q^{tr}(\vec{\omega})\vec{r}(\vec{\omega}))\cdot\vec{F}(\vec{\omega})\Big]\\
&&\hskip 1cm
+\,\Big(i(\dot{\vec{\omega}}\cdot\partial_{\vec{\omega}}\vec{r}(\vec{\omega}))\cdot\vec{F}(\vec{\omega})-\frac{1}{2}\Big[\vec{r}(\vec{\omega})\cdot\vec{F}(\vec{\omega})\,,\,(\dot{\vec{\omega}}\cdot\partial_{\vec{\omega}}\vec{r}(\vec{\omega}))\cdot\vec{F}(\vec{\omega})\Big]\Bigg)\,W^1_{\vec{r}}(\vec{\omega})\ .
\eeann
The proof of Lemma~\ref{tool3} holds also if one acts on $W^1_{\vec{r}}(\vec{\omega})$ with $\dot{\vec{\omega}}\cdot\partial_{\vec{\omega}}$, so that
$$
\Big(i(\dot{\vec{\omega}}\cdot\partial_{\vec{\omega}}\vec{r}(\vec{\omega}))\cdot\vec{F}(\vec{\omega})-\frac{1}{2}\Big[\vec{r}(\vec{\omega})\cdot\vec{F}(\vec{\omega})\,,\,(\dot{\vec{\omega}}\cdot\partial_{\vec{\omega}}\vec{r}(\vec{\omega}))\cdot\vec{F}(\vec{\omega})\Big]\Big)\,W^1_{\vec{r}}(\vec{\omega})=\dot{\vec{\omega}}\cdot\partial_{\vec{\omega}}W^1_{\vec{r}}(\vec{\omega})\ .
$$
Therefore, when evaluated at $\vec{\omega}$, the time-derivative of the dynamics at $t=0$, together with the expression $Q(\vec{\omega})=-i\sigma(\vec{\omega})\,\widetilde{B}\,+\,D(\vec{\omega})$ (see~\eqref{matrices1}-\eqref{matrices3}), yields
\beann
\hskip-1cm
\left.\frac{{\rm d}}{{\rm d}t}\Phi_t^{ext}[W^f_r]\right|_{t=0}(\vec{\omega})&=&
\dot{\vec{\omega}}\cdot\partial_{\vec{\omega}}\,W^f_r(\vec{\omega})\,+\,\vec{r}(\omega)\cdot\Big(\Big(\sigma(\vec{\omega})\,\widetilde{B}\,+i\,D(\vec{\omega})\Big)\vec{F}(\vec{\omega})\Big)\, W_r^f(\vec{\omega})\\
\hskip-1cm
&-&
\frac{1}{2}\,\vec{r}(\vec{\omega})\cdot\Big(\Big(\sigma(\vec{\omega})\,\Big(A+\,\widetilde{B}\Big)\,\sigma^{tr}(\vec{\omega})+iD(\vec{\omega})\sigma^{tr}(\vec{\omega})\Big)\,\vec{r}(\vec{\omega})\Big)\,W^f_r(\vec{\omega})\ .
\eeann
Using that $\sigma^{tr}(\vec{\omega})=-\sigma(\vec{\omega})$, by means of the rotation matrix $R(\vec{\omega})$ in~\eqref{invsigma}, of the decomposition~\eqref{01dec} and with the notation of Remark~\ref{reminv}, $\vec{s}(\vec{\omega})=R(\vec{\omega})\vec{r}$, one finally gets
\bea
\label{extgen1a}
\left.\frac{{\rm d}}{{\rm d}t}\Phi_t^{ext}[W^f_r]\right|_{t=0}(\vec{\omega})&=&
\dot{\vec{\omega}}\cdot\partial_{\vec{\omega}}\,W^f_{\vec{r}}(\vec{\omega})\\
\label{extgen1b}
&&\hskip-5cm
+\,\vec{s}_0(\vec{\omega})\cdot\Big(i\,D^{00}(\vec{\omega})\vec{G}_0(\vec{\omega})+i\,D^{01}(\vec{\omega})\vec{G}_1(\vec{\omega})\Big)\,W^f_{\vec{r}}(\vec{\omega})\\
\label{extgen1c}
&&\hskip-5cm
+\,\vec{s}_1(\vec{\omega})\cdot\Big(\Big(i\,D^{10}(\vec{\omega})+\widetilde{\sigma}^{11}(\vec{\omega})\,\widetilde{B}^{10}(\vec{\omega})\Big)\vec{G}_0(\vec{\omega})\Big)\,W^f_{\vec{r}}(\vec{\omega})\\
\label{extgen1d}
&&\hskip-5cm
+\,\vec{s}_1(\vec{\omega})\cdot\Big(\Big(i\,D^{11}(\vec{\omega})+\widetilde{\sigma}^{11}(\vec{\omega})\,\widetilde{B}^{11}(\vec{\omega})\Big)\vec{G}_1(\vec{\omega})\Big)\,W_{\vec{r}}^f(\vec{\omega})\\
\label{extgen1e}
&&\hskip-5cm
+\,\frac{1}{2}\,\vec{s}_0(\vec{\omega})\cdot\Big(i\,D^{01}(\vec{\omega})\widetilde{\sigma}^{11}(\vec{\omega})\vec{s}_1(\vec{\omega})\Big)\Big)\,W^f_{\vec{r}}(\vec{\omega})
\\
\label{extgen1f}
&&\hskip-5cm
+\,\frac{1}{2}\,\vec{s}_1(\vec{\omega})\cdot\Big(\Big(\widetilde{\sigma}^{11}(\vec{\omega})\,\Big(A^{11}(\vec{\omega})\,+\,\widetilde{B}^{11}(\vec{\omega})\Big)\,\widetilde{\sigma}^{11}(\vec{\omega})\,+i\,D^{11}(\vec{\omega})\,\widetilde{\sigma}^{11}(\vec{\omega})\Big)\,\vec{s}_1(\vec{\omega})\Big)\,W^f_{\vec{r}}(\vec{\omega})\ .
\eea
Let us try to write the component $\LL_{\vec{\omega}}$ of the generator $\LL^{ext}$ in the customary Lindblad form
\beann
\hskip-1cm
\LL_{\vec{\omega}}[W^1_{\vec{r}}(\vec{\omega})]&=&i\Big[\sum_{\mu,\nu=1}^{d^2}H_{\mu\nu}(\vec{\omega})\,G_\mu(\vec{\omega})\,G_\nu(\vec{\omega})\,,\,W^1_{\vec{r}}(\vec{\omega})\Big]\\
\hskip-1cm
&+&\sum_{\mu,\nu=1}^{d^2}K_{\mu\nu}(\vec{\omega})\,\Big(G_\mu(\vec{\omega})\,W^1_{\vec{r}}(\vec{\omega})\,G_\nu(\vec{\omega})\,-\,\frac{1}{2}\Big\{G_\mu(\vec{\omega})\,G_\nu(\vec{\omega})\,,\,W^1_{\vec{r}}(\vec{\omega})\Big\}\Big)\ ,
\eeann 
where the $d^2\times d^2$ matrices $H(\vec{\omega})$ and $K(\vec{\omega})$ are both hermitian and the operators $G_\mu(\vec{\omega})$ are those appearing in the decomposition~\eqref{split} of the Weyl operators into classical and quantum contributions.

By decomposing the matrix $K(\vec{\omega})$ and $H(\vec{\omega})$ as in~\eqref{01dec}, since the components of 
$\vec{G}_0(\vec{\omega})$ commute with all the others, there are no contributions to $\LL_{\vec{\omega}}$ from either $H^{00}(\vec{\omega})$ or $K^{00}(\vec{\omega})$, while those from $K^{01}(\vec{\omega})$ and $K^{10}(\vec{\omega})$ can be put together with the contributions from $H^{01}(\vec{\omega})$ and $H^{10}(\vec{\omega})$ in the hamiltonian matrix. Thus, one can, without restriction, set 
$$
H(\vec{\omega})=\begin{pmatrix}0&H^{01}(\vec{\omega})\cr H^{10}(\vec{\omega})& H^{11}(\vec{\omega})\end{pmatrix}\ ,\ K(\vec{\omega})=\begin{pmatrix}0&0\cr 0& K^{11}(\vec{\omega})\end{pmatrix}\ .
$$
Then, the relations~\eqref{COMSIGMA1} yield 
$\displaystyle
W^1_{\vec{r}}(\vec{\omega})\,\vec{G}_{\vec{\omega}}\,\Big(W^1_{\vec{r}}(\vec{\omega})\Big)^\dag=\vec{G}_{\vec{\omega}}\,+\,\sigma(\vec{\omega})\vec{s}(\vec{\omega})$, whence
\bea
\label{extgen2a}
&&\hskip -1cm
\LL_{\vec{\omega}}[W^1_{\vec{r}}(\vec{\omega})]=
\vec{s}_1(\vec{\omega})\cdot\Big(2\,i\,\widetilde{\sigma}^{11}(\vec{\omega})\,H^{10}(\vec{\omega})\vec{G}_0(\vec{\omega})\Big)\,W^1_{\vec{r}}(\vec{\omega})\\
\label{extgen2b}
&&\hskip-.5cm
+\vec{s}_1(\vec{\omega})\cdot\Big(\widetilde{\sigma}^{11}(\vec{\omega})\,\Big(2\,i\,H^{11}(\vec{\omega})\,+\,\frac{K^{11}(\vec{\omega})-(K^{11}(\vec{\omega}))^{tr}}{2}\Big)\,\vec{G}_1(\vec{\omega})\Big)\,W^1_{\vec{r}}(\vec{\omega})\\
\label{extgen2c}
&&\hskip-.5cm
+\vec{s}_1(\vec{\omega})\cdot\left(\widetilde{\sigma}^{11}(\vec{\omega})\,\Big(i\,H^{11}(\vec{\omega})\,
\,+\,\frac{1}{2}\,K^{11}(\vec{\omega})\Big)\,\widetilde{\sigma}^{11}(\vec{\omega})\,\vec{s}_1(\vec{\omega})\right)\,W^1_{\vec{r}}(\vec{\omega})\ .
\eea
From comparing equations~\eqref{extgen1a}-\eqref{extgen1f} and~\eqref{extgen2a}-\eqref{extgen2b}, one finds
\bea
\label{extgen3a}
&&\hskip-2cm
2\,i\,\widetilde{\sigma}^{11}(\vec{\omega})\,H^{10}(\vec{\omega})=iD^{10}(\vec{\omega})\,+\,\widetilde{\sigma}^{11}(\vec{\omega})\,\widetilde{B}^{10}(\vec{\omega})\\
\label{extgen3b}
&&\hskip-2cm
2\,i\,\widetilde{\sigma}^{11}(\vec{\omega})\,H^{11}(\vec{\omega})\,+\,\widetilde{\sigma}^{11}(\vec{\omega})\frac{K^{11}(\vec{\omega})-(K^{11}(\vec{\omega}))^{tr}}{2}=i\,D^{11}(\vec{\omega})\,+\,\widetilde{\sigma}^{11}(\vec{\omega})\,\widetilde{B}^{11}(\vec{\omega})\\
\nonumber
&&\hskip-2cm
\widetilde{\sigma}^{11}(\vec{\omega})\Big(i\,H^{11}(\vec{\omega})+\frac{1}{2}\,K^{11}(\vec{\omega})\Big)\widetilde{\sigma}^{11}(\vec{\omega})=\frac{1}{2}\,\widetilde{\sigma}^{11}(\vec{\omega})\,\Big(A^{11}(\vec{\omega})\,+\,\widetilde{B}^{11}(\vec{\omega})\Big)\,\widetilde{\sigma}^{11}(\vec{\omega})\\
\nonumber
&&\hskip 6cm
+\,\frac{i}{2}\,D^{11}(\vec{\omega})\,\widetilde{\sigma}^{11}(\vec{\omega})=\\
\nonumber
&&\hskip-2cm
=\frac{1}{2}\,\widetilde{\sigma}^{11}(\vec{\omega})\,\Big(A^{11}(\vec{\omega})\,+\,B^{11}(\vec{\omega})\Big)\,\widetilde{\sigma}^{11}(\vec{\omega})\,+\,i\,\widetilde{\sigma}^{11}(\vec{\omega})\,(h^{(re)}(\vec{\omega}))^{11}\,\widetilde{\sigma}^{11}(\vec{\omega})\\
\label{extgen3c}
&&\hskip6cm
+\,\frac{i}{2}\,D^{11}(\vec{\omega})\,\widetilde{\sigma}^{11}(\vec{\omega})
\ ,
\eea
where, in the last expression, use has been made of~\eqref{hmat2}, namely of $\widetilde{B}=B\,+\,2\,i\,h^{(re)}$.

The relations $D^\dag(\vec{\omega})=-D(\vec{\omega})$ (see~\eqref{remD}) and 
$\sigma^\dag(\vec{\omega})=-\sigma(\vec{\omega})$
also hold for the matrices diagonal blocks with respect to the decomposition~\eqref{01dec} after rotating them by $R(\vec{\omega})$ . Thus, taking the hermitean conjugate of both sides of~\eqref{extgen3c} yields
\beann
&&\hskip-.5cm
\widetilde{\sigma}^{11}(\vec{\omega})\Big(-i\,H^{11}(\vec{\omega})+\frac{1}{2}\,K^{11}(\vec{\omega})\Big)\widetilde{\sigma}^{11}(\vec{\omega})=\\
&&\hskip-.5cm
=\frac{1}{2}\,\widetilde{\sigma}^{11}(\vec{\omega})\,\Big(A^{11}(\vec{\omega})\,+\,B^{11}(\vec{\omega})\Big)\,\widetilde{\sigma}^{11}(\vec{\omega})\,-\,i\,\widetilde{\sigma}^{11}(\vec{\omega})\,(h^{(re)}(\vec{\omega}))^{11}\,\widetilde{\sigma}^{11}(\vec{\omega})\,-\,\frac{i}{2}\,\widetilde{\sigma}^{11}(\vec{\omega})\,D^{11}(\vec{\omega})\ .
\eeann
By adding and subtracting the last equality from~\eqref{extgen3c} itself and by multiplying both sides of the resulting equalities by the inverse of $\widetilde{\sigma}^{11}(\vec{\omega})$ one gets
\beann
H^{11}(\vec{\omega})&=&(h^{(re)}(\vec{\omega}))^{11}\,+\,\frac{1}{4}
\Big\{(\widetilde{\sigma}^{11}(\vec{\omega}))^{-1}\,,\,D^{11}(\vec{\omega})\Big\}\ ,\\
K^{11}(\vec{\omega})&=&C^{11}(\vec{\omega})\,+\frac{i}{2}\Big[(\widetilde{\sigma}^{11}(\vec{\omega}))^{-1}\,,\,D^{11}(\vec{\omega})\Big]\ ,
\eeann
where $C(\vec{\omega})=R(\vec{\omega})\,(A+B)\,R^\dag(\vec{\omega})$ is the positive semi-definite Kossakowski matrix in the microscopic generator~\eqref{mflind1a} rotated into the representation~\eqref{01dec} of the symplectic matrix. These $d_1(\vec{\omega})\times d_1(\vec{\omega})$ hermitian matrices also solve~\eqref{extgen3b}, while from~\eqref{extgen3a},
\beann
H^{10}(\vec{\omega})&=&\frac{1}{2}(\widetilde{\sigma}^{11}(\vec{\omega}))^{-1}\,D^{10}(\vec{\omega})\,-\frac{i}{2}\,B^{10}(\vec{\omega})\,+\,(h^{(re)}(\vec{\omega}))^{10}\ ,\\ 
H^{01}(\vec{\omega})&=&\frac{1}{2}D^{01}(\vec{\omega})\,
(\widetilde{\sigma}^{11}(\vec{\omega}))^{-1}\,+\frac{i}{2}\,\widetilde{B}^{01}(\vec{\omega})\,+\,(h^{(re)}(\vec{\omega}))^{01}\ .
\eeann
The still unmatched terms in~\eqref{extgen1b} and~\eqref{extgen1e} can only be recovered by acting on the Weyl operators in a way that involves both the commuting degrees of freedom represented by the first $d_0(\vec{\omega})$ components of $\vec{G}_0(\vec{\omega})$ and the remaining non-commuting ones. Then,
\beann
&&\vec{s}_0(\vec{\omega})\cdot\Big(i\,D^{00}(\vec{\omega})\vec{G}_0(\vec{\omega})+i\,D^{01}(\vec{\omega})\vec{G}_1(\vec{\omega})\Big)\,W^f_{\vec{r}}(\vec{\omega})=\\
&&\hskip 1cm
=\sum_{\mu=1}^{d_0(\vec{\omega})}\left(\sum_{\nu=1}^{d_0(\vec{\omega})}D^{00}_{\mu\nu}(\vec{\omega})\,G_\nu(\vec{\omega})\,+\,\sum_{\nu=d_0(\vec{\omega})+1}^{d^2}D^{01}_{\mu\nu}(\vec{\omega})\,G_\nu(\vec{\omega})\right)\frac{\partial\,W^f_{\vec{r}}(\vec{\omega})
}{\partial G_\mu(\vec{\omega})}\\
&&
\frac{1}{2}\vec{s}_0(\vec{\omega})\cdot\Big(i\,D^{01}(\vec{\omega})\,\widetilde{\sigma}^{11}(\vec{\omega})\vec{s}_1(\vec{\omega})\Big)\Big)\,W^f_r(\vec{\omega})=\\
&&\hskip 1cm
=-\frac{1}{2}\sum_{\mu=1}^{d_0(\vec{\omega})}\sum_{\nu=d_0(\vec{\omega})+1}^{d^2}D^{01}(\vec{\omega})\,\Big[G_\nu(\vec{\omega})\,,\,\frac{\partial\,W^f_{\vec{r}}(\vec{\omega})}{\partial G_\mu(\vec{\omega})}\Big]\ .
\eeann
Summing the right hand sides of the above equalities yields the classical contribution to the generator, $\LL_{\vec{\omega}}^{cc}$  in~\eqref{extgenth1b}, respectively  the mixed classical-quantum one, $\LL_{\vec{\omega}}^{cq}$ in~\eqref{extgenth1c}.
\qed
\end{proof}

\section{Appendix A}
\label{app0}

In the case of a clustering state $\omega$, one can then consider the large $N$ limit of 
$\omega\left(b^\dagger X^{(N)}\, c\right)$ where $b,c\in\ca$, obtaining
\be
\lim_{N\to\infty}\omega\left(b^\dagger X^{(N)}\, c\right)=\omega(b^\dag c)\,\omega(x)\ .
\label{MET}
\ee
Indeed, for any integer $N_0<N$ one can write:
$$
\lim_{N\to\infty} \omega\left(b^\dagger X^{(N)}\, c\right)=
\lim_{N\to\infty} \omega\Bigg( b^\dagger \bigg( \frac{1}{N} \sum_{k=0}^{N_0} x^{(k)} 
+ \frac{1}{N} \sum_{k=N_0+1}^{N-1} x^{(k)}\bigg)\, c\Bigg)\ .
$$
While the first contribution at the r.h.s. vanishes, concerning the second term we argue as follows.
Since strictly local operators are norm dense in $\mathcal{A}$, without
loss of generality one can assume $c$ to have support within $[-N_0,N_0]$, so that it commutes with $\sum_{k=N_0+1}^{N-1} x^{(k)}$. Using the clustering property (\ref{clustates})
one immediately gets the result (\ref{MET}). This means that, in the so-called weak operator topology,
{\it i.e.} under the state average, the large $N$ limit of $X^{(N)}$ is a scalar multiple of the identity operator:
$$
w-\lim_{N\to\infty} X^{(N)} = \omega(x)\, {\bf 1}\ .
$$
The relation~\eqref{macro1bis} can be proved as follows: because of definition \eqref{macro}, it is equivalent to 
$$
\lim_{N\to\infty}\omega\bigg(a^\dag\,\big(X^{(N)}-\omega(x)\big)\big(Y^{(N)}-\omega(y)\big)\,b\bigg)=0
$$
for all $a,b\in\ca$. Set
$$
\widetilde{X}_N=\frac{1}{N}\sum_{k=0}^{N-1}\underbrace{\bigg(x^{(k)}-\omega(x)\bigg)}_{\widetilde{x}^{(k)}}\ ,\quad \widetilde{Y}_N=\frac{1}{N}\sum_{k=0}^{N-1}\underbrace{\bigg(y^{(k)}-\omega(y)\bigg)}_{\widetilde{y}^{(k)}}\ ,
$$
so that $\omega(\widetilde{x}^{(k)})=\omega(\widetilde{x})=0$, $\omega\left(\widetilde{X}_N\right)=0$ and similarly for $\widetilde{y}$, $\widetilde{Y}_N$.
Then, as shown in the main text for a single variable, the quasi-locality of $a,b$ and the clustering properties of the state yield:
$$
\lim_{N\to\infty}\omega\bigg(a^\dag\,\big(X^{(N)}-\omega(x)\big)\big(Y^{(N)}-\omega(y)\big)\,b\bigg)=\omega(a^\dag b)\lim_{N\to\infty}\omega\bigg(\widetilde{X}_N\widetilde{Y}_N\bigg)\ .
$$
Further, one can write:
$$
\omega\bigg(\widetilde{X}_N\widetilde{Y}_N\bigg)=\frac{1}{N^2}\sum_{k=0}^{N-1}
\omega\bigg(\widetilde{x}^{(k)}\widetilde{y}^{(k)}\bigg)\,+\,
\frac{1}{N^2}\sum_{k\neq\ell=0}^{N-1}
\omega\bigg(\widetilde{x}^{(k)}\widetilde{y}^{(\ell)}\bigg)\ .
$$
Since $\omega$ is translation-invariant, the first term vanishes as $\omega\big(\widetilde{x}\widetilde{y}\big)/N$ when $N\to\infty$.
Moreover, thanks to the clustering property (\ref{clustates}), for any small $\epsilon>0$, 
there exists an integer $N_\epsilon$,
such that for $|k-\ell|^2>N_\epsilon$ one has:
$$
\left|\omega(\big(\widetilde{x}^{(k)}\widetilde{y}^{(\ell)}\big)-\omega(\widetilde{x})\,\omega(\widetilde{y})\right|=
\left|\omega(\big(\widetilde{x}^{(k)}\widetilde{y}^{(\ell)}\big)\right|\leq\epsilon\ .
$$
Then, using this result, one can finally write:
\beann
\left|\frac{1}{N^2}\sum_{k\neq\ell=0}^{N-1}
\omega\bigg(\widetilde{x}^{(k)}\widetilde{y}^{(\ell)}\bigg)\right|&\leq&
\frac{1}{N^2}\sum_{0<|k-\ell|\leq N_\epsilon}\left|\omega\bigg(\widetilde{x}^{(k)}\widetilde{y}^{(\ell)}\bigg)\right|\\
&+&\frac{1}{N^2}\sum_{|k-\ell|> N_\epsilon}\left|\omega\bigg(\widetilde{x}^{(k)}\widetilde{y}^{(\ell)}\bigg)\right|\\
&\leq&4\frac{2N_\epsilon+1}{N}\,\|x\|\,\|y\|\,+\,\epsilon\ ,
\eeann
so that, in the large $N$ limit, the relation \eqref{macro1bis} is indeed satisfied.
Notice that~\eqref{macro1bis} entails that, in the GNS representation,
\bea
\nonumber
&&
\lim_{N\to\infty}\omega\bigg(a^\dag \big(X^{(N)}-\omega(x)\big)^\dag\,\big(X^{(N)}-\omega(x)\big)\,a\bigg)=\\
\label{macro2}
&&\hskip 2cm
=\lim_{N\to\infty}\left\|\pi_\omega\big(X^{(N)}-\omega(x)\big)\vert\Psi_a\rangle\right\|^2=0\ ,
\eea
for all $a\in\ca$. Namely, mean-field spin observables converge to their expectations with respect to $\omega$ in the strong operator topology on the GNS Hilbert space $\mathbb{H}_\omega$.\\
For what concerns \eqref{macro1}, notice that 
\begin{equation*}
\begin{split}
\Big|\omega\left(X^{(N)}Y^{(N)}\right)-\omega(x)\omega(y)\Big|\le\frac{1}{N^2}\sum_{k,h=0}^{N-1}\Big|\omega\left(x^{(k)}y^{(h)}\right)-\omega(x)\omega(y)\Big|=\\
=\frac{1}{N^2}\sum_{k=0}^{N-1}\sum_{\ell=-k}^{N-1-k}\Big|\omega\left(x^{(0)}y^{(\ell)}\right)-\omega(x)\omega(y)\Big|\ ,
\end{split}
\end{equation*}
where the last equality holds because of the translation invariance of the state $\omega$. Now, assuming \eqref{stclus} one has
$$
\Big|\omega\left(X^{(N)}Y^{(N)}\right)-\omega(x)\omega(y)\Big|\le\frac{1}{N}\sum_{\ell\in\mathbb{Z}}\Big|\omega\left(x^{(0)}y^{(\ell)}\right)-\omega(x)\omega(y)\Big|
$$
that proves the scaling \eqref{macro1}. Notice that, by recursion, using the norm-boundedness of the mean-field quantities and the strong-limit in \eqref{macro2}, one can show that
\be
\label{recursion}
\lim_{N\to\infty}\omega\left(X_1^{(N)}X_2^{(N)}\cdots X_p^{(N)}\right)=\prod_{j=1}^p
\omega(x_j)\ .
\ee

\section{Appendix B}
\label{app3}

\begin{lemma}
\label{tool1}
Given the local dissipative dynamics $\gamma^{(N)}_t$ on the subalgebra $\ca_{[0,N-1]}\subset\ca$ and a state $\omega$ on the quasi-local algebra $\ca$, with the notation \eqref{meslim1a}, the following generalized Cauchy-Schwartz inequality holds:
\begin{equation*}
\left|\omega_{\vec{r}\vec{s}}\left(\gamma^{(N)}_t[xy]\right)\right|\le\sqrt{\omega\left(W^{(N)}(\vec{r})\,\gamma^{(N)}_t[x\,x^\dag](W^{(N)}(\vec{r}))^\dag\right)}
\sqrt{\omega\left((W^{(N)}(\vec{s}))^\dag\,\gamma^{(N)}_t[y^\dag y]\,W^{(N)}(\vec{s})\right)}
\end{equation*}
for all $\vec{r},\vec{s}\in\RR^{d^2}$ and $x\,,\,y\in\ca_{[0,N-1]}$.
\end{lemma}
\medskip

\begin{proof}
Using the Kraus representation of completely positive maps
\beann
\gamma^{(N)}_t[x\,y]&=&\sum_j\mathcal{V}_j(t)^\dagger\,x\,y\,\mathcal{V}_j(t)\\
&=&\Tr_2\left(\left(\sum_j\mathcal{V}_j(t)^\dagger\otimes \ket{j}\bra{j}\right)\,x\,y\,\otimes\mathbf{1}\left(\sum_\ell\mathcal{V}_\ell(t)\otimes \ket{\ell}\bra{\ell}\right)\right)\ ,
\eeann
where the $\mathcal{V}_j(t)\in\ca_{[0,N-1]}$ are operators such that $\sum_j\mathcal{V}_j^\dag(t)\mathcal{V}_j(t)=\bf{1}$, the vectors $\ket{j}$ constitute an orthonormal basis in the auxiliary Hilbert space and $\Tr_2$ denotes the trace over it.
Setting $\mathcal{V}_t:=\sum_j\mathcal{V}_j(t)\otimes\ket{j}\bra{j}$, one can write
$$
\omega_{\vec{r}\vec{s}}\left(\gamma^{(N)}_t[x\,y]\right)=
\omega\otimes\Tr_2\Big(\Big(W^{(N)}(\vec{r})\otimes\mathbf{1}\Big)\,\mathcal{V}_t^\dagger\,(x\,y\,\otimes\mathbf{1})\,\mathcal{V}_t\,\Big(W^{(N)}(\vec{s})\otimes\mathbf{1}\Big)\Big)\ .
$$
Then, the Cauchy-Schwarz inequality for the positive, not normalized, functional $\omega\otimes\Tr_2$
yields
\beann
\left|\omega_{\vec{r}\vec{s}}\left(\gamma^{(N)}_t[x\,y]\right)\right|&\le&\sqrt{\omega\otimes\Tr_2\bigg(\bigg[W^{(N)}(\vec{r})\otimes\mathbf{1}\mathcal{V}^\dagger_t\,x\bigg]\,\bigg[x^\dag\,\mathcal{V}_t\,(W^{(N)}(\vec{r}))^\dag\otimes\mathbf{1}\bigg]\bigg)}\\
&\times&\sqrt{\omega\otimes\Tr_2\bigg(\bigg[(W^{(N)}(\vec{s}))^\dag\otimes\mathbf{1}\,\mathcal{V}_t^\dagger y^\dagger\bigg]\,\bigg[y\, \mathcal{V}_t\,W^{(N)}(\vec{s})\otimes\mathbf{1}\big]\bigg)}\ .
\eeann
The result then follows by computing $\Tr_2$.
\qed
\end{proof}
\bigskip

\begin{lemma}
\label{bound}
Given the generator $\LL^{(N)}$ in \eqref{mflind1a} and the time-dependent fluctuations \eqref{fluct},
then, $\forall \mu=1,2,\ldots,d^2$ and $\forall\, t\in[0,T]$, $T\geq 0$,
\begin{equation*}
\lim_{N\to\infty}\omega\left(W^{(N)}(\vec{r})\gamma_t^{(N)}\left[\left(F_\mu^{(N)}(t)\right)^2\right]\left(W^{(N)}(\vec{r})\right)^\dagger\right)<\infty \ .
\end{equation*}
\end{lemma}
\medskip

\begin{proof}
Since $\left(F^{(N)}_\mu(t)\right)^2$ is a positive matrix, the following quantity:
\begin{equation}
\label{auxq1}
G^{(N)}(\vec{r},t):=1+\sum_{\beta=1}^{d^2}\omega\left(W^{(N)}(\vec{r})\,\gamma^{(N)}_t\left[\left(F_\mu^{(N)}(t)\right)^2\right]\,(W^{(N)}(\vec{r}))^\dag\right)
\end{equation}
satisfies 
\be
\label{aidd0}
G^{(N)}(\vec{r},t)\geq\max\left\{1\,,\,\omega\left(W^{(N)}(\vec{r})\,\gamma^{(N)}_t\left[\left(F^{(N)}_\mu(t)\right)^2\right]\,
(W^{(N)}(\vec{r}))^\dag\right)\right\}\ .
\ee
The Lemma is proved if we show that $G(\vec{r},t):=\lim_{N\to\infty} G^{(N)}(\vec{r},t)$ is finite $\forall t\ge0$. Let us then consider
\bea
\nonumber
&&\hskip-4cm
\frac{d}{dt}G^{(N)}(\vec{r},t)=
\sum_{\beta=1}^{d^2}\omega\Bigg(W^{(N)}(\vec{r})\,\gamma^{(N)}_t\Bigg[\LL^{(N)}\left[\left(F_\beta^{(N)}(t)\right)^2\right]\\
&&
-2\,\sqrt{N}\,\left(\frac{d}{dt}\omega^{(N)}_\beta(t)\right)\,F^{(N)}_\beta(t)\Bigg]\,(W^{(N)}(\vec{r}))^\dag\Bigg)\ ,
\label{derbound}
\eea
where \eqref{fluct} and \eqref{macr1a} have been used.

By splitting the generator as $\LL^{(N)}=\HH^{(N)}+\widetilde{\AA}^{(N)}\,+\,\widetilde{\BB}^{(N)}$ as in \eqref{mflind1a}, we first consider 
\beann
&&\hskip -.5cm
\widetilde{\AA}^{(N)}\left[\left(F^{(N)}_\beta(t)\right)^2\right]=\frac{1}{2}\sum_{\mu,\nu=1}^{d^2}\widetilde{A}_{\mu\nu}\,\left[\left[V^{(N)}_\mu\,,\,\left(F^{(N)}_\beta(t)\right)^2\right]\,,\,V^{(N)}_\nu\right]\\
&&\hskip -.5cm
=\frac{1}{2}\sum_{\mu,\nu=1}^{d^2}\widetilde{A}_{\mu\nu}\,\Bigg(F^{(N)}_\beta(t)\Big[\Big[V^{(N)}_\mu\,,\,F^{(N)}_\beta(t)\Big]\,,\,V^{(N)}_\nu\Big]\,+\,\Big[F^{(N)}_\beta(t)\,,\,
V^{(N)}_\nu\Big]\,\Big[V^{(N)}_\mu\,,\,F^{(N)}_\beta(t)\Big]\\
&&\hskip-.5cm
+
\Big[V^{(N)}_\mu\,,\,F^{(N)}_\beta(t)\Big]\,\Big[F^{(N)}_\beta(t)\,,\,V^{(N)}_\nu\Big]\,+\,\Big[\Big[V^{(N)}_\mu\,,\,F^{(N)}_\beta(t)\Big]\,,\,V^{(N)}_\nu\Big]\,F^{(N)}_\beta(t)\Bigg)\ .
\eeann
Since spin operators at different sites commute, the commutators read
\beann
\left[V^{(N)}_\mu\,,\,F^{(N)}_\beta(t)\right]&=&\frac{1}{N}\sum_{k=0}^{N-1}[v^{(k)}_\mu\,,
\,v_\beta^{(k)}]\\ 
\left[\left[V^{(N)}_\mu\,,\,F^{(N)}_\beta(t)\right]\,,\,V^{(N)}_\nu\right]&=&\frac{1}{N^{3/2}}\sum_{k=0}^{N-1}\Big[\Big[v^{(k)}_\mu\,,
\,v_\beta^{(k)}\Big]\,,\,v^{(k)}_\nu\Big]\ .
\eeann
Then, one readily obtains the uniform upper bound 
\be
\label{aidd2}
\left\|\widetilde{\AA}^{(N)}\left[\left(F^{(N)}_\beta(t)\right)^2\right]\right\|\,\le\frac{c}{2}\,(2vd)^4\ ,\ v=\max_\alpha\|v_\beta\|\ ,\ c=\max_{\mu\nu}\left|\widetilde{A}_{\mu\nu}|\,,\,|\widetilde{B}_{\mu\nu}|\right\}\ ,
\ee
where for later convenience we have also included $\widetilde{B}$ in the definition of the quantity $c$.
Since $\gamma^{(N)}_t$ is a contraction, it follows that the contribution of $\widetilde{\AA}^{(N)}$ to \eqref{derbound} is uniformly bounded in $N$ and $t$:
\bea
\nonumber
\sum_{\beta=1}^{d^2}\omega\Bigg(W^{(N)}(\vec{r})\gamma^{(N)}_t\Bigg[\widetilde{\AA}^{(N)}\left[\left(F_\beta^{(N)}(t)\right)^2\right]\Bigg](W^{(N)}(\vec{r}))^\dag\Bigg)&\leq&d^2\,\left\|\gamma^{(N)}_t\Bigg[\widetilde{\AA}^{(N)}\left[\left(F_\beta^{(N)}(t)\right)^2\right]\Bigg]\right\|\\
\label{aidd2b}
&\leq& \frac{c}{2}\,(2v)^4\,d^6
\eea
Let us then concentrate on the action of 
\bea
\label{aidd2c}
\widetilde{\BB}^{(N)}\left[\left(F^{(N)}_\beta(t)\right)^2\right]&=&
\frac{1}{2}\sum_{\mu,\nu=1}^{d^2}\widetilde{B}_{\mu\nu}\,\left\{F^{(N)}_\beta(t)\,\left[V^{(N)}_\mu\,,\,F^{(N)}_\beta(t)\right]\,,\,V^{(N)}_\nu\right\}\\
\label{aidd2d}
&+&\frac{1}{2}\sum_{\mu,\nu=1}^{d^2}\widetilde{B}_{\mu\nu}\,\left\{\left[V^{(N)}_\mu\,,\,F^{(N)}_\beta(t)\right]\,F^{(N)}_\beta(t)\,,\,V^{(N)}_\nu\right\}\ .
\eea
We shall denote by $B^{(N)}_1$ the contribution in~\eqref{aidd2c} and by $B^{(N)}_2$ the one in~\eqref{aidd2d}. 
We start focussing upon the first one; by adding and subtracting the mean value of $V^{(N)}_\nu$, we split $B^{(N)}_1=B^{(N)}_{11}+B^{(N)}_{12}$, where, using \eqref{macr1a},
\bea
\label{aidd3}
B^{(N)}_{11}&=&\frac{1}{2}\sum_{\mu,\nu=1}^{d^2}\widetilde{B}_{\mu\nu}\,\left\{F^{(N)}_\beta(t)\,\left[V^{(N)}_\mu\,,\,F^{(N)}_\beta(t)\right]\,,\,F^{(N)}_\nu(t)\right\}\\
\label{aidd4}
B^{(N)}_{12}&=&
\sum_{\mu,\nu=1}^{d^2}\widetilde{B}_{\mu\nu}\,\omega^{(N)}_\nu(t)\,F^{(N)}_\beta(t)\,\left[\sum_{k=0}^{N-1}v_\mu^{(k)}\,,\,F^{(N)}_\beta(t)\right]\ .
\eea
Using \eqref{aid0}, the commutator in \eqref{aidd4} can be recast as 
$$
\left[\sum_{k=0}^{N-1}v_\mu^{(k)}\,,\,F^{(N)}_\beta(t)\right]=
\frac{1}{\sqrt{N}}\sum_{k=0}^{N-1}\left[v_\mu^{(k)}\,,\,v_\beta^{(k)}\right]=
\sum_{\alpha=1}^{d^2}\,J^\alpha_{\mu\beta}\,\frac{1}{\sqrt{N}}\sum_{k=0}^{N-1}v_\alpha^{(k)}\ ;
$$
then, by adding and subtracting suitable mean-values, it can finally be expressed in terms of local fluctuations. Explicitly,  using
\eqref{macr1b} and \eqref{aid1}, it reads
\be
\label{aidd5}
\left[\sum_{k=0}^{N-1}v_\mu^{(k)}\,,\,F^{(N)}_\beta(t)\right]=
\frac{1}{\sqrt{N}}\sum_{k=0}^{N-1}\left[v_\mu^{(k)}\,,\,v_\beta^{(k)}\right]=
\sum_{\alpha=1}^{d^2}\,J^\alpha_{\mu\beta}\,F^{(N)}_\alpha(t)\,+\,
\sqrt{N}\,\omega_{\mu\beta}^{(N)}(t)\ ,
\ee
whence
\be
\label{aidd6}
B^{(N)}_{12}=\sum_{\gamma,\mu,\nu=1}^{d^2}\widetilde{B}_{\mu\nu}\,J^\gamma_{\mu\beta}\,\omega^{(N)}_\nu(t)\,F^{(N)}_\beta(t)\,F^{(N)}_\gamma(t)\,+\,\sqrt{N}
\sum_{\mu,\nu=1}^{d^2}\widetilde{B}_{\mu\nu}\,\omega^{(N)}_\nu(t)\,\omega_{\mu\beta}^{(N)}(t)\,F^{(N)}_\beta(t)\ .
\ee
Using the bounds in~\eqref{aidd2}, Lemma~\ref{tool1} and the fact that
and the fact that 
\be
\label{lastaid}
\omega\left(W^{(N)}(\vec{r})\,\gamma^{(N)}_t\left[\left(F^{(N)}_\nu(t)\right)^2\right]\,
(W^{(N)}(\vec{r}))^\dag\right)\leq G^{(N)}(\vec{r},t)\ ,
\ee 
the contribution to~\eqref{derbound} of the first term in~\eqref{aidd6} can be estimated form above by
\beann
&&
\hskip-2cm
\sum_{\alpha,\mu,\nu=1}^{d^2}|\widetilde{B}_{\mu\nu}|\,\left|J^\alpha_{\mu\beta}\right|\,|\omega^{(N)}_\nu(t)|\,
\left|\omega\Bigg(W^{(N)}(\vec{r})\,\gamma^{(N)}_t\left[F^{(N)}_\alpha(t)\,F^{(N)}_\beta(t)\right]\,(W^{(N)}(\vec{r}))^\dag\Bigg)\right|\leq\\ 
&&\hskip 2cm
\leq\ 2\,c\,v^3\,d^6\,G^{(N)}(\vec{r},t)\ ;
\eeann
Indeed, $\Big|J^\alpha_{\mu\beta}\Big|=\Big|{\rm Tr}\Big(\Big[v_\mu\,,\,v_\beta\Big]\,v_\alpha\Big)\Big|\leq2\,v^2$.

Concerning the contribution in\eqref{aidd3}, observe that
$\displaystyle
\left[V^{(N)}_\mu\,,\,F^{(N)}_\beta(t)\right]=\frac{1}{N}\sum_{k=0}^{N-1}\Big[v_\mu^{(k)}\,,\,v^{(k)}_\beta\Big]$ are the entries $T^{(N)}_{\mu\beta}$ of the operator-valued matrix $T_N=[T^{(N)}_{\mu\nu}]$ in \eqref{tcomm}.
Then, \eqref{aidd3} can be rewritten as
\beann
B^{(N)}_{11}&=&\frac{1}{2}\sum_{\mu,\nu=1}^{d^2}\widetilde{B}_{\mu\nu}\,\left(F^{(N)}_\beta(t)\,T^{(N)}_{\mu\beta}\,F^{(N)}_\nu(t)\,+\,F^{(N)}_\nu(t)\,T^{(N)}_{\mu\beta}\,F^{(N)}_\beta(t)\right)\,+\,C^{(11)}_N\\
C^{(N)}_{11}&=&\frac{1}{2}\sum_{\mu,\nu=1}^{d^2}\widetilde{B}_{\mu\nu}\,F^{(N)}_\nu(t)\,\left[F^{(N)}_\beta(t)\,,\,T^{(N)}_{\mu\beta}\right]\ ,
\quad
\left\|C^{(11)}_N\right\|\leq b\ ,\ b=4\,c\,v^4\,d^4\ ,
\eeann
the last estimate following as for the analogous uniform bound in \eqref{aidd2}.\\
Let then consider terms of the form
$$
\Delta^{(N)}(t):=\omega\left(W^{(N)}(\vec{r})\,\gamma^{(N)}_t\left[F^{(N)}_\beta(t)\,T^{(N)}_{\mu\beta}\,
F^{(N)}_\nu(t)\right]\,(W^{(N)}(\vec{r}))^\dag\right)\ .
$$
From Lemma \ref{tool1} it follows that $\left|\Delta^{(N)}(t)\right|^2$ is upper bounded by
\beann
&&
\omega\left(W^{(N)}(\vec{r})\,\gamma^{(N)}_t\left[F^{(N)}_\beta(t)\,T^{(N)}_{\beta\mu}T^{(N)}_{\mu\beta}\,F^{(N)}_\beta(t)\right]\,(W^{(N)}(\vec{r}))^\dag\right)\,\times\\
&&\hskip 4cm
\times\,\omega\left(W^{(N)}(\vec{r})\,\gamma^{(N)}_t\left[\left(F^{(N)}_\nu(t)\right)^2\right]\,
(W^{(N)}(\vec{r}))^\dag\right)\ .
\eeann
From $0\leq T^{(N)}_{\beta\mu}T^{(N)}_{\mu\beta}\le 4\,v^4$, \eqref{aidd0} and~\eqref{lastaid},
one then concludes
\be
\left|\Delta^{(N)}(t)\right|\,\leq \,2\,v^2\,G^{(N)}(\vec{r},t)\, .
\label{aidd7}
\ee
Similar considerations as before can be made for $B^{(N)}_2$ in~\eqref{aidd2d}: $B^{(N)}_2=B^{(N)}_{21}\,+\,B^{(N)}_{22}$, where
\bea
\label{lastaid2}
&&\hskip-1.5cm
B^{(N)}_{21}=\frac{1}{2}\sum_{\mu,\nu=1}^{d^2}\widetilde{B}_{\mu\nu}\,\left\{\left[V^{(N)}_\mu\,,\,F^{(N)}_\beta(t)\right]\,F^{(N)}_\beta(t)\,,\,F^{(N)}_\nu(t)\right\}\\
\label{lastaid3}
&&\hskip-1.5cm
B^{(N)}_{22}=\sum_{\gamma,\mu,\nu=1}^{d^2}\widetilde{B}_{\mu\nu}\,J^\gamma_{\mu\beta}\,\omega^{(N)}_\nu(t)\,F^{(N)}_\gamma(t)\,F^{(N)}_\beta(t)\,+\,\sqrt{N}
\sum_{\mu,\nu=1}^{d^2}\widetilde{B}_{\mu\nu}\,\omega^{(N)}_\nu(t)\,\omega_{\mu\beta}^{(N)}(t)\,F^{(N)}_\beta(t)\ .
\eea 
The final term to consider is the hamiltonian one contributed by the action of $\HH^{N)}$ that, by similar arguments as before, can be recast as
\be
\HH^{(N)}\left[\left(F^{(N)}_\beta(t)\right)^2\right]=\sum_{\gamma,\mu=1}^{d^2}\epsilon_\mu\,J^\gamma_{\mu\beta}\,\left\{F^{(N)}_\gamma(t),F^{(N)}_\beta(t)\right\}+
2\sqrt{N}\sum_{\gamma,\mu=1}^{d^2}J^\gamma_{\mu\beta}\,\omega^{(N)}_\gamma(t)\,F^{(N)}_\beta(t)\ .
\label{lastaid1}
\ee
As before, one can derive from the first term an upper bound to the derivative~\eqref{derbound} of the form $4\epsilon\,v^2\,G^{(N)}(\vec{r},t)$.

From $G^{(N)}(\vec{r},t)\ge1$, it also follows that, given any uniform upper bound $k$ to the time-derivative \eqref{derbound}, one can replace it by $k\,G^{(N)}(\vec{r},t)$, whence all upper bounds collected so far can be grouped together in an upper bound of the form $K\, G^{(N)}(\vec{r},t)$.
The only terms which escape this rule are the ones increasing with $\sqrt{N}$ in \eqref{aidd6}, \eqref{lastaid3} and~\eqref{lastaid1}. Therefore, recalling \eqref{derbound}, 
one is left with studying the large $N$-limit of
\begin{eqnarray}
\nonumber
I^{(N)}(\vec{r},t)&:=&\,2\,\sqrt{N}\sum_{\beta=1}^{d^2}\omega\Bigg(W^{(N)}(\vec{r})\,\gamma^{(N)}_t\Big[\,F^{(N)}_\beta(t)\Big]\,(W^{(N)}(\vec{r}))^\dag\Bigg)\times\\
&\times&\Bigg(\sum_{\mu,\nu=1}^{d^2}\Big(\widetilde{B}_{\mu\nu}\,\omega_{\mu\beta}^{(N)}(t)\,+\,\,J^\nu_{\mu\beta}\Big)\,\omega^{(N)}_\nu(t)
-\,\frac{d}{dt}\omega^{(N)}_\beta(t)\Bigg)\ .
\label{auxeq5}
\end{eqnarray}
First we consider the case $\vec{r}=0$; then, since $\omega\left(\gamma_t^{(N)}\left[F^{(N)}_\beta(t)\right]\right)=0$, we get
$I^{(N)}(0,t)=0$.
Therefore, for all $N$,
$$
\frac{d}{dt}G^{(N)}(0,t)<K_0\, G^{(N)}(0,t)\qquad\hbox{implies}\qquad 
G^{(N)}(0,t)<\,{\rm e}^{t\, K_0} G^{(N)}(0)<\infty\ .
$$
When $r\neq0$, we estimate
\beann
&&\hskip-2cm
\Big|I^{(N)}(\vec{r},t)\Big|\le\,2\,\sum_{\beta=1}^{d^2}\left|\omega\left(W^{(N)}(\vec{r})\gamma_t^{(N)}\left[F^{(N)}_\beta(t)\right]\left(W^{(N)}\right)^\dagger(\vec{r})\right)\right|\,\times\\
&&\hskip2cm
\times\,\sqrt{N}\left|\sum_{\mu,\nu=1}^{d^2}\Big(\widetilde{B}_{\mu\nu}\,\omega_{\mu\beta}^{(N)}(t)\,+\,J^\nu_{\mu\beta}\Big)\,\omega^{(N)}_\nu(t)\,-\,\frac{d}{dt}\omega_\beta^{(N)}(t)\right|\ .
\eeann
Considering the time-derivative, one has:
\begin{eqnarray*}
\frac{d}{dt}\omega_\beta^{(N)}(t)&=&\omega\left(\gamma_t^{(N)}\left[\HH^{(N)}\left[\frac{1}{N}\sum_{k=0}^{N-1}v_\beta^{(k)}\right]\right]\right)\,+\,\omega\left(\gamma_t^{(N)}\left[\widetilde{\AA}^{(N)}\left[\frac{1}{N}\sum_{k=0}^{N-1}v_\beta^{(k)}\right]\right]\right)\\
&+&\frac{1}{2}\sum_{\mu,\nu=1}^{d^2}\widetilde{B}_{\mu\nu}\,\omega\left(\gamma_t^{(N)}\left[\left\{\frac{1}{N}\sum_{k=0}^{N-1}\left[v_\mu^{(k)},v_{\beta}^{(k)}\right],\frac{1}{N}\sum_{h=0}^{N-1}v_\nu^{(h)}\right\}\right]\right)\ .
\end{eqnarray*}
The $\widetilde{\AA}^{(N)}$ contribution scales as $1/N$; together with the factor $\sqrt{N}$ and the divergence as $\sqrt{N}$ of the upper bound to the norm of the fluctuation $F^{(N)}_\beta(t)$, it contributes with an upper bound to $\Big|I^{(N)}(\vec{r},t)\Big|$ of the form $4d^6v^4c\leq\,4d^6v^4 c\,G^{(N)}(\vec{r},t)$, using~\eqref{auxq1}. 

Taking into account that 
$$
\omega\left(\gamma_t^{(N)}\left[\HH^{(N)}\left[\frac{1}{N}\sum_{k=0}^{N-1}v_\beta^{(k)}\right]\right]\right)=\sum_{\mu\,\nu=1}^{d^2}\epsilon_\mu\,J^\nu_{\mu\beta}\,\omega^{(N)}_\nu(t)\ ,
$$
and writing
$$
\sqrt{N}\,\frac{1}{N}\sum_{h=0}^{N-1}v_\nu^{(h)}=F_\nu^{(N)}(t)\,+\,\sqrt{N}\,\omega^{(N)}_\nu(t)\ ,
$$ 
since the latter term is a scalar multiple of the identity, one gets
\begin{eqnarray*}
&&\hskip-1cm
\Big|I^{(N)}(\vec{r},t)\Big|\,<\,4d^6\,v^4c \,G^{(N)}(\vec{r},t)\,+\\
&&\hskip-.5cm
+\sum_{\beta=1}^{d^2}\left|\omega\left(W^{(N)}(\vec{r})\gamma_t^{(N)}\left[F^{(N)}_\beta(t)\right]\left(W^{(N)}(\vec{r})\right)^\dagger\right)\right|\times\\
&&\times\left|\sum_{\mu,\nu=1}^{d^2}\widetilde{B}_{\mu\nu}\,\omega\left(\gamma_t^{(N)}\left[\left\{\frac{1}{N}\sum_{k=0}^{N-1}\left[v_\mu^{(k)},v_\beta^{(k)}\right],F_\nu^{(N)}(t)\right\}\right]\right)\right|\ .
\end{eqnarray*}
Finally, by means of Lemma \ref{tool1} and the Cauchy-Schwarz inequality, one gets 
$$
\Big|I^{(N)}(\vec{r},t)\Big|\le K_1 G^{(N)}(\vec{r},t)\,+\,K_2 G^{(N)}(0,t)\, G^{(N)}(\vec{r},t)\ ,
$$
with $K_1=4d^6v^4c$ and $K_2=2v^2d^6c$, implying
$$
\frac{d}{dt}G^{(N)}(\vec{r},t)<\Big(K_1+K_2 G^{(N)}(0,t)\Big)\,G^{(N)}(\vec{r},t)\ ,
$$
from which the boundedness of $G_\mu(\vec{r},t):=\lim_{N\to\infty} G^{(N)}(\vec{r},t)$ follows.
\qed
\end{proof}
\bigskip

\begin{lemma}
\label{tool3}
Let $M_t$ be a time-dependent hermitean matrix and $\displaystyle N_t={\rm e}^{iM_t}$. 
Then,
\begin{equation}
\label{appb1}
\dot{N_t}:=\frac{{\rm d}N_t}{{\rm d}t}=O_t\,N_t\ ,\quad O_t:=\sum_{k=1}^\infty\frac{i^k}{k!}\mathbb{K}_{M_t}^{k-1}[\dot{M}_t]\ ,
\end{equation}
where 
$\mathbb{K}^n_{M_t}[\dot{M}_t]=\Big[M_t\,,\,\mathbb{K}_{M_t}^{n-1}[\dot{M}_t]\Big]$ and $\mathbb{K}^0_{M_t}[\dot{M}_t]=\dot{M}_t$.
\end{lemma}
\medskip

Using \eqref{aux1}, from $[N_t\,,\,M_t]=0$ and $N_tM_tN^\dag_t=M_t$ one derives
$$
N_t\,\dot{M}_t\,N^\dag_t=\dot{M}_t\,-\,\dot{N}_t\,M_t\,N^\dag_t\,-\,N_t\,M_t\,\dot{N}^\dag_t\ .
$$
Since $\displaystyle \dot{N}_t\,N^\dag_t=-N_t\,\dot{N}^\dag_t$ and, for $n\geq 1$,
$\displaystyle \mathbb{K}_A^n[B]=\Big[A\,,\,\mathbb{K}^{n-1}_A[B]\Big]$, it follows that
$$
\hskip -.5cm
N_t\,\dot{M}_t\,N^\dag_t-\dot{M}_t=\sum_{n=1}^\infty\frac{i^n}{n!}\mathbb{K}^n_{M_t}[\dot{M}_t]=\Big[M_t\,,\,O_t\Big]
=\Big[M_t\,,\,\dot{N}_t\Big]\,N_t^\dag\ .
$$
Let $\vert m_a(t)\rangle$ denote the orthogonal eigenvectors of $M_t$; then 
$$
0=\frac{{\rm d}}{{\rm d}t}\Big(\langle m_a(t)\vert m_b(t)\rangle\Big)=\langle \dot{m}_a(t)\vert m_b(t)\rangle\,+\,
\langle m_a(t)\vert\dot{m}_b(t)\rangle\ .
$$
If $\vert m_a(t)\rangle$ and $\vert m_b(t)\rangle$ correspond to different eigenvalues, 
\begin{eqnarray*}
\langle m_a(t)\vert \dot{M}_t\vert m_b(t)\rangle&=&\Big(m_a(t)-m_b(t)\Big)\,\langle\dot{m}_b(t)\vert m_a(t)\rangle\\
\langle m_a(t)\vert O_t\,N_t\vert m_b(t)\rangle&=&\sum_{k=1}^\infty \frac{i^k}{k!}(m_a(t)-m_b(t))^{k-1}\,\langle m_a(t)\vert \dot{M}_t\vert m_b(t)\rangle\,{\rm e}^{im_b(t)}\\
&=&{\rm e}^{im_a(t)}-{\rm e}^{im_b(t)}=\langle m_a(t)\vert \dot{N}_t\vert m_b(t)\rangle\ .
\end{eqnarray*}
On the other hand if $\vert m_a(t)\rangle$ and $\vert m_b(t)$ correspond to a same (real) eigenvalue $m(t)$, whence
\begin{eqnarray*}
\langle m_a(t)\vert O_t\,N_t\vert m_b(t)\rangle&=&i\,\langle m_a(t)\vert \dot{M}_t\vert m_b(t)\rangle\, {\rm e}^{im(t)}\, 
\delta_{ab}
=i\dot{m}(t)\,{\rm e}^{im(t)}\,\delta_{ab}\\ 
&=&\langle m_a(t)\vert\dot{N}_t\vert m_b(t)\rangle\ .
\end{eqnarray*}
\hfill\qed

\end{document}